%% file: linear_solvers_stco.tex
\RequirePackage{fix-cm}
\documentclass[twocolumn,envcountsame]{svjour3}          
\smartqed  
\PassOptionsToPackage{hyphens}{url}
\usepackage{hyperref}

\usepackage{microtype}

\usepackage{graphicx}
\usepackage[utf8]{inputenc}
\usepackage{mathtools, amssymb}
\usepackage{xcolor}
\usepackage[title]{appendix}

\usepackage{bm}
\usepackage{natbib}
\usepackage{nicefrac}
\usepackage{stmaryrd} 

\usepackage{algorithm}
\usepackage{algpseudocode}

\newcommand\numberthis{\addtocounter{equation}{1}\tag{\theequation}}
\algrenewcommand{\algorithmiccomment}[1]{\hfill {$\sslash$ \scriptsize #1}}
\algrenewcommand\alglinenumber[1]{\tiny #1}

\newcommand{\algalign}[3]{\raisebox{0pt}[-1pt][-1pt]{$\phantom{#1}\mathllap{#2{}}$ $\gets {}#3$}}

\makeatletter
\def\therule{\makebox[\algorithmicindent][l]{\hspace*{.5em}\vrule height .75\baselineskip depth .25\baselineskip}}%

\newtoks\therules
\therules={}
\def\appendto#1#2{\expandafter#1\expandafter{\the#1#2}}
\def\gobblefirst#1{
  #1\expandafter\expandafter\expandafter{\expandafter\@gobble\the#1}}%
\def\LState{\State\unskip\the\therules}
\def\pushindent{\appendto\therules\therule}%
\def\popindent{\gobblefirst\therules}%
\def\printindent{\unskip\the\therules}%
\def\printandpush{\printindent\pushindent}%
\def\popandprint{\popindent\printindent}%

\algdef{SE}[WHILE]{While}{EndWhile}[1]
  {\printandpush\algorithmicwhile\ #1\ \algorithmicdo}
  {\popandprint\algorithmicend\ \algorithmicwhile}%
\algdef{SE}[FOR]{For}{EndFor}[1]
  {\printandpush\algorithmicfor\ #1\ \algorithmicdo}
  {\popandprint\algorithmicend\ \algorithmicfor}%
\algdef{S}[FOR]{ForAll}[1]
  {\printindent\algorithmicforall\ #1\ \algorithmicdo}%
\algdef{SE}[LOOP]{Loop}{EndLoop}
  {\printandpush\algorithmicloop}
  {\popandprint\algorithmicend\ \algorithmicloop}%
\algdef{SE}[REPEAT]{Repeat}{Until}
  {\printandpush\algorithmicrepeat}[1]
  {\popandprint\algorithmicuntil\ #1}%
\algdef{SE}[IF]{If}{EndIf}[1]
  {\printandpush\algorithmicif\ #1\ \algorithmicthen}
  {\popandprint\algorithmicend\ \algorithmicif}%
\algdef{C}[IF]{IF}{ElsIf}[1]
  {\popandprint\pushindent\algorithmicelse\ \algorithmicif\ #1\ \algorithmicthen}%
\algdef{Ce}[ELSE]{IF}{Else}{EndIf}
  {\popandprint\pushindent\algorithmicelse}%
\algdef{SE}[PROCEDURE]{Procedure}{EndProcedure}[2]
   {\printandpush\algorithmicprocedure\ \textproc{#1}\ifthenelse{\equal{#2}{}}{}{(#2)}}%
   {\popandprint\algorithmicend\ \algorithmicprocedure}%
\algdef{SE}[FUNCTION]{Function}{EndFunction}[2]
   {\printandpush\algorithmicfunction\ \textproc{#1}\ifthenelse{\equal{#2}{}}{}{(#2)}}%
   {\popandprint\algorithmicend\ \algorithmicfunction}%
\makeatother

\makeatletter
\def\cl@chapter{}
\makeatother

\usepackage[capitalise]{cleveref}

\definecolor{lred}{RGB}{200,0,0}
\definecolor{dred}{RGB}{130,0,0} 
\definecolor{dblu}{RGB}{0,0,130}
\definecolor{dgre}{RGB}{0,130,0} 
\definecolor{dgra}{RGB}{50,50,50}
\definecolor{mgra}{RGB}{100,100,100}
\definecolor{lgra}{RGB}{220,220,220}
\definecolor{MPG}{RGB}{000,125,122}
\definecolor{orange}{RGB}{255, 128, 0}

\hypersetup{
  colorlinks=true,
  linkcolor=dred,          
  citecolor=dblu,        
  filecolor=dgra,      
  urlcolor=orange           
}

\usepackage{tikz,pgfplots}
\pgfplotsset{compat=newest}
\usetikzlibrary{plotmarks}
\usepgfplotslibrary{groupplots}
\pgfplotsset{plot coordinates/math parser=false}       
\pgfkeys{/pgfplots/mystyle/.style={
  semithick,
  tick style={major tick length=4pt,semithick,gray},
  xtick align = inside,
  ytick align = inside,
  xlabel near ticks,
  ylabel near ticks,
  }}
\usepgfplotslibrary{external} 
\tikzsetexternalprefix{tikz_out/}
\newlength\figheight%
\newlength\figwidth%

\makeatletter
\newcommand{\superimpose}[2]{
  {\ooalign{$#1\@firstoftwo#2$\cr\hfil$#1\@secondoftwo#2$\hfil\cr}}}
\makeatother

\newcommand{\Exp}{\mathbb{E}} 
\newcommand{\sk}{\mathpalette\superimpose{{\otimes}{\ominus}}} 

\usepackage{colonequals}
\newcommand{\ce}{\colonequals} 

\DeclareMathOperator*{\argmin}{arg\,min} 

\newcommand{\reals}{\mathbb{R}}
\renewcommand{\Re}{\reals}
\newcommand{\N}{\mathcal{N}} 

\newcommand{\set}[1]{\{#1\}}
\newcommand{\norm}[1]{\|#1\|}
\newcommand{\inner}[1]{\langle #1 \rangle}

\renewcommand{\vec}{\boldsymbol} 
\newcommand{\mat}[1]{#1} 
\newcommand{\vecm}[1]{\overrightarrow{#1}} 
\newcommand{\Trans}{^{\top}} 
\newcommand{\kr}{\otimes} 

\newcommand{\A}{\mat A} 
\newcommand{\Ai}{\A^{-\!1}} 
\renewcommand{\b}{\vec b} 

\usepackage{comment}

\newcommand*{\arXiv}[1]{\bgroup\color{blue}\href{http://arxiv.org/abs/#1}{arXiv:#1}\egroup}

\newcommand{\spa}{\operatorname{span}}
\newcommand{\qq}{\qquad}

\begin{document}

\title{Probabilistic Linear Solvers: A Unifying View}
\author{Simon Bartels$^*$,
 Jon Cockayne$^*$,
 Ilse C.~F.~Ipsen and Philipp Hennig}


\institute{
Simon Bartels \at
University of Tübingen\\
and Max Planck Institute for Intelligent Systems\\
72076 Tübingen, Germany \\
\email{sbartels@tue.mpg.de}
\and
Jon Cockayne \at
Department of Statistics \\
University of Warwick \\
Coventry, CV4 7AL, UK \\
\email{j.cockayne@warwick.ac.uk}
\and
Ilse C.~F.~Ipsen \at
Department of Mathematics \\
North Carolina State University \\
Raleigh, NC 27695-8205, USA \\
\email{ipsen@ncsu.edu}
\and
Philipp Hennig \at
University of Tübingen\\
and Max Planck Institute for Intelligent Systems\\
72076 Tübingen, Germany \\
\email{ph@tue.mpg.de}
\\
\\
\\
\\
$~^*$ Authors contributed equally.
}

\date{Received: date / Accepted: date}

\maketitle 
 
\begin{abstract}
Several recent works have developed a new, probabilistic interpretation for numerical algorithms solving linear systems in which the solution is inferred in a Bayesian framework, either directly or by inferring the unknown action of the matrix inverse.
These approaches have typically focused on replicating the behavior of the conjugate gradient method as a prototypical iterative method.
In this work surprisingly general conditions for equivalence of these disparate methods are presented. 
We also describe connections between probabilistic linear solvers and projection methods for linear systems, providing a probabilistic interpretation of a far more general class of iterative methods.
In particular, this provides such an interpretation of the generalised minimum residual method.
A probabilistic view of preconditioning is also introduced.
These developments unify the literature on probabilistic linear solvers, and provide foundational connections to the literature on iterative solvers for linear systems.

\keywords{Probabilistic linear solvers \and Projection methods \and Iterative methods \and Preconditioning}
\end{abstract}

\section{Introduction}
Consider the linear system
\begin{equation}
	A\bm{x}^* = \bm{b} \label{eq:system}
\end{equation}
where $A \in \reals^{d\times d}$ is an invertible matrix, $\bm{b} \in \reals^d$ is a given vector and $\bm{x}^* \in \reals^d$ is an unknown to be determined.
Recent work \citep{Hennig2015, Cockayne:2018} has constructed iterative solvers for this problem which output \emph{probability measures}, constructed to quantify uncertainty due to terminating the algorithm before the solution has been identified completely.
On the surface the approaches in these two works appear different:
In the matrix-based inference (MBI) approach of \cite{Hennig2015}, a posterior is constructed on the matrix $\Ai$, while in the solution-based inference (SBI) method of \cite{Cockayne:2018} a posterior is constructed on the solution vector $\bm{x}^*$.

These algorithms are instances of \emph{probabilistic numerical methods} (PNM) in the sense of \cite{HenOsbGirRSPA2015} and \cite{Cockayne:2017}.
PNM are numerical methods which output posterior distributions that quantify uncertainty due to discretisation error.
An interesting property of PNM is that they often result in a posterior distributions whose mean element coincides with the solution given by a classical numerical method for the problem at hand.
The relationship between PNM and classical solvers has been explored for integration \citep[e.g.][]{Karvonen2017}, ODE-solvers \citep{SchoberDH2014, SchoberSarkkaHennig2018, 2018arXiv180709737K} and PDE solvers \citep{Cockayne:2016} in some generality. For linear solvers, attention has thus far been restricted to the conjugate gradient (CG) method.
Since CG is but a single member of a larger class of iterative solvers, and applicable only if the matrix $A$ is symmetric and positive-definite, extending the probabilistic interpretation is still an interesting endeavour.
Probabilistic interpretations provide an alternative perspective on numerical algorithms, and can also provide extensions such as the ability to exploit noisy or corrupted observations.
The probabilistic view has also been used to the develop new numerical methods \citep{Xi:2018}, and \emph{Bayesian} PNM can be incorporated rigorously into pipelines of computation \citep{Cockayne:2017}.

\emph{Preconditioning}---mapping Eq.~\eqref{eq:system} to a better conditioned system whith the same solution---is key to the fast convergence of iterative linear solvers, particularly those based upon Krylov methods \citep{Liesen:2012tt}. The design of preconditioners has been referred to as ``a combination of art and science'' \citep[p. 283]{Saad2003iterative}.
In this work we also provide a new, probabilistic interpretation of preconditioning as a form of prior information.

\subsection{Contribution}
This text contributes three primary insights:
\begin{enumerate}
\item It is shown that, for particular choices of the generative model, matrix-based inference (MBI) and solution-based inference (SBI) can be equivalent (Section~\ref{sec:probabilistic_linear_solvers}).
\item A general probabilistic interpretation of projection methods \citep{Saad2003iterative} is described (Section~\ref{sec:projection_methods}), leading to a probabilistic interpretation of the generalised minimum residual method (GMRES; \citet{Saad1986}, Section \ref{sec:gmres}). The connection to CG is expanded and made more concise in Section \ref{sec:cg}.
\item A probabilistic interpretation of preconditioning is presented in Section~\ref{sec:preconditioning}.
\end{enumerate}
Most of the proofs are presented inline; lengthier proofs are deferred to Appendix~\ref{sec:proofs}.
While an important consideration, the predominantly theoretical contributions of this paper will not consider the impact of finite numerical precision.

\subsection{Notation}
For a symmetric positive-definite matrix $\mat M \in \reals^{d\times d}$ and two vectors $\vec v, \vec w \in \reals^d$, we write $\inner{\vec v, \vec w}_M = \vec v\Trans \mat M \vec w$ for the inner product induced by $\mat M$, and $\norm{\vec v}_M^2 = \inner{\vec v, \vec v}_M$ for the corresponding norm.

A set of vectors $\vec s_1, \dots, \vec s_m$ is called \emph{$\mat M$-orthogonal} or \emph{$M$-conjugate} if $\inner{\vec s_i, \vec s_j}_M = 0$ for $i\neq j$, and \emph{$\mat M$-orthonormal} if, in addition, $\|\vec s_i\|_M = 1$ for $1\leq i\leq m$.

For a square matrix $A =\begin{bmatrix} \bm{a}_1 & \ldots & \bm{a}_d\end{bmatrix}\Trans\in \reals^{d\times d}$,
the \emph{vectorisation operator} $\textup{vec} : \reals^{d\times d} \to \reals^{d^2}$ stacks the rows\footnote{Stacking the columns is equivalently possible and common. It is associated with a permutation in the definition of the Kronecker product, but the resulting inferences are equivalent.} 
of $A$ into one long vector:
\begin{equation*}
\vecm{\A} \equiv \textup{vec}(\A) = \begin{bmatrix} \vec a_1 \\ \vdots \\ \vec a_d \end{bmatrix},\quad\text{with}\quad \left[\vecm{\A}\right]_{(ij)} = [A]_{ij} .
\end{equation*}
The \emph{Kronecker product} of two matrices $A, B \in \reals^{d\times d}$ is $A \kr B$ with $[A\kr B]_{(ij),(k\ell)} = [A]_{ik}[B]_{j\ell}$.
A list of its properties is provided in Appendix~\ref{sec:kron_properties}.

The Krylov space of order $m$ generated by the matrix $A\in\reals^{d\times d}$ and the vector 
$\bm{b}\in\reals^d$ is 
\begin{equation*}
	K_m(\A, \b) = \textup{span}(\b, \A\b, \A^2\b, \dots, \A^{m-1}\b).
\end{equation*}
We will slightly abuse notation to describe shifted and scaled subspaces of $\reals^d$: Let $\mathbb{S}$ be an $m$-dimensional linear subspace of $\reals^d$ with basis $\set{\vec s_1, \dots, \vec s_m}$.
Then for a vector $\vec v \in \reals^d$ and a matrix $M \in \reals^{d\times d}$, let
\begin{equation*}
	\vec v + M \mathbb{S} = \textup{span}(\vec v + M\vec s_1, \dots, \vec v + M\vec s_m).
\end{equation*}

\section{Probabilistic Linear Solvers}
\label{sec:probabilistic_linear_solvers}

Several probabilistic framework describing the solution of \cref{eq:system} have been constructed in recent years.
They primarily differ in the subject of inference:
SBI approaches such as \citet{Cockayne:2018}, of which \emph{BayesCG} is an example, place a prior distribution on the solution $\bm{x}^*$ of \cref{eq:system}.
Conversely, the MBI approach of \citet{Hennig2015} and \citet{Bartels:2016eh} places a prior on $\Ai$, treating the action of the inverse operator as an unknown to be inferred\footnote{\cite{Hennig2015} also discusses inference over $\A$. This model class will not be discussed further in the present work. It has the disadvantage that the associated marginal on $\bm{x}^*$ is non-analytic, but more easily lends itself to situations with noisy or otherwise perturbed matrix-vector products as observations.}.
This section reviews each approach and adds some new insights. 
In particular, SBI can be viewed as strict special case of MBI (Section~\ref{sec:equivalence}).
Throughout this section, we will assume that the search directions $S_m$ in $S_m^\top A \bm{x} = S_m^\top \bm{b}$ are given \emph{a-priori}; Section~\ref{sec:cg} examines algorithms which iteratively generate search directions adapted to the problem at hand.

\subsection{Background on Gaussian conditioning}
The propositions in this section follow from the following two classic properties of Gaussian distributions.

\begin{lemma} \label{lem:gaussian_projection}
  Let $\vec x \in \reals^d$ be Gaussian distributed with density $p(\vec x)=\N(\vec x; \vec x_0, \mat \Sigma)$ for $\vec x_0 \in \reals^d$ and $\Sigma \in \reals^{d\times d}$ a positive semi-definite matrix.
  Let $\mat M \in \reals^{n \times d}$ and $\vec z \in \reals^n$. Then $\vec v = \mat M \vec x + \vec z$ is also Gaussian, with
  \begin{equation*}
    p(\vec v) = \N(\vec v; M\vec x_0 + \vec z, \mat M\mat \Sigma \mat M\Trans).
  \end{equation*}
\end{lemma}

\begin{lemma} \label{lem:gaussian_conditioning}
Let $\vec x \in \reals^d$ be distributed as in Lemma \ref{lem:gaussian_projection}, and let observations $\vec y \in \reals^n$ be generated from the conditional density
\begin{equation*}
  p(\vec y\mid \vec x)=\N(\vec y; \mat M \vec x + \vec z, \mat \Lambda)
\end{equation*}
with $\mat M \in \reals^{n \times d}$, $\vec z \in \reals^n$, and $\Lambda \in \reals^{n \times n}$ again positive-semidefinite. 
Then the associated conditional distribution on $\vec x$ after observing $\vec y$ is again Gaussian, with
\begin{align*}
p(\vec x\mid \vec y)&=\N(\vec x; \bar{\vec x}, \bar{\mat \Sigma}) \qquad\text{where}\\
\bar{\vec x} &= 
\vec x_0+\mat \Sigma \mat M\Trans(\mat M\mat{\Sigma}\mat M\Trans + \mat \Lambda)^{-\!1}(\vec y - \mat M\vec x_0 - \vec z) \\
\bar{\mat \Sigma} &= \mat \Sigma - \mat \Sigma \mat M\Trans(\mat M\mat \Sigma \mat M\Trans + \mat \Lambda)^{-\!1}\mat M\mat \Sigma).
\end{align*}
This formula also applies if $\Lambda = 0$, i.e. observations are made without noise, with the caveat that if $M\Sigma M\Trans$ is singular, the inverse should be interpreted as a pseudo-inverse.
\end{lemma}

\subsection{Solution-Based Inference}
\label{sec:bayescg}
To phrase the solution of Eq.~\eqref{eq:system} as a form of probabilistic inference, \citet{Cockayne:2018} consider a Gaussian prior over the solution $\vec x^*$, and condition on observations provided by a set of \emph{search directions} $\vec s_1, \dots, \vec s_m$, $m < d$.
Let $S_m \in \reals^{d \times m}$ be given by $S_m = [\vec s_1, \hdots, \vec s_m]$, and let information be given by $\vec y_m\ce \mat S_m\Trans \A\vec x^*=\mat S_m\Trans\b$. 
Since the information is clearly a linear projection of $\vec{x}^*$, the posterior distribution is a Gaussian distribution on $\vec{x}^*$:

\begin{lemma}[\citet{Cockayne:2018}]\label{lemma:solution_posterior}
Assume that the columns of $S_m$ are linearly independent. Consider the prior
\begin{equation*}
p(\vec x)=\N(\vec x; \vec x_0, \mat \Sigma_0). 
\end{equation*}
The posterior from SBI is then given by
\begin{equation*}
  p(\vec x\mid \vec y_m)=\N(\vec x; \vec x_m, \mat \Sigma_m)
\end{equation*}
where
\begin{align}
\vec x_m &= \vec x_0 + \mat \Sigma_0 \A\Trans \mat S_m
(\mat S_m\Trans \A \mat \Sigma_0 \A\Trans \mat S_m)^{-\!1}\mat S_m\Trans \vec r_0\label{eq:sbi_mean}
\\\mat{\Sigma}_m &= \mat \Sigma_0 - \mat \Sigma_0 \A\Trans \mat S_m(\mat S_m\Trans \A \mat \Sigma_0\A\Trans \mat S_m)^{-\!1}\mat S_m\Trans \mat \Sigma_0,\nonumber
\end{align}
and $\vec r_0=\vec b-A\vec x_0$.
\end{lemma}
The following proposition establishes an optimality property of the posterior mean $\vec{x}_m$.
This is a relatively well-known property of Gaussian inference, but has not appeared before in the literature on these methods and will prove useful in subsequent sections.

\begin{proposition} \label{prop:generic_solution_optimality}
If $\mathbb{S}_m = \textup{range}(\mat S_m)$, then the posterior mean in
Lemma~\ref{lemma:solution_posterior} satisfies the optimality property
  \begin{equation*}
\vec x_m = \argmin_{\vec x \in \mat \vec x_0 + \Sigma_0 \A\Trans \mathbb{S}_m}
{ \norm{\bm{x} - \bm{x}^*}_{\Sigma_0^{-\!1}}}.
  \end{equation*}
\end{proposition}
\begin{proof}
With the abbreviations $X=\Sigma_0A^\top S_m$ and $\bm{y}=\bm{x}^*-\bm{x}_0$
the mean in Lemma~\ref{lemma:solution_posterior} can be written as
\begin{equation*}
\bm{x}_m=\bm{x}_0+X\bm{c}_m, 
\end{equation*}
where
\begin{equation*}
\bm{c}_m=(X^\top \Sigma_0^{-\!1}X)^{-\!1}X^\top \Sigma_0^{-\!1}\bm{y}
\end{equation*}
is the solution of the weighted least squares problem \cite[Section  6.1]{GolubVanLoan2013}
\begin{eqnarray*}
\bm{c}_m &=& \argmin_{\bm{c} \in \reals^m}{\|X\bm{c}-\bm{y}\|_{\Sigma_0^{-\!1}}}\\
&=& \argmin_{\bm{c} \in \reals^m}{\|\bm{x}_0+\Sigma_0 A^\top  S_m\bm{c}-\bm{x}^*\|_{\Sigma_0^{-\!1}}}.
\end{eqnarray*}
This is equivalent to the desired statement.
\qed\end{proof}

\subsection{Matrix-Based Inference} \label{sec:mbi}
In contrast to SBI, the MBI approach of \citet{Hennig2015} treats the matrix inverse $\Ai$ as the unknown in the inference procedure. 
As in the previous section, search directions $S_m$ yield matrix-vector products $Y_m \in \reals^{d\times m}$. 
In \citet{Hennig2015} these arise from \emph{right}-multiplying\footnote{This work also considers a model class that explicitly encodes \emph{symmetry} of $A$, such that the distinction between left- and right- multiplication vanishes. See Section~\ref{sec:cg_in_mbi} and Prop.~\ref{thm:cg_right_multiplied} for more.} $\mat A$ with $\mat S_m$, i.e. $\mat Y_m = \A \mat S_m$. Note that
\begin{equation} \label{eq:mbi_right_information}
  \mat S_m = \Ai \mat Y_m, \text{ or, equivalently } \vecm{\mat S_m} = (I \kr \mat Y_m\Trans) \vecm{\Ai}.
\end{equation} 
Thus $\mat S_m$ is a linear transformation of $\Ai$ and Lemma \ref{lem:gaussian_conditioning} can again be applied:
\begin{lemma}[Lemma 2.1 in \citet{Hennig2015}\footnote{This corrects a printing error in \citet{Hennig2015}. The notation has been adapted to fit the context.}]
\label{lemma:mbi}
Consider the prior
\begin{equation*}
	p\left(\vecm{\Ai}\right)=\N\left(\vecm{\Ai_0}, \mat \Sigma_0 \kr \mat W_0\right).
\end{equation*}
Then the posterior given the observations $\vecm{\mat S_m}= \Ai\mat Y_m$ is given by
\begin{equation*}
	p\left(\vecm{\Ai}\,\middle|\, \vecm{\mat S_m}\right)=\N\left(\vecm{\Ai_m}, \mat \Sigma_0 \kr \mat W_m\right)
\end{equation*}
with
\begin{align*}
\Ai_m &=  \Ai_0 + (\mat S_m - \Ai_0 \mat Y_m)(\mat Y_m\Trans \mat W_0\mat Y_m)^{-\!1}\mat Y_m\Trans\mat W_0 \\
\mat W_m &=  \mat W_0 - \mat W_0\mat Y_m(\mat Y_m\Trans \mat W_0\mat Y_m)^{-\!1}\mat Y_m\Trans \mat W_0.
\end{align*}
\end{lemma}
For linear solvers, the object of interest is $\vec{x}^*=\Ai \b$. 
Writing $\Ai\b=(\mat I\kr \b\Trans)\vecm{\Ai}$, and again using Lemma~\ref{lem:gaussian_projection}, we see that the associated marginal is also Gaussian, and given by
\begin{equation}
	p(\vec x\mid \mat S, \mat Y)=\N(\vec x; \Ai_m \b, \b\Trans \mat W_m \b \cdot \mat \Sigma_0).
\end{equation}
In the Kronecker product specification for the prior covariance on $\Ai$, the first matrix, here $\Sigma_0$, describes the dependence between the columns of $\Ai$.
The second matrix, $W_0$, captures the dependency between the rows of $\Ai$.
Note that in \cref{lemma:mbi}, the posterior covariance has the form $\Sigma_0 \otimes W_m$.
When compared to the prior covariance, $\Sigma_0 \otimes W_0$, it is clear that the observations have conveyed no new information to the first term of the Kronecker product covariance.

\subsection{Equivalence of MBI and SBI}
\label{sec:equivalence}
In practise \citet{Hennig2015} notes that inference on $\Ai$ should be performed only implicitly, avoiding the $d^2$ storage cost and the mathematical complexity of the operations involved in Lemma~\ref{lemma:mbi}. 
This raises the question of when MBI is equivalent to SBI. 
Although, based on Lemma~\ref{lem:gaussian_projection}, one might suspect SBI and MBI to be equivalent, in fact the posterior from Lemma \ref{lemma:mbi} is structurally different to the posterior in Lemma \ref{lemma:solution_posterior}:
After projecting into solution space, the posterior covariance in Lemma \ref{lemma:mbi} is a scalar multiple of the matrix $\Sigma_0$, which is not the case in general in Lemma \ref{lemma:solution_posterior}.

However, the implied posterior over the solution vector can be made to coincide with the posterior from SBI if one considers observations in MBI as
\begin{equation} \label{eq:mbi_left_info}
	S_m^\top = Y_m^\top A^{-1}.
\end{equation}
That is, as \emph{left}-multiplications of $A$. We will refer to the observation model of Eq.~\eqref{eq:mbi_right_information} as \emph{right-multiplied information}, and to Eq.~\eqref{eq:mbi_left_info} as \emph{left-multiplied information}.

\begin{proposition} \label{prop:matrix_prior_equivalent_solution_prior}
  Consider a Gaussian MBI prior 
  \begin{equation*}p(\A^{-1})= \mathcal{N}(\A^{-1};\vecm{\mat{A_0^{-1}}}, \mat\Sigma_0 \otimes \mat W_0),
  \end{equation*}
  conditioned on the left-multiplied information of Eq.~\eqref{eq:mbi_left_info}. 
  The associated marginal on $\bm{x}$ is identical to the posterior on $\bm{x}$ arising in Lemma~\ref{lemma:solution_posterior} from $p(\bm{x})= \mathcal{N}(\bm{x};\bm{x}_0, \Sigma_0)$ under the conditions
		\[A_0^{-1} \bm{b} = \bm{x}_0 \quad\text{and}\quad\bm{b}^\top W_0 \bm{b} = 1.\]
\end{proposition}
\begin{proof}
See Appendix~\ref{sec:proofs}.
\qed\end{proof}
The first of the two conditions requires that the prior mean on the matrix inverse be consistent with the prior mean on the solution, which is natural. 
The second condition demands that, after projection into solution space, the relationship between the rows of $\Ai$ modelled by $W_0$ does not inflate the covariance $\Sigma_0$.
Note that this condition is trivial to enforce for an arbitrary covariance $\bar{\mat{W}_0}$ by setting $W_0 = (\vec b\Trans\bar{\mat W_0} \vec b)^{-1} \bar{\mat{W}_0}$.

\subsection{Remarks}
The result in Proposition~\ref{prop:matrix_prior_equivalent_solution_prior} shows that any result proven for SBI applies immediately to MBI with left-multiplied observations.
Though MBI has more model parameters than SBI, there are situations in which this point of view is more appropriate.
Unlike in SBI, the information obtained in MBI need not be specific to a particular solution vector $\bm{x}^*$ and thus can be propagated and recycled over several linear problems, similar to the notion of subspace recycling \citep{Soodhalter2014}. Secondly, MBI is able to utilise both left- and right-multiplied information, while SBI is restricted to left-multiplied information. This additional generality may prove useful in some applications.

\section{Projection Methods as Inference}
This section discusses a connection between probabilistic numerical methods for linear systems and the classic framework of projection methods for the iterative solution of linear problems. \cref{sec:projection_methods} reviews this established class of solvers, while \cref{sec:projection_connection} presents the novel results.

\subsection{Background} \label{sec:projection_methods}
Many iterative methods for linear systems, including CG and GMRES, belong to the class 
of projection methods \citep[p.~130f.]{Saad2003iterative}. 
Saad describes a projection method as an iterative scheme in which, at each iteration, a solution vector $\vec x_m$ is constructed by projecting $\vec x^*$ into a solution space $\mathbb{X}_m\subset \reals^d$, subject to the restriction that the residual $\vec r_m = \vec b - \A \vec x_m$ is orthogonal to a constraint space $\mathbb{U}_m\subset \reals^d$.

More formally, each iteration of a projection method is defined by two matrices
$X_m, U_m \in \reals^{d\times m}$, and by a starting point $\vec x_0$.
The matrices $X_m$ and $U_m$ each encode the solution and constraint spaces as $\mathbb{X}_m=\mathrm{range}(X_m)$ and $\mathbb{U}_m=\mathrm{range}(U_m)$.
The projection method then constructs $\vec x_m$ as 
$\vec x_m = \vec x_0 + X_m\vec\alpha_m$ with $\vec \alpha_m\in\reals^m$ determined
by the constraint $U_m^\top\vec r_m = \vec 0$. 
This is possible only if $\mat U_m\Trans \A\mat X_m$ is nonsingular, in which case one obtains
\begin{align} \label{eq:projection_step}
	\vec \alpha_m &= (U_m\Trans \A \mat X_m)^{-1} U_m\Trans \vec r_0, \text{ and thus}\\
\label{eq:projection_method_x}
\vec x_m &= \vec x_0 + \mat X_m (\mat U_m\Trans \A \mat X_m)^{-1} U_m\Trans \vec r_0.
\end{align}
From this perspective CG and GMRES perform only a single step with the number of iterations $m$ fixed and determined in advance. 
For CG the spaces are $\mathbb{U}_m = \mathbb{X}_m = K_m(\A, \b)$, while for GMRES they are $\mathbb{X}_m=K_m(\A, \b)$ and $\mathbb{U}_m=\A K_m(\A, \b)$ \citep[Proposition 5.1]{Saad2003iterative}.

\subsection{Probabilistic Perspectives} \label{sec:projection_connection}
In this section we first show, in Proposition~\ref{prop:matrix_prior_implies_projection_method}, that the conditional mean from SBI after $m$ steps corresponds to some projection method.
Then, in Proposition~\ref{thm:generic_projection_method_replication} we prove the converse: that each projection method is also the posterior mean of a probabilistic method, for some prior covariance and choice of information.

\begin{proposition}\label{prop:matrix_prior_implies_projection_method}
Let the columns of $S_m$ 
be linearly independent.
Consider SBI under the prior 
\begin{equation*} 
p(\vec x)=\mathcal{N}(\vec x_0, \mat\Sigma_0),
\end{equation*} 
and with observations $\vec y_m=\mat S_m\Trans\vec b$.
Then the posterior mean $\vec x_m$ in Lemma~\ref{lemma:solution_posterior}
is identical to the iterate from a projection method
defined by the matrices  $U_m=S_m$ and $X_m=\mat\Sigma_0\A\Trans S_m$, and 
the starting vector~$\vec x_0$.
\end{proposition}
\begin{proof} 
Substituting $U_m = S_m$ and $X_m = \mat \Sigma_0 \A\Trans S_m$ 
into Lemma~\ref{lemma:solution_posterior} gives Eq.~\eqref{eq:projection_method_x}, as required.
\qed\end{proof}
The converse to this also holds:
\begin{proposition}\label{thm:generic_projection_method_replication}
Consider a projection method defined by the matrices 
$X_m,U_m\in\reals^{d\times m}$, each with linearly independent columns, and the starting vector
$\bm{x}_0 \in \reals^d$.
Then the iterate $\vec x_m$ in Eq.~\eqref{eq:projection_method_x} is identical to the SBI posterior mean in Lemma~\ref{lemma:solution_posterior} under the prior
\begin{align}
	\label{eq:replicating_prior}
	p(\vec x)=\N(\vec x; \vec x_0, X_mX_m^\top )
\end{align}
when search directions $S_m = U_m$ are used.
\end{proposition}

\begin{proof}
Abbreviate $Z=X_m^\top A^\top U_m$ and 
write the projection method iterate from Eq.~\eqref{eq:projection_method_x} as
\begin{eqnarray*}
\vec x_m = \vec x_0+ X_m Z^{-T} U_m^\top \vec r_0.
\end{eqnarray*} 
Multiply the middle matrix by the identity,
\begin{eqnarray*}
Z^{-T}& = &ZZ^{-1}Z^{-T}=Z(Z^\top Z)^{-1}\\
&=&X_m^\top A^\top U_m(U_m^\top A\Sigma_0 A^\top U_m)^{-1},
\end{eqnarray*}
and insert this into the expression for $\vec x_0$,
 \begin{eqnarray*}
\vec x_m = \vec x_0+ \Sigma_0A^\top U_m(U_m^\top A\Sigma_0 A^\top U_m)^{-1}U_m^\top \vec r_0.
\end{eqnarray*} 
Setting $U_m=S_m$ gives the mean in Lemma~\ref{lemma:solution_posterior}.
\qed\end{proof}

Including a basis of the solution space in the prior may seem problematic.
A direct way to enforce the posterior occupying the solution space is by placing a prior on the coefficients $\vec \alpha$ in $\vec x = \vec x_0 + \mat X_m \vec \alpha$.
Under a unit Gaussian prior $\vec \alpha \sim \N(\vec 0, \mat I)$, the implied prior on $\vec x$ naturally has the form of Eq.~\eqref{eq:replicating_prior}.
However, this prior is nevertheless unsatisfying both since it requires the solution space to be specified \emph{a-priori}, precluding adaptivity in the algorithm, and, perhaps more worryingly, because the posterior uncertainty over the solution is a matrix of zeros even though the solution is not fully identified.
Again taking $Z = X_m\Trans A\Trans U_m$:
\begin{align*}
\mat \Sigma_m&=\mat \Sigma_0 - \mat \Sigma_0 \A\Trans U_m(U_m\Trans \A \mat \Sigma_0 \A\Trans U_m)^{-\!1}\mat U_m\Trans A \mat \Sigma_0
\\ &= \mat X_m\mat X_m\Trans - \mat X_m \mat Z(\mat Z\Trans \mat Z)^{-1}\mat Z\Trans \mat X_m\Trans
\\ &=  \mat X_m\mat X_m\Trans - \mat X_m\mat X_m\Trans
\\ &= \mat 0 .
\end{align*}
\citep{Hennig2015} and \citep{Bartels:2016eh} each proposed to address this issue by adding additional uncertainty in the null space of $X_m$. 
This empirical uncertainty calibration step has not yet been analysed in detail. 
Such analysis is left for future work.
Nevertheless, the proposition provides a probabilistic view for \emph{arbitrary} projection methods and does not require knowledge of $\Ai$, unlike some of the results presented in \citep{Hennig2015,Cockayne:2017} and in the following propositions.

This prior is not unique.
The next proposition establishes more restrictive conditions under which a projection method may have a probabilistic interpretation and still result in a nonzero posterior uncertainty.

\begin{proposition} \label{prop:restrictive_projection_method_prior}
Consider a projection method defined by $X_m, U_m\in\reals^{d\times m}$ and the starting vector $\bm{x}_0$.
Further suppose that $U_m = R X_m$ for some invertible $R \in \reals^{d\times d}$, and that $A^\top R$ is symmetric positive-definite.
Then under the prior
\begin{equation*}
p(\vec x)=\N\left(\vec x; \vec x_0, (A^\top R)^{-1} \right) 
\end{equation*}
and the search directions $S_m = U_m = R X_m$, the iterate in the projection method is identical to the posterior mean in Lemma~\ref{lemma:solution_posterior}.
\end{proposition}
\begin{proof}
First substitute $X_m=R^{-1}U_m$ into Eq.~\eqref{eq:projection_method_x} to obtain
\begin{align*} 
  &\vec x_m = \vec x_0 + \mat R^{-1} U_m (\mat U_m\Trans \A \mat R^{-1} U_m)^{-1} U_m\Trans \vec r_0 \\
  &= \vec x_0 + \mat R^{-1} \A^{-\top} \A\Trans \mat U_m (\mat U_m^\top \A \mat R^{-1} \A^{-\top} \A\Trans \mat U_m)^{-1} \mat U_m^\top \vec r_0 \\
  &= \vec x_0 + \Sigma_0 \A\Trans \mat U_m (\mat U_m^\top \A \Sigma_0 \A\Trans \mat U_m)^{-1} \mat U_m^\top \vec r_0.
\end{align*}
The third line uses $\Sigma_0 = (\A\Trans R)^{-1} = R^{-1}A^{-T}$.
This is equivalent to the posterior mean in Eq.~\eqref{eq:sbi_mean} with $S_m = U_m$.
\qed\end{proof}
A corollary which provides further insight arises when one considers the \emph{polar decomposition} of $A$.
Recall that an invertible matrix $A$ has a unique polar decomposition $A = PH$, where $P \in \reals^{d \times d}$ is orthogonal and $H \in \reals^{d\times d}$ is symmetric positive-definite.
\begin{corollary} \label{corr:polar}
  Consider a projection method defined by $X_m, U_m\in\reals^{d\times m}$ and the starting vector $\bm{x}_0$, and suppose that $U_m = P X_m$, where $P$ arises from the polar decomposition $A = PH$.
  Then under the prior
  \begin{equation*}
  p(\vec x)=\N\left(\vec x; \vec x_0, H^{-1}\right)
  \end{equation*}
  and the search directions $S_m = U_m = P X_m$, the iterate in the projection method is identical to the posterior mean in Lemma~\ref{lemma:solution_posterior}.
\end{corollary}
\begin{proof}
  This follows from Proposition~\ref{prop:restrictive_projection_method_prior}.
  Setting $R = P$ aligns the search directions in Corollary~\ref{corr:polar} with those in Proposition~\ref{prop:restrictive_projection_method_prior}.
  Since $P$ is orthogonal, $P^{-1} = P^\top$, and since $H$ is symmetric positive-definite, $A^\top P = P^\top A = H$ by definition of the polar decomposition, which gives the prior covariance required for Proposition~\ref{prop:restrictive_projection_method_prior}.
\qed\end{proof}
This is an intuitive analogue of similar results in \cite{Hennig2015} and \cite{Cockayne:2017} which show that CG is recovered under certain conditions involving a prior $\Sigma_0 = \Ai$.
When $\A$ is not symmetric and positive definite it cannot be used as a prior covariance.
This corollary suggests a natural way to select a prior covariance still linked to the linear system, though this choice is still not computationally convenient.
Furthermore, in the case that $\A$ is symmetric positive-definite, this recovers the prior which replicates CG described in \cite{Cockayne:2018}.
Note that each of $H$ and $P$ can be stated explicitly as $H = (\A\Trans\A)^\frac{1}{2}$ and $P = \A(\A\Trans\A)^{-\frac{1}{2}}$.
Thus in the case of symmetric positive-definite $A$ we have that $H = \A$ and $P = I$, so that the prior covariance $\Sigma_0 = A^{-1}$ arises naturally from this interpretation.

\section{Preconditioning} \label{sec:preconditioning}
This section discusses probabilistic views on preconditioning.
Preconditioning is a widely-used technique accelerating the convergence of iterative methods \citep[Sections 9 and 10]{Saad2003iterative}.
A preconditioner~$P$ is a nonsingular matrix satisfying two requirements: 
\begin{enumerate}
  \item Linear systems $Pz=c$ can be solved at low computational cost (i.e.~``analytically'')
  \item $P$ is ``close'' to $A$ in some sense.
\end{enumerate}
In this sense, solving systems based upon a preconditioner can be viewed as approximately inverting $A$, and indeed many preconditioners are constructed based upon this intuition. 
One distinguishes between \emph{right preconditioners} $P_r$
and \emph{left preconditioners} $P_l$, depending on whether they act on $A$ from the left or the right. 
Two-sided preconditioning with nonsingular matrices $P_l$ and $P_r$ transforms implicitly Eq.~\eqref{eq:system} into a new linear problem
\begin{equation} \label{eq:system_preconditioned}
	\mat P_l \A \mat P_r \,\vec z^*=\mat P_l \b, \qquad \text{with}\quad \vec x^*=P_r\vec z^*.
\end{equation}
The preconditioned system can then be solved using arbitrary projection methods as described in Section~\ref{sec:projection_methods}, from the starting point $\vec z_0$ defined by $\vec x_0 = \mat P_r \vec z_0$.
The probabilistic view can be used to create a nuanced description of preconditioning as a form of prior information.
In the SBI framework, Proposition~\ref{prop:rightP} below shows that solving a right-preconditioned system is equivalent to modifying the prior, while in Proposition~\ref{prop:leftP} shows that left-preconditioning is equivalent to making a different choice of observations.

\begin{proposition}[Right preconditioning]\label{prop:rightP}
Consider the right-preconditioned system
\begin{equation}\label{eq:rightP}
 \A \mat P_r \vec z^*= \b \qquad \text{where} \quad \vec x^* = \mat P_r \vec z^*.
 \end{equation}
SBI on Eq.~\eqref{eq:rightP} under the prior
\begin{equation}\label{eq:rightP_prior}
\vec z\sim\N(\vec z; \vec z_0, \Sigma_0)
\end{equation}
is equivalent to solving Eq.~\eqref{eq:system} under the prior
\begin{equation*}
\vec x \sim \N(\vec x; \mat P_r\vec z_0, \mat P_r \mat \Sigma_0 \mat P_r\Trans) . 
\end{equation*}
\end{proposition}

\begin{proof}
Let $p(x)=\N(\vec x; \vec x_0, \Sigma_r)$. Lemma~\ref{lemma:solution_posterior} implies that after observing information from search directions~$S_m$,
the posterior mean equals
\begin{equation*}
\vec x_m = \vec x_0 + \Sigma_r A^\top  S_m (S_m^\top  A \Sigma_r A^\top S_m)^{-1}S_m^\top  \vec r_0
\end{equation*}
where $\vec r_0 = \vec b-A\vec x_0$. Setting $\vec x_0=P_r \vec z_0$ and letting $\Sigma_r=P_r\Sigma_0P_r^\top $ gives
 \begin{equation*}
\vec x_m = P_r\vec z_0 + P_r\Sigma_ 0B^\top  S_m (S_m^\top  B \Sigma_0 B^\top S_m)^{-1}S_m^\top  \hat{\vec r}_0
\end{equation*}
where $B\ce AP_r$ and $\hat{\vec r}_0 = \vec b-B\vec z_0$.
Left multiplying by $P_r^{-1}$ shows that this is equivalent to 
 \begin{align*}
\vec z_m&\ce P_r^{-1}\vec x_m \\
&=  \vec z_0 + \Sigma_ 0B^\top  S_m (S_m^\top  B \Sigma_0 B^\top S_m)^{-1}S_m^\top  \hat{\vec r}_0.
\end{align*}
Thus $\vec z_m$ is the posterior mean of the system $B\vec z^* = \vec b$ with prior Eq.~\eqref{eq:rightP_prior}
after observing search directions $S_m$.
\qed\end{proof}

\begin{proposition}[Left preconditioning]\label{prop:leftP}
Consider\newline the left-preconditioned system
\begin{equation}\label{eq:leftP}
 P_l A \vec x^*= P_l \vec b 
 \end{equation}
And the SBI prior
\begin{equation*}
  p(\vec x) = \N(\vec x; \vec x_0, \mat \Sigma_0).
\end{equation*}
Then the posterior from SBI on Eq.~\eqref{eq:leftP} under search directions $S_m$ is equivalent to the posterior from SBI applied to the system Eq.~\eqref{eq:system} under search directions $\mat P_l\Trans S_m$.
\end{proposition}
\begin{proof}
Lemma~\ref{lemma:solution_posterior} implies that after observing search directions~$T_m$,
the posterior mean over the solution of Eq.~\eqref{eq:system} equals
\begin{equation*}
\vec x_m = \vec x_0 + \Sigma_0 A^\top  T_m (T_m^\top  A \Sigma_0 A^\top T_m)^{-1}T_m^\top  \vec r_0
\end{equation*}
where $\vec r_0 = \vec b-A\vec x_0$.
Setting $T_m= P_l^\top S_m$ gives
\begin{equation*}
\vec x_m = \vec x_0 + \Sigma_0 B^\top  S_m (S_m^\top  B \Sigma_0 B^\top  S_m)^{-1}S_m^\top P_l \hat{\vec r}_0
\end{equation*}
where $\mat B \ce P_l A$ and $\hat{\vec r}_0 = P_l \vec b-P_lA\vec x_0$.
Thus, $\vec x_m$ is the posterior mean of the system $B\vec x^* = P_l \vec b$
after observing search directions $S_m$.
\qed\end{proof}
If a probabilistic linear solver has a posterior mean which coincides with a projection method (as discussed in Section~\ref{sec:projection_methods}), the Propositions~\ref{prop:rightP} and \ref{prop:leftP} show how to obtain a probabilistic interpretation of the \emph{preconditioned} version of that algorithm.
Furthermore, the equivalence demonstrated in Section~\ref{sec:equivalence} shows that the reasoning from Propositions \ref{prop:rightP} and~\ref{prop:leftP} carries over to MBI based on left-multiplied observations: right-preconditioning corresponds to a change in prior belief,
while left-pre\-cond\-ition\-ing corresponds to a change in observations.

We do not claim that this probabilistic interpretation of preconditioning is unique.
For example, when using MBI with right-multiplied observations, the same line of reasoning can be used to show the converse: right-preconditioning corresponds to a change in the observations and left-preconditioning to a change in the prior.

\section{Conjugate Gradients} \label{sec:cg}
Conjugate gradients has been studied from a probabilistic point of view before by \citet{Hennig2015} and \citet{Cockayne:2018}.
This section generalizes the results of \citet{Hennig2015} and leverages Proposition \ref{prop:matrix_prior_equivalent_solution_prior} for new insights on BayesCG. 
For this Section (but not thereafter) assume that $\A$ is a symmetric and positive definite matrix.

\subsection{Left-multiplied view}
The BayesCG algorithm proposed by \citet{Cockayne:2018} encompasses conjugate gradients as a special case.
BayesCG uses left-multiplied observations and was derived in the solution-based perspective.

The posterior in Lemma~\ref{lemma:solution_posterior} does not immediately result in a practical algorithm as it involves the solution of a linear system based on the matrix
$\mat S_m\Trans \A \mat \Sigma_0 \A\Trans \mat S_m\in\reals^{m\times m}$, which requires $\mathcal{O}(m^3)$ arithmetic operations.
BayesCG avoids this cost by constructing search directions that are $A\Sigma_0A^\top $-orthonormal, 
as shown below, see \cite[Proposition~7]{Cockayne:2018}.

\begin{proposition}[Proposition 7 of \cite{Cockayne:2018} (BayesCG)] \label{prop:bcg_search}
	Let $\tilde{\vec s}_1 = \vec b - A \vec x_0$, and let $\vec s_1 = \tilde{\vec s}_1 / \norm{\tilde{\vec s}_1}$.
	For $j = 2,\dots,m$ let
	\begin{align*}
		\tilde{\vec s}_j &= \b - \A\vec x_{j-1} - \inner{\vec b - \A \vec x_{j-1}, \vec s_{j-1}}_{\A \Sigma_0 \A\Trans} \vec s_{j-1} \\
		\vec s_j &= \tilde{\vec s}_j / \norm{\tilde{\vec s}_j}_{\A \Sigma_0 \A\Trans} .
	\end{align*}
	Then the set $\set{\vec s_1, \dots, \vec s_m}$ is $\A \Sigma_0 \A\Trans$-orthonormal, and consequently $S_m\Trans \A \Sigma_0 \A \Trans S_m = I$.
\end{proposition}
With these search directions constructed, BayesCG becomes an iterative method:

\begin{proposition}[Proposition~6 of \citet{Cockayne:2018}]\label{prop_3}
	Using the search directions from Proposition \ref{prop:bcg_search}, the posterior from Lemma \ref{lemma:solution_posterior} reduces to:
	\begin{align*}
		\vec x_m &= \vec x_{m-1} + \mat \Sigma_0 \A\Trans \vec s_m (\vec s_m \Trans (\vec b - A \vec x_{m-1})) \\
		\mat \Sigma_m &= \mat \Sigma_{m-1} - \mat \Sigma_0 \A\Trans \vec s_m \vec s_m\Trans \A \Sigma_0
	\end{align*}
\end{proposition}
In Proposition~4 of \citet{Cockayne:2018} it was shown that the BayesCG posterior mean corresponds to the CG solution estimate when the prior covariance is taken to be $\mat\Sigma_0 = \Ai$, though this is not a practical choice of prior covariance as it requires access to the unavailable $A^{-1}$.
Furthermore, in Proposition~9 it was shown that when using the search directions from Proposition \ref{prop:bcg_search}, the posterior mean from BCG has the following optimality property:
\begin{equation*}
	x_m = \argmin_{\vec x \in K_m(\mat\Sigma_0 \A\Trans \A, \mat \Sigma_0 \A\Trans \b)} \norm{\vec x - \vec x^*}_{\Sigma_0^{-1}}
\end{equation*}
Note that this is now a trivial special case of Proposition~\ref{prop:generic_solution_optimality}.

The following proposition leverages these results along with Proposition~\ref{prop:matrix_prior_equivalent_solution_prior} to show that there exists an MBI method which, under a particular choice of prior and with a particular methodology for the generation of search directions, is consistent with CG.

\begin{proposition}
  Consider the MBI prior
  \begin{equation*}
    p(\vecm{\Ai}) = \N(\vecm{\Ai}; \vecm{\A_0^{-1}}, \mat \Ai \otimes \mat W_0)
  \end{equation*}
  where $W_0 \in \reals^d$ is symmetric positive-definite and so that $\b\Trans W_0 \b = 1$. 
  Suppose left-multiplied information is used, and that the search directions are generated sequentially according to:
  \begin{align*}
    \tilde{\vec s}_1 &= (\mat I - \A \A_0^{-1}) \b \\
    \vec s_1 &= \frac{\tilde{\vec s_1}}{\norm{\tilde{\vec s_1}}_A}
  \end{align*}
  and for $j=2,\dots,m$
  \begin{align*}
    \tilde{\vec s}_{j} &= (\mat I - \A \A_{j-1}^{-1}) \b - \b\Trans(\mat I - \A\A_{j-1}^{-1})\Trans \A\vec s_{j-1} \cdot \vec s_{j-1} \\
    \vec s_j &= \frac{\tilde{\vec s_j}}{\norm{\tilde{\vec s_j}}_A}.
  \end{align*}
  Then it holds that the implied posterior mean on solution space, given by $\A_m^{-1} \b$, corresponds to the CG solution estimate after $m$ iterations, with starting point $\vec x_0 = \A_0^{-1} \b$.
\end{proposition}

\begin{proof}
  First note that, by Proposition~\ref{prop:matrix_prior_equivalent_solution_prior}, since left-multi\-plied observations are used and since $\b\Trans W_0 \b = 1$, the implied posterior distribution on solution space from MBI is identical to the posterior distribution from SBI under the prior 
  \begin{equation*}
    p(\vec x) = \N(\vec x; \A_0^{-1} \b, \Ai) .
  \end{equation*}
  It thus remains to show that the sequence of search directions generated is identical to those in Proposition~\ref{prop:bcg_search} for this prior.
  For $\tilde{\vec s_1}$:
  \begin{equation*}
    \tilde{\vec s_1} = (\mat I - \A \A_0^{-1}) \b = \b - \A \vec x_0
  \end{equation*}
  as required.
  For $\tilde{\vec s_j}$:
  \begin{align*}
    \tilde{\vec s_j} 
    &= \vec (I - \A \A_{j-1}^{-1}) \b - \b\Trans (\mat I - \A\A_{j-1}^{-1})\Trans \A \vec s_{j-1} \cdot \vec s_{j-1} \\
    &= \vec b - \A \vec x_{m-1} - (\b - \A \vec x_{j-1})\Trans \A \vec s_{j-1} \cdot \vec s_{j-1} \\
    &= \vec b - \A \vec x_{m-1} - \inner{\b - \A \vec x_{j-1}, \vec s_{j-1}}_\A \cdot \vec s_{j-1}
  \end{align*}
  where the second line uses that $\Ai_{j-1} \b = \vec x_{j-1}$.
  Thus, the search directions coincide with those in Proposition~\ref{prop:bcg_search}.
  It therefore holds that the implied posterior mean on solution space, $\Ai_m \b$, coincides with the solution estimate produced by CG.
\qed\end{proof}

\subsection{Right-multiplied view}
\label{sec:cg_in_mbi}
Interpretations of CG (and general projection methods) that use right-multiplied observations seems to require more care than those based on left-multiplied observations.
Nevertheless, \citet{Hennig2015} provided an interpretation for CG in this framework, essentially showing\footnote{Algorithm~\ref{alg:spd_proj} is not included in this form in the op.cit.} that Algorithm~\ref{alg:spd_proj} reproduces both the search directions and solution estimates from CG under the prior 
\begin{equation*}
p(\Ai)=\N(\vecm{\Ai}; \vecm{\alpha\mat I}, \beta \mat A^{-1} \sk \mat A^{-1}).
\end{equation*}
where $\alpha\in \Re\setminus\{0\}$, $\beta \in \Re^+$ and $\sk$ denotes the symmetric Kronecker product (see Section~\ref{sec:symmkron}).
The posterior under such a prior is described in Lemma 2.2 of \citet{Hennig2015} (see Lemma \ref{prop:sym_kron_posterior}), though we note that the sense in which the solution estimate $\vec x_m$ output by this algorithm is related to the posterior over $A^{-1}$ differs from that in the previous section, in the sense that $\Ai_m \b \neq \vec x_m$. (More precisely, $\vec x_m=\Ai_m (\b-\A\vec x_0)-\vec x_0 - (1-\alpha_m) \vec d_m$, as the CG estimate is corrected by the step size computed in line 6. Fixing this rank-1 discrepancy would complicate the exposition of Algorithm 1 and yield a more cumbersome algorithm).
The following proposition generalizes this result.

\begin{proposition}
\label{thm:cg_right_multiplied}
Consider the prior 
\begin{eqnarray*}
p(\Ai)=\N(\Ai; \alpha\mat I, (\beta \mat I+\gamma \Ai)\sk (\beta \mat I + \gamma \Ai)).
\end{eqnarray*}
For all choices $\alpha\in\Re\setminus\{0\}$ and $\beta,\gamma \in \Re_{+,0}$ with $\beta + \gamma>0$, \cref{alg:spd_proj} is equivalent to CG, in the sense that it produces the exact same sequence of estimates $\vec x_i$ and scaled search directions $\vec s_i$.
\end{proposition}
\begin{proof}
  The proof is extensive and has been moved to Appendix~\ref{sec:proofs}.
\qed \end{proof}
  
\begin{algorithm}
\caption{The algorithm referred to by Proposition~\ref{thm:cg_right_multiplied}, which reproduces the search directions and solution estimates from CG.}\label{alg:spd_proj}
\begin{algorithmic}[1]
\LState $\vec x_0\gets \Ai_0\b$ \Comment{initial guess}
\LState $\vec r_0\gets \A \vec x_0 - \b $
\For{$i = 1, \ldots , m$}
\LState \algalign{\Ai}{\vec d_i}{-\Ai_{i-1} \vec r_{i-1}} \Comment{compute optimization direction}
\LState \algalign{\Ai}{\vec z_i}{\A \vec d_i} \Comment{\textbf{observe}}
\LState \algalign{\Ai}{\alpha_i}{-\frac{\vec d_i\Trans\vec r_{i-1}}{\vec d_i\Trans \vec z_i}} \Comment{optimal step-size}
\LState \algalign{\Ai}{\vec s_{i}}{\alpha_i \vec d_i} \Comment{re-scale step}
\LState \algalign{\Ai}{\vec y_i}{\alpha_i\vec z_i} \Comment{re-scale observation}
\LState \algalign{\Ai}{\vec x_i}{\vec x_{i-1}+\vec s_i} \Comment{update estimate for $\vec x$}
\LState \algalign{\Ai}{\vec r_i}{\vec r_{i-1}+\vec y_i} \Comment{new gradient at $\vec x_i$}
\LState \algalign{\Ai}{\Ai_i}{\Exp_{p(\Ai\mid \mat S, \mat Y)}\Ai } \Comment{estimate $\Ai$}
\EndFor
\LState \Return $\vec x_m$ 
\end{algorithmic}
\end{algorithm}
Note that, unlike previous propositions, Proposition~\ref{thm:cg_right_multiplied} proposes a prior that does not involve $\Ai$ for the case when $\gamma = 0$.

\section{GMRES} \label{sec:gmres}
The \textit{Generalised Minimal Residual Method} \citep[Section 6.5]{Saad2003iterative} 
applies to general nonsingular matrices $A$. At iteration $m$, GMRES minimises 
the residual over the affine space $\bm{x}_0 + K_m(A,\bm{ r}_0)$. 
That is, $\bm{r}_m = \bm{r}_0 - A\bm{x}_m$ satisfies
\begin{eqnarray} \label{eq:gmres_optimal}
\|\bm{r}_m\|_2 &= &\min_{\bm{x} \in K_m(A, \bm{r}_0)}{\|A\bm{x}-\bm{r}_0\|_2}\\
&=&\min_{x \in \bm{x}_0+K_m(A, \bm{r}_0)}{\|A\bm{x} - \bm{b}\|_2}.\nonumber
\end{eqnarray}
Since $A \bm{x}-\bm{b}= A  (\bm{x}- \bm{x}^*)$, this corresponds to minimizing the error in the $A^\top A$ norm.

We present a brief development of GMRES, starting with Arnoldi's method
(Section~\ref{sec:arnoldi}) and the GMRES algorithm (Section~\ref{sec:gmres_detail}),
before presenting our Bayesian interpretation (Section~\ref{sec:BGMRES}).

\subsection{Arnoldi's Method} \label{sec:arnoldi}
GMRES uses Arnoldi's method \cite[Section 6.3]{Saad2003iterative}
to construct orthonormal bases for Krylov spaces of general, nonsingular matrices~$A$.
Starting with $\bm{q}_1 = \bm{r}_0/\|\bm{r}_0\|_2$,  Arnoldi's method  recursively computes the orthonormal basis 
\begin{equation*}
Q_m = \begin{bmatrix}\bm{q}_1 &\ldots& \bm{q}_m\end{bmatrix}\in \reals^{d \times m}
\end{equation*}
for $K_m(A, \bm{r}_0)$. The basis vectors satisfy the relations
\begin{equation} \label{eq:arnoldi_recursive}
	A Q_m = Q_{m+1} \tilde{H}_m=Q_m H_m+h_{m+1,m}\bm{q}_{m+1}\bm{e}_m^\top 
	\end{equation}
and $Q_m^\top AQ_m = H_m$, where the \emph{upper Hessenberg} matrix $H_m$ is defined as
\begin{equation*}
	H_m = \begin{bmatrix} 
		h_{11}  & h_{12} & h_{13} & \dots & h_{1,m-1} & h_{1m} \\
		h_{21}  & h_{22} & h_{23} & \dots & h_{2,m-1} & h_{2m} \\
		0 		& h_{32} & h_{33} & \dots & h_{3,m-1} & h_{3m} \\
		\vdots		& 0		 & h_{43} & \dots & h_{4,m-1} & h_{3m} \\
		\vdots	&  & \ddots & \ddots & \vdots   & \vdots \\
		0		& \dots		 & \dots	  & 0& h_{m,m-1} & h_{mm}
	\end{bmatrix} \in \reals^{m \times m}
	\end{equation*}
and
\begin{equation*}
\tilde{H}_m=\begin{bmatrix} H_m \\ h_{m+1,m}\bm{e}_m^\top \end{bmatrix}\in \reals^{(m+1)\times m}.
\end{equation*}

\subsection{GMRES} \label{sec:gmres_detail}
GMRES computes the iterate
\begin{equation*}
\bm{x}_m =\bm{x}_0+ Q_m \bm{c}_m
\end{equation*} 
based on the optimality condition in Eq.~\eqref{eq:gmres_optimal}, which can equivalently be expressed as
\begin{align}
\bm{c}_m &= \argmin_{\bm{c} \in \reals^m} \norm{A Q_m \bm{c} - \bm{r}_0}_2 \label{eq:gmres_initial_optim}\\
&= \left((AQ_m)^\top (AQ_m)\right)^{-1}(AQ_m)^\top \bm{r}_0.\nonumber
\end{align}
Thus
\begin{equation}\label{eq:xmgmres}
\bm{x}_m =\bm{x}_0+ Q_m \left(Q_m^\top A^\top AQ_m\right)^{-1}Q_m^\top A^\top \bm{r}_0,
\end{equation} 
confirming that GMRES is a projection method with 
$\mat X_m=Q_m$ and $\mat U_m=\A\mat Q_m$.

GMRES solves the least squares problem in Eq.~\eqref{eq:gmres_initial_optim}.
efficiently by projecting it to a lower dimensional space via Arnoldi's method.
To this end, express the starting vector in the Krylov basis,
\begin{equation*}
\bm{r}_0 =\|\bm{r}_0\|_2 \bm{q}_1=\|\bm{r_0}\|_2 Q_{m+1} \bm{e_1},
\end{equation*}
and exploit the Arnoldi recursion from Eq.~\eqref{eq:arnoldi_recursive},
\begin{eqnarray*}
AQ_m\bm{c}-\bm{r_0}=Q_{m+1}\left(\tilde{H}_{m+1} \bm{c}-\|\bm{r}_0\|_2 \bm{e}_1\right),
\end{eqnarray*}
followed by the unitary invariance of the two-norm,
\begin{equation*}
\|A Q_m \bm{c} - \bm{r}_0\|_2 = \|\tilde{H}_m \bm{c} - \norm{\bm{r}_0}_2 \, \bm{e}_1\|_2.
\end{equation*}
Thus, instead of solving the least squares problem Equation \eqref{eq:gmres_initial_optim} 
with $d$ rows, GMRES solves instead a problem with only $m+1$ rows,
\begin{equation} \label{eq:gmres_final_optim}
	\bm{c}_m = \argmin_{\bm{c} \in \reals^m} \norm{\tilde{H}_m \bm{c} - \norm{\bm{r}_0}_2 \, \bm{e}_1}_2.
\end{equation}
The computations are summarized in  Algorithm~\ref{alg:gmres}.

\begin{algorithm}
\caption{GMRES \cite[Algorithm 6.9]{Saad2003iterative}}\label{alg:gmres}
\begin{algorithmic}[1] 
\LState $\vec r_0\gets \b - \A\vec x_0$, $\beta \gets \norm{\vec r_0}_2$, $\vec q_1\gets \vec r_0/\beta$
\For{$j=1, \ldots, m$}
\LState $\vec w_j\gets A \vec q_j$
\For{$i=1, \ldots, j$}
\LState \algalign{h_{ij}}{h_{ij}}{\inner{\vec w_j, \A\vec q_i}}
\LState \algalign{h_{ij}}{\vec w_j}{\vec w_j - h_{ij} \vec q_i}
\EndFor
\LState $h_{j+1, j}\gets \|\vec w_j\|_2$
\If{$h_{j+1,j} =0$}
\LState $m\gets j$, go to 14
\EndIf  
\LState $\vec q_{j+1} \gets \vec w_j/h_{j+1,j}$
\EndFor
\State Define  $\tilde{H}_m\in\reals^{(m+1)\times m}$ with elements $h_{ij}$
\State \algalign{\vec c_m}{\vec c_m}{\argmin_{\vec c}{\|\tilde{H}_m\vec c- \beta \vec e_1\|_2}}
\State \algalign{\vec c_m}{\vec x_m}{\vec x_0 + Q_m \vec c_m}
\end{algorithmic}
\end{algorithm}

\subsection{Bayesian Interpretation of GMRES}\label{sec:BGMRES}
We now present probabilistic linear solvers with posterior means that coincide with the solution estimate from GMRES.
\subsubsection{Left-multiplied view}
\begin{proposition}\label{prop:BG}
Under the SBI prior
\begin{equation*}p(\vec x)=\N(\vec x; \vec x_0, \mat{\Sigma}_0) \qquad \text{where} \quad \mat \Sigma_0=(\A\Trans\A)^{-1}
\end{equation*}
and the search directions $U_m = A Q_m$, the posterior mean 
is identical to the GMRES iterate $\vec x_m$ in Eq.~\eqref{eq:xmgmres}.
\end{proposition}
\begin{proof}
Substitute $R=A$ and $U_m = \A Q_m$ into Proposition~\ref{prop:restrictive_projection_method_prior}.
\qed\end{proof}
Proposition~\ref{prop:BG} is intuitive in the context 
of Proposition~\ref{prop:generic_solution_optimality}: Setting 
$\Sigma_0 = (\A\Trans \A)^{-1}$ ensures that the norm being minimised 
coincides with that of GMRES, as does the solution space $X_m = A Q_m$.
This interpretation exhibits an interesting duality with CG for which $\mat \Sigma_0=\Ai$.

Another probabilistic interpretation follows from Proposition \ref{thm:generic_projection_method_replication}.

\begin{corollary}\label{corr:BG}
Under the prior
 \begin{align}\label{eq:gmres_prior_arnoldi}
	p(\vec x)=\N(\vec x; \vec x_0, \Sigma_0) \qquad \text{where}\quad \Sigma_0=Q_mQ_m^\top ,
\end{align}
and with observations $\vec y_m =Q_m^\top  \b$, the posterior mean from SBI is identical to the GMRES iterate $\vec x_m$ in Eq.~\eqref{eq:xmgmres}.
\end{corollary}

Note that Proposition~\ref{prop:BG} has a posterior covariance which is not practical, as it involves $\Ai$. 
\citep{Cockayne:2017} proposed replacing $\Ai$ in the prior covariance with a preconditioner to address this, which does yield a practically computable posterior, but this extension was not explored here.
Furthermore, that approach yields poorly calibrated posterior uncertainty, as described in that work.
Corollary~\ref{corr:BG} does not have this drawback, but the posterior covariance is a matrix of zeroes.

\subsubsection{Right-multiplied view}
As for CG in Section \ref{sec:cg_in_mbi}, finding interpretations of GMRES that use right-multiplied observations appears to be more difficult.
\begin{proposition}
Under the prior 
\begin{equation}
p(\mat A^{-1})=\mathcal{N}(\mat 0, \mat \Sigma \kr \mat I)
\end{equation}
and given $\mat Y_m=\A \mat Q_m$, the implied posterior mean on the solution space given by $\mat A_m^{-1} \vec b$ is equivalent to the GMRES solution.
This correspondence breaks when $\vec x_0\neq \vec 0$.
\end{proposition}
\begin{proof}
Under this prior, $\b$ applied to the posterior mean is 
\begin{align*}
	\mat A_m^{-1}\b=&\mat A_0^{-1}\b+(\mat Q_m-\mat A_m^{-1}\mat Y_m)(\mat Y_m\Trans \mat Y_m)^{-1}\mat Y_m\Trans\b
	\\=&\mat Q_m(\mat Y_m\Trans \mat Y_m)^{-1}\mat Y_m\Trans \b
	\\=&\mat Q_m(\mat Q_m\Trans \A\Trans\A\mat Q_m)^{-1}\mat Q_m\Trans \A\Trans \b
\end{align*}
which is the GMRES projection step if $\vec x_0=\vec 0$.
\qed\end{proof}
\subsection{Simulation Study}
\label{sec:simulation}
\tikzset{external/force remake=false}

In this section the simulation study of \citet{Cockayne:2018} will be replicated to demonstrate that the uncertainty produced from GMRES in Proposition \ref{prop:BG} is similarly poorly calibrated, owing to the dependence of $\mat Q_m$ on $\vec x^*$ by way of its dependence on $\vec b$.
Throughout the size of the test problems is set to $d=100$.
The eigenvalues of $\A$ were drawn from an exponential distribution with parameter $\gamma=10$, and eigenvectors uniformly from the Haar-measure over rotation-matrices (see \citet{Diaconis1987}).
In contrast to \citet{Cockayne:2018} the entries of $\b$ are drawn from a standard Gaussian distribution, rather than $\vec x_*$.
By Lemma \ref{lem:gaussian_projection}, the prior is then perfectly calibrated for this scenario, providing justification for the expectation that the posterior should be equally well-calibrated for $m\geq 1$.

Figure \ref{fig:traces} shows on the left the convergence of GMRES and on the right the convergence rate of the trace of the posterior covariance.

Figure \ref{fig:uq} repeats the uncertainty quantification study of \citet{Cockayne:2018}.
\citet{Cockayne:2018} argue that if the uncertainty is well-calibrated then $\vec x^*$ can be considered as a draw from the posterior.
Under this assumption, i.e.~$\mat \Sigma_m^{-\nicefrac{1}{2}}(\vec x^*-\vec x_m) \sim \N(\vec 0, \vec I)$ they derive the test statistic:
\begin{equation*}
Z(\vec x^*)\ce \norm{\Sigma_m^{-\nicefrac{1}{2}}(\vec x^*-\vec x_m)}\sim \chi^2_{d-m}.
\end{equation*}
It can be seen that the same poor uncertainty quantification occurs in BayesGMRES; even after just 10 iterations, the empirical distribution of the test statistic exhibits a profound left-shift, indicating an overly conservative posterior distribution.
Producing well-calibrated posteriors remains an open issue in the field of probabilistic linear solvers.

\begin{figure*}
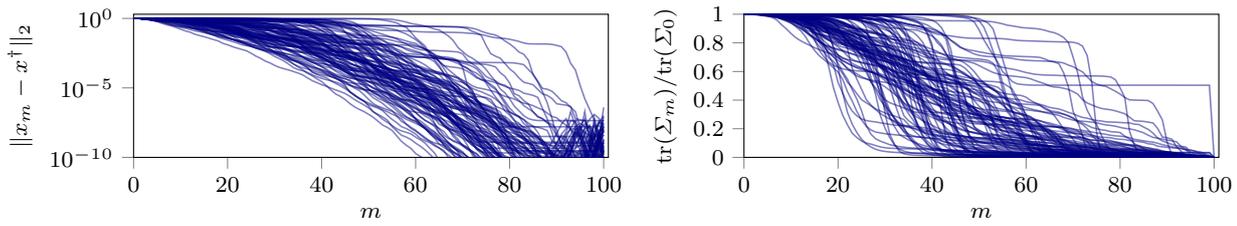
 
\label{fig:traces}
 \setlength{\figwidth}{0.45\textwidth}
 \setlength{\figheight}{0.2\textwidth}
 \input{simulation/notebooks/figures/mean_convergence__seed_12345__eq_type_spd_reverse.tikz}
 \input{simulation/notebooks/figures/var_convergence__seed_12345__eq_type_spd_reverse.tikz}
 \caption{
 Convergence of posterior mean and variance of the probabilistic interpretation of GMRES from Proposition \ref{prop:BG}. 
}
\end{figure*}

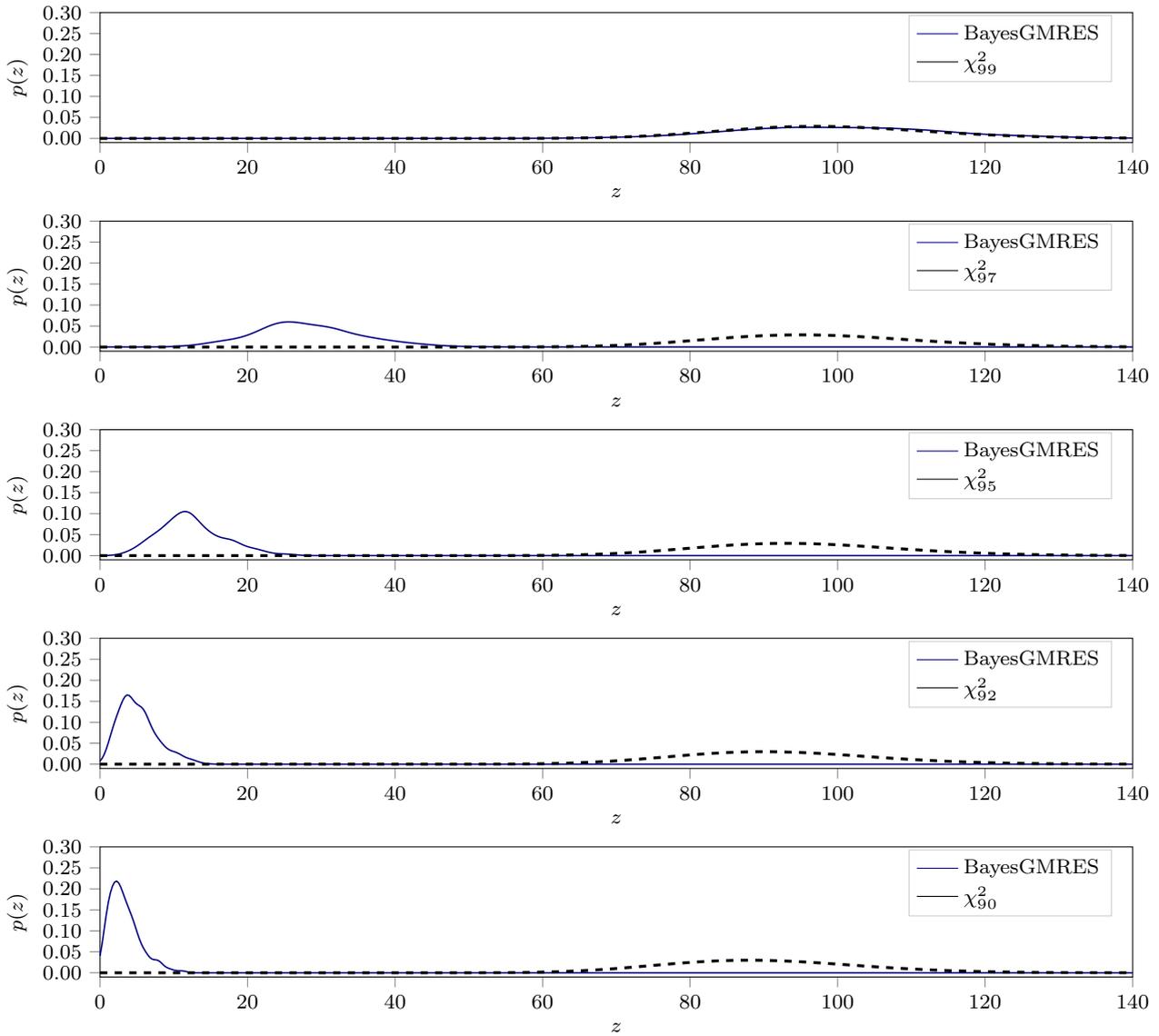
\begin{figure*}
\centering
\label{fig:uq}
 \setlength{\figwidth}{0.95\textwidth}
 \setlength{\figheight}{0.2\textwidth}
 \input{simulation/notebooks/figures/bgmres_gaussian_uq__seed_12345__eq_type_spd_reverse__iteration_1.tikz}
  
 \input{simulation/notebooks/figures/bgmres_gaussian_uq__seed_12345__eq_type_spd_reverse__iteration_3.tikz}
 
 \input{simulation/notebooks/figures/bgmres_gaussian_uq__seed_12345__eq_type_spd_reverse__iteration_5.tikz}
 
 \input{simulation/notebooks/figures/bgmres_gaussian_uq__seed_12345__eq_type_spd_reverse__iteration_8.tikz}

 \input{simulation/notebooks/figures/bgmres_gaussian_uq__seed_12345__eq_type_spd_reverse__iteration_10.tikz}

\caption{
Assessment of the uncertainty quantification. Plotted are kernel density estimates for the statistic $Z$ based on $500$ randomly sampled test
problems for steps $m=\{1, 3, 5, 8, 10\}$.
These are compared with the theoretical  distribution  of $Z$ when the posterior distribution is well-calibrated. 
}
\end{figure*}

\section{Discussion}
We have established many new connections between probabilistic linear solvers and a broad class of iterative methods.
Matrix-based and solution-based inference were shown to be equivalent in a particular regime, showing that results from SBI transfer to MBI with left-multiplied observations.
Since SBI is a special case of MBI, future research will establish what additional benefits the increased generality of MBI can provide.

We also established a connection between the wide class of projection methods and probabilistic linear solvers. The common practise of preconditioning has an intuitive probabilistic interpretation, and all probabilistic linear solvers can be interpreted as projection methods.
While the converse was shown to hold, the conditions under which generic projection methods can be reproduced are somewhat restrictive; however, GMRES and CG, which are among the most commonly used projection methods, have a well-defined probabilistic interpretation.
Probabilistic interpretations of other widely used iterative methods can, we anticipate, be established from the results presented in this work.

Posterior uncertainty remains a challenge for probabilistic linear solvers. 
Direct probabilistic interpretations of CG and GMRES yield posterior covariance matrices which are not always computable, and even when the posterior can be computed the uncertainty remains poorly calibrated.
This is owed to the dependence of the search directions in Krylov methods on $\A\vec x^* = \b$, resulting in an algorithm which is not strictly Bayesian.
Mitigating this issue without sacrificing the fast rate of convergence provided by Krylov methods remains an important focus for future work.
 
\section*{Acknowledgements}
Ilse Ipsen was supported in part by NSF grant DMS-1760374.
Mark Girolami was supported by EPSRC grants [EP/R034710/1, EP/R018413/1, EP/R004889/1, EP/P020720/1], an EPSRC Established Career Fellowship EP/J016934/3, a Royal Academy of Engineering Research Chair, and The Lloyds Register Foundation Programme on Data Centric Engineering.
Philipp Hennig was supported by an ERC grant [757275\allowbreak/PANAMA].
\begin{appendices}

\section{Properties of Kronecker Products} \label{sec:kron_properties}

The following identities about Kronecker products and the vectorization operator are easily derived, but recalled here for the convenience of the reader:

\begin{align}
	\tag{K1} (A \otimes B) \vecm{C} = & \vecm{A C B^\top} \label{eq:K1}
	\\\tag{K2} (A \otimes B) (C \otimes D) = & (AC) \otimes (BD) \label{eq:K2}
	\\\tag{K3} (A \otimes B)^{-1} = & A^{-1} \otimes B^{-1} \label{eq:K3}
	\\\tag{K4} (A \otimes B)^\top = & A^\top \otimes B^\top \label{eq:K4}
	\\\tag{K5} (A + B) \otimes C = & A\otimes C + B\otimes C \label{eq:K5}
\end{align}

\subsection{The Symmetric Kronecker Product} \label{sec:symmkron}

\begin{definition}[symmetric Kronecker-product]

The \emph{symmetric} Kronecker-product for two square matrices $\mat A,\mat B\in \Re^{N\times N}$ of equal size is defined as
$$\mat A\sk\mat B\ce \mat \Gamma (\mat A\kr \mat B)\mat \Gamma$$
where $[\mat \Gamma]_{ij,kl}\ce\nicefrac{1}{2}\delta_{ik}\delta_{jl}+\nicefrac{1}{2}\delta_{il}\delta_{jk}$ satisfies $$\mat \Gamma\vecm{\mat C}=\nicefrac{1}{2}\vecm{\mat C}+\nicefrac{1}{2}\vecm{\mat C\Trans}$$ for all square-matrices $\mat C\in\Re^{N\times N}$.
\end{definition}

\begin{proposition}[Theorem 2.3 in \citet{Hennig2015}]
\label{prop:sym_kron_posterior}
\newcommand{\G}{\mat G}
Let $\mat W\in \Re^{d\times d}$ be symmetric and positive definite.
Assume a Gaussian prior of symmetric mean $\Ai_0$ and covariance $\mat W\sk \mat W$ on the elements of a symmetric matrix $\Ai$.
After $m$ linearly independent noise-free observations of the form $\mat S = \Ai \mat Y$, $\mat Y\in \Re^{d\times m}, \operatorname{rk}(\mat Y)=m$, the posterior belief over $\Ai$ is a Gaussian with mean
\begin{align*}
\Ai_m&=\Ai_0+(\mat S - \Ai_0 \mat Y)\G\mat Y\Trans \mat W \\
&\qq + \mat W\mat Y \G(\mat S-\Ai_0\mat Y)\Trans\\
&\qq + \mat W\mat Y\G\mat Y\Trans(\mat S - \Ai_0\mat Y)\G\mat Y\Trans \mat W \numberthis \label{eq:sym_mean}
\end{align*}
and posterior covariance
\begin{align}
\label{eq:sym_var}
\mat V_m = & (\mat W - \mat W\mat Y\G\mat Y\Trans \mat W)\sk (\mat W - \mat W\mat Y\G\mat Y\Trans \mat W)
\end{align}
where $\G\ce (\mat Y\Trans \mat W\mat Y)^{-1}$.
\end{proposition}
\begin{remark} \label{remark:posterior_symmetric}
Since $A_0^{-1}$ is symmetric and the symmetric prior places mass only on symmetric matrices, the posterior mean $A_m^{-1}$ is also symmetric. 
\end{remark}

\section{Proofs} \label{sec:proofs}
\subsection{Proposition \ref{prop:matrix_prior_equivalent_solution_prior}}
\begin{proof}[Proof of Proposition \ref{prop:matrix_prior_equivalent_solution_prior}]

	Let $H = A^{-1}$ and let $A_0^{-1} = H_0$.
	First note that by right-multiplying the information in \cref{eq:mbi_left_info} by $H$:
	\begin{align*}
		Y_m^\top H &= S_m^\top \\
		\implies \vecm{Y_m^\top H} &= \vecm{S_m^\top} \\
		\implies (Y_m \otimes I) \vecm{H} &= \vecm{S_m^\top} \quad\text{(from K1)}
	\end{align*}
	Now the implied posterior on $\vecm{H}$ can be computed using the standard laws of Gaussian conditioning:
	\begin{align*}
		\vecm{H} &\sim \mathcal{N}(\vecm{H_0}, \Sigma \otimes W) \\
		\implies \vecm{H} | \mathcal{D} &\sim \mathcal{N}(\vecm{H_m}, \Omega_m) .
	\end{align*}
	Let $\Omega_0 = \Sigma_0 \otimes W_0$ and let $P = Y_m^\top \otimes I$. 
	Then
	\begin{align*}
		\vecm{H_m} &= \vecm{H_0} + [P \Omega_0]^\top 
			[P \Omega_0 P^\top]^{-1}
			(\vecm{S_m^\top} - \vecm{Y_m^\top H_0}) \\
		\Omega_m &= \Omega_0 - [P \Omega_0]^\top [P \Omega_0 P^\top]^{-1} (P \Omega_0)
	\end{align*}
	Now note that
	\begin{align*}
		P \Omega_0 &= (Y_m^\top \otimes I)(\Sigma_0 \otimes W) \\
		&= (Y_m^\top \Sigma_0) \otimes W \\
		\implies
		(P \Omega_0)^\top &= (\Sigma_0 Y_m) \otimes W\\
	\end{align*}
	where the second line uses \cref{eq:K2} and the third uses \cref{eq:K4}.
	Thus
	\begin{align*}
		P \Omega_0 P^\top &= (Y_m^\top \otimes I) (\Sigma_0 \otimes W_0) (Y_m^\top \otimes I)^\top &\\
		&= (Y_m^\top \Sigma_0 Y_m)\otimes W_0 \\
		\implies (P \Omega_0 P^\top)^{-1} &= (Y_m^\top \Sigma_0 Y_m)^{-1} \otimes W_0^{-1}
	\end{align*}
	where the second line is again using \cref{eq:K2} and \cref{eq:K4}, while the third line uses \cref{eq:K3}.
	We conclude that
	\begin{align*}
		&(P \Omega_0)^\top (P \Omega_0 P^\top)^{-1} \\
		&\quad= [(\Sigma_0 Y_m) \otimes W] [(Y_m^\top \Sigma_0 Y_m)^{-1} \otimes W^{-1}] \\
		&\quad= (\Sigma_0 Y_m (Y_m^\top \Sigma_0 Y_m)^{-1}) \otimes I \\
		\implies &(P \Omega_0)^\top (P \Omega_0 P^\top)^{-1} (P\Omega_0) \\
		&\quad= (\Sigma_0 Y_m (Y_m^\top \Sigma_0 Y_m)^{-1} Y_m^\top \Sigma_0 ) \otimes W_0 .
	\end{align*}
	From these expressions it is straightforward to simplify the expressions for $\vecm{H_m}$:
	\begin{align*}
		\vecm{H_m} &= \vecm{H_0} + (\Sigma_0 Y_m (Y_m^\top \Sigma_0 Y_m)^{-1} \otimes I) (\vecm{S_m^\top} - \vecm{Y_m^\top H_0}) \\
		&= \textup{vec}\left(H_0 + \Sigma_0 Y^m (Y_m^\top \Sigma_0 Y_m)^{-1} (S_m^\top - Y_m^\top H_0) \right)
	\end{align*}
	where the last line follows from K1. For $\Omega_m$:
	\begin{align*}
		\Omega_m &= \Sigma_0 \otimes W - (\Sigma_0 Y_m (Y_m^\top \Sigma_0 Y_m)^{-1} Y_m^\top \Sigma_0 ) \otimes W_0 \\
		&= (\Sigma_0 - \Sigma_0 Y_m (Y_m^\top \Sigma_0 Y_m)^{-1} Y_m^\top \Sigma_0) \otimes W_0
	\end{align*}
	where the last line is from application of K5.

	It remains to project the posterior into $\reals^d$ by performing the matrix-vector product $H\bm{b}$.
	\begin{equation*}
		\bm{x} = \vecm{H\bm{b}} = (I \otimes \bm{b}^\top) H. \quad\text{(from K1)}
	\end{equation*}
	Thus, the implied posterior is $\bm{x} \sim \mathcal{N}(\bar{\bm{x}}_m, \bar{\Sigma}_m)$, with
	\begin{align*}
		\bar{x}_m &= (I \otimes \bm{b}^\top) \textup{vec}\big( \\
		&\qquad H_0 + \Sigma_0 Y_m (Y_m^\top \Sigma_0 Y_m)^{-1}(S_m^\top - Y_m^\top H_0)\big) \\
		&= \textup{vec}\left(H_0 \bm{b} + \Sigma_0 Y_m (Y_m^\top \Sigma_0 Y_m)^{-1}(S_m^\top\bm{b} - Y_m^\top H_0 \bm{b})\right) \\
		&= \bm{x}_0 + \Sigma_0 A^\top S_m (S_m^\top A \Sigma_0 A^\top S_m)^{-1}S_m^\top(\bm{b} - A \bm{x}_0)
	\end{align*}
	where in the last line we have used that $H_0 \bm{b} = \bm{x}_0$ and that $Y_m = A^\top S_m$.
	Furthermore
	\begin{align*}
		\bar{\Sigma}_m &= (I \otimes \bm{b}^\top ) \\
		&\qquad\cdot\left[
			(\Sigma_0 - \Sigma_0 Y_m (Y_m^\top \Sigma_0 Y_m)^{-1} Y_m^\top \Sigma_0) \otimes W_0
		\right] \\
		&\qquad\cdot(I \otimes \bm{b}^\top )^\top \\
		&= (\Sigma_0 - \Sigma_0 Y_m (Y_m^\top \Sigma_0 Y_m)^{-1} Y_m^\top \Sigma_0) \times \bm{b}^\top W_0 \bm{b} \\
		&= \Sigma_0 - \Sigma_0 A^\top S_m (S_m^\top A \Sigma_0 A^\top S_m)^{-1} S_m^\top A \Sigma_0
	\end{align*}
	where in the second line we have used K2 and the fact that $\bm{b}^\top W_0 \bm{b}$ is a scalar, while in the third line we have used that $\bm{b}^\top W_0 \bm{b} = 1$ and that $Y_m = A^\top S_m$.

	Note that $\bm{x}_m = \bar{\bm{x}}_m$ and $\Sigma_m = \bar{\Sigma}_m$, as defined in \citet{Cockayne:2018}.
	Thus, the proof is complete.
\qed\end{proof}
\subsection{Theorem \ref{thm:cg_right_multiplied}}
\begin{proof}[Proof of Theorem \ref{thm:cg_right_multiplied}.]
Denote by $\vec x_i^{CG}$ the conjugate gradient estimate in iteration $i$ and with $\vec p_i$ the search direction in that iteration.
From one iteration to the next, the update to the solution can be written as \citep[p.~108]{nocedal1999numerical}
\begin{align}
\vec x_{i+1}^{CG}=\vec x_i^{CG} + \frac{\vec r_i\Trans\vec p_i}{\vec p_i\Trans \A\vec p_i}\vec p_i.
\end{align}
Comparing this update to lines 7 to 10 in Algorithm \ref{alg:spd_proj} it is sufficient to show that $\vec d_i\propto \vec p_i$ which follows from Lemma \ref{lem:cg_right__d_is_scaled_s}.
\qed\end{proof}

\begin{lemma}
\label{lem:cg_right__d_is_scaled_s}
Assume that CG does not terminate before $d$ iterations.
Using the prior of Theorem \ref{thm:cg_right_multiplied} in Algorithm \ref{alg:spd_proj}, the directions $\vec d_i$ are scaled conjugate gradients search directions, i.e.~
\begin{align*}
\vec d_i & = \gamma_i\vec p_i^{CG}
\end{align*}
where $\vec p_{i}^{CG}$ is the CG search direction in iteration $i$ and $\gamma_i\in \Re\setminus\{0\}$.
\end{lemma}
\begin{proof}

The proof proceeds by induction. 
Throughout we will suppress the superscript $CG$ on the CG search directions, i.e. $\vec p_i^\text{CG} = \vec p_i$.
For $i=1$, $\Ai_{i-1}=\alpha I$ by assumption and therefore $\vec d_i=\alpha \vec r_0$ which is the first CG search direction scaled by $\gamma_1=\alpha\neq 0$.

For the inductive step, suppose that the search directions $\vec s_1, ..., \vec s_{i-1}$ are scaled CG directions and that the vectors $\vec x_1, \dots, \vec x_{i-1}$ are the same as the first $i-1$ solution estimates produced by CG.
We will prove that $\vec s_i$ is the $i$\textsuperscript{th} CG search direction, and that $\vec x_i$ is the $i$\textsuperscript{th} solution estimate from CG.
Lemma \ref{lem:cg_right__d_in_span_S} states that $\vec d_i$ can be written as
\begin{equation}
	\vec d_i = \Ai_{i-1}\vec r_{i-1}=\sum_{j<i} \nu_j \vec s_j + \nu_i \vec r_{i-1}. \label{eq:LA:proof_thm4_sum}
\end{equation}
where $\nu_j \in \reals, j=1,\dots,i$.
Under the prior, the posterior mean $\Ai_i$ is always symmetric as stated in Remark~\ref{remark:posterior_symmetric}.
This allows application of Lemma \ref{lem:cg_right__d_is_A_conj}, so that $\{\vec s_1, \dots, \vec s_{i-1}, \vec d_i \}$ is an $A$-conjugate set. 
Thus we have, for $\ell<i$:
\begin{align*}
	0 = \vec s_\ell \Trans \A \vec d_i &= \nu_\ell \vec s_\ell\Trans \A \vec s_\ell + \nu_i \vec s_\ell\Trans \A \vec r_{i-1}\\
	&= \nu_\ell \vec s_\ell\Trans \A \vec s_\ell + \nu_i \vec y_\ell \Trans \vec r_{i-1}.\numberthis\label{eq:ortho_statement}
\end{align*}

Now note that
\begin{equation*}
	\vec y_\ell \Trans \vec r_{i-1} = (\vec r_\ell - \vec r_{\ell-1})\Trans \vec r_{i-1} .
\end{equation*}
This follows from Line 10 of Algorithm~\ref{alg:spd_proj}, from which it is clear that $\vec y_\ell=\vec r_\ell - \vec r_{\ell-1}$.
Recall that the CG residuals $\vec r_j$ are orthogonal \cite[p.~109]{nocedal1999numerical}, and that from the inductive assumption, Algorithm~\ref{alg:spd_proj} is equivalent to CG up to iteration $i-1$).
Thus, for $\ell < i-1$ we have that
\begin{align*}
  \vec y_\ell \Trans \vec r_{i-1} &= 0 \\
  \implies \vec s_\ell \A \vec d_i &= \nu_\ell \vec s_\ell\Trans \A \vec s_\ell = 0\qq\forall\,\ell < i-1
\end{align*}
where the second line is from application of the first line in Eq.~\eqref{eq:ortho_statement}.
However, $\A$ is positive definite and by assumption the algorithm has not converged, so $\vec d_\ell\neq \vec 0$.
Furthermore clearly $\vec s_\ell\Trans \A \vec s_\ell \neq 0$.
Hence we must have that
\begin{equation*}
	\nu_\ell = 0 \qq\forall\,j<i-1.
\end{equation*}
Equation~\eqref{eq:LA:proof_thm4_sum} thus simplifies to
\begin{equation}
	\vec d_i = \nu_{i-1} \vec s_{i-1} + \nu_i \vec r_{i-1} = \nu_{i-1} \alpha_{i-1} \vec d_{i-1} + \nu_i \vec r_{i-1}. \label{eq:simplified_ortho_statmenet}
\end{equation}
Now, again by Lemma \ref{lem:cg_right__d_is_A_conj}, $\vec d_i$ must be conjugate to $\vec s_{i-1}$ which implies $\nu_i\neq 0$.
Pre-multiplying Eq.~\eqref{eq:simplified_ortho_statmenet} by $\vec s_{i-1}\Trans \A$ gives 
\begin{align*}
	0 &= \nu_{i-1} \alpha_{i-1} \vec s_{i-1}\Trans \A \vec d_{i-1} + \nu_i \vec s_{i-1}\Trans \A \vec r_{i-1}\\
	\implies \nu_{i-1}\alpha_{i-1} &= -\nu_i \frac{\vec s_{i-1}\Trans\A\vec r_{i-1}}{\vec s_{i-1}\Trans \A\vec d_{i-1}} .
\end{align*}
Thus, $\vec d_i$ can be written as
\begin{align*}
 \vec d_i & = \nu_i \left(\vec r_{i-1} - \frac{\vec s_{i-1}\Trans \A \vec r_{i-1}}{\vec s_{i-1}\Trans \A \vec d_{i-1}}\vec d_{i-1} \right)\\
 & = \nu_i \left(\vec r_{i-1} - \frac{\vec p_{i-1}\Trans \A \vec r_{i-1}}{\vec p_{i-1}\Trans \A \vec p_{i-1}}\vec p_{i-1} \right)	
\numberthis\label{eq:LA:proof_Thm4_inductionclose}
\end{align*}
where the second line again applies the inductive assumption, that $\vec d_{i-1}$ and $\vec s_{i-1}$ are proportional to the CG search direction $\vec p_{i-1}$, noting that the proportionality constants on numerator and denominator cancel.
The term inside the brackets is precisely the $i$\textsuperscript{th} CG search direction.
This completes the result.
\qed\end{proof}

\begin{lemma}
\label{lem:cg_right__d_is_A_conj}
If the belief over $\Ai_m$ is symmetric for all $m=0,\dots,d$ and $\A$ is symmetric and positive definite, then Algorithm \ref{alg:spd_proj} produces $\A$-conjugate directions.
\end{lemma}
\begin{proof}
The proof is by induction.
Note that the case $i = 1$ is irrelevant since a set consisting of one element is trivially $A$-conjugate.
On many occasions the proof relies on the consistency of the MBI belief, i.e.~$\Ai_i \vec z_k=\vec d_k$ for $k\leq i$ and by symmetry $\vec z_k\Trans \Ai_i=\vec d_k\Trans$.
Thus, for the base case $i=2$ we have:
\begin{align*}
	\vec d_1\Trans \A \vec d_2 &= -\vec d_1\Trans \A (\Ai_1 \vec r_1)\\
	&= -\vec d_1\Trans \A (\Ai_1(\vec y_1 + \vec r_0))\\
	&= -\vec d_1\Trans \A (\vec s_1 + \Ai_1 \vec r_0)
\end{align*}
where the second line is by Line 10 of Algorithm~\ref{alg:spd_proj}.
Now recall that $\alpha_1=-\nicefrac{\vec d_1\Trans \vec r_0}{\vec d_1\Trans \A \vec d_1}$ to give:
\begin{align*}
\vec d_1\Trans \A \vec d_2
	&= -\alpha_1 \vec d_1\Trans \A \vec d_1 - \vec d_1\Trans \A \Ai_1 \vec r_0\\
	&= \vec d_1\Trans \vec r_0 - \vec d_1\Trans \A \Ai_1 \vec r_0 \\
	&= \vec d_1\Trans \vec r_0 - \vec z_1\Trans \Ai_1 \vec r_0 \\
	&= \vec d_1\Trans \vec r_0 - \vec d_1 \Trans \vec r_0  \numberthis\label{eq:LA:symmetry_used}\\
	&= 0.
\end{align*}
Here, the symmetry of the estimator $\Ai_i$ is used in Eq.~\eqref{eq:LA:symmetry_used}.
For the inductive step, assume $\{\vec d_0,\dots,\vec d_{i-1}\}$ are pairwise $\A$-conjugate.
For any $k<i$ we have:
\begin{align*}
	\vec d_k \Trans \A \vec d_i &= -\vec d_k\Trans \A (\Ai_i \vec r_i)\\
	&= -\vec d_k\Trans \A \Ai_i\left(\sum_{j\leq i} \vec y_j + \vec r_0\right)
\end{align*}
where the second line follows from the fact that $\vec r_i = \vec r_{i-1} + \vec y_i$. 
Thus, we have:
\begin{align*}
	\vec d_k \Trans \A \vec d_i&= -\vec d_k\Trans \A \left(\sum_{j\leq i} \vec s_j + \Ai_i \vec r_0\right)\\
	&= -\vec d_k\Trans \A \left(\sum_{j\leq i} \alpha_j \vec d_j + \Ai_i \vec r_0\right).
\end{align*}
Now, applying the conjugacy from the inductive assumption:
\begin{align*}
	\vec d_k \Trans \A \vec d_i&= -\alpha_k \vec d_k\Trans \A \vec d_k - \vec d_k\Trans \A (\Ai_i \vec r_0)\\
	&= \vec d_k\Trans \vec r_{k-1} - \vec d_k\Trans \vec r_0 \\
	&= \vec d_k\Trans \left(\sum_{j<k} \vec y_j + \vec r_0\right) - \vec d_k\Trans \vec r_0 = 0\\
	&= \sum_{j<k} \alpha_j \vec d_k\Trans \A \vec d_j = 0.
\end{align*}
where the second line rearranges line 6 of the algorithm to obtain $\alpha_i \vec d_i\Trans \vec z_i = -\vec d_i\Trans \vec r_{i-1}$.
The third line again uses that $\vec r_i = \vec r_{i-1} + \vec y_i$, while the fourth line is from the assumed conjugacy. 
\qed\end{proof}

\begin{lemma}
\label{lem:cg_right__d_in_span_S} 
Under the prior in Theorem \ref{thm:cg_right_multiplied} and given scaled CG search directions $\vec p_1, ..., \vec p_i$, it holds that \allowbreak
$\Ai_{i}\vec r_i\in \spa\{\vec p_1, ..., \vec p_{i}, \vec r_{i}\}.$
\end{lemma}
\begin{proof}
Recall first that under the prior in Theorem~\ref{thm:cg_right_multiplied}, $\Ai_0 = \alpha I$.
Then by inspection of Eq.~\eqref{eq:sym_mean} we have $\Ai_i \vec r_i \in \mathcal{S}$ where
\begin{equation*}
\mathcal{S} = \spa\{\vec r_{i}, \allowbreak \vec p_1, ..., \vec p_{i}, \allowbreak \vec y_1, ..., \vec y_i, \mat W \vec y_1, ..., \mat W\vec y_i\}
\end{equation*}
By choice of $\mat W=\beta \mat I+\gamma \Ai$, $\mathcal{S}=\spa\{\vec r_{i}, \vec p_1, ..., \vec p_{i}, \allowbreak \vec y_1, ..., \vec y_i\}$.
From line 10 of Algorithm \ref{alg:spd_proj} $\vec y_i = \vec r_i - \vec r_{i-1}$ and therefore $\mathcal{S}=\spa\{\vec r_1, ..., \vec r_{i}, \allowbreak \vec p_1, ..., \vec p_{i}\}$.
By Theorem 5.3 in \cite[p.~109]{nocedal1999numerical} the span of the conjugate gradients residuals and search directions are equivalent.
Therefore $\mathcal{S}\subseteq \{\vec r_i, \vec p_1, ..., \vec p_i\}$.
\qed\end{proof}
\end{appendices}

\bibliographystyle{abbrvnat}
\bibliography{refs,bibfile}
\end{document}

%% file: simulation/notebooks/figures/bgmres_gaussian_uq__seed_12345__eq_type_spd_reverse__iteration_1.tikz
\begin{tikzpicture}

\begin{axis}[
xlabel={$z$},
ylabel={$p(z)$},
xmin=0, xmax=140,
ymin=-0.01, ymax=0.3,
width=\figwidth,
height=\figheight,
xtick={0,20,40,60,80,100,120,140},
xticklabels={$0$,$20$,$40$,$60$,$80$,$100$,$120$,$140$},
ytick={-0.05,0,0.05,0.1,0.15,0.2,0.25,0.3},
yticklabels={,$0.00$,$0.05$,$0.10$,$0.15$,$0.20$,$0.25$,$0.30$},
tick align=outside,
tick pos=left,
x grid style={white!69.01960784313725!black},
y grid style={white!69.01960784313725!black},
legend cell align={left},
legend style={draw=white!80.0!black},
legend entries={{BayesGMRES},{$\chi^2_{99}$}}
]
\addlegendimage{no markers, blue!50.98039215686274!black}
\addlegendimage{no markers, black}
\addplot [semithick, blue!50.98039215686274!black]
table {%
0 4.83897571318291e-56
0.280561122244489 1.42807276758821e-55
0.561122244488978 4.19370272046987e-55
0.841683366733467 1.22544931624326e-54
1.12224448897796 3.56322687287176e-54
1.40280561122244 1.03096044967981e-53
1.68336673346693 2.9681848466467e-53
1.96392785571142 8.50335461721026e-53
2.24448897795591 2.42404147926798e-52
2.5250501002004 6.87606877536496e-52
2.80561122244489 1.94084469817193e-51
3.08617234468938 5.45119562910294e-51
3.36673346693387 1.52350245801538e-50
3.64729458917836 4.23686824979062e-50
3.92785571142285 1.17245765978052e-49
4.20841683366733 3.22849253146164e-49
4.48897795591182 8.84611972290793e-49
4.76953907815631 2.41188306661147e-48
5.0501002004008 6.54349979313438e-48
5.33066132264529 1.76650286688403e-47
5.61122244488978 4.74535829589049e-47
5.89178356713427 1.26845212639903e-46
6.17234468937876 3.37387945426472e-46
6.45290581162325 8.92967085959336e-46
6.73346693386774 2.35175323711116e-45
7.01402805611222 6.16308927561184e-45
7.29458917835671 1.60714690513248e-44
7.5751503006012 4.17025971911876e-44
7.85571142284569 1.07676527024279e-43
8.13627254509018 2.76649178668001e-43
8.41683366733467 7.07274728147774e-43
8.69739478957916 1.79927400405165e-42
8.97795591182365 4.55466961595305e-42
9.25851703406814 1.14727324082254e-41
9.53907815631263 2.87559194440454e-41
9.81963927855711 7.17196336456333e-41
10.1002004008016 1.7799150558658e-40
10.3807615230461 4.39552648343642e-40
10.6613226452906 1.08012230806902e-39
10.9418837675351 2.64110325796044e-39
11.2224448897796 6.42611120335686e-39
11.503006012024 1.55582771371316e-38
11.7835671342685 3.74822092369551e-38
12.064128256513 8.98543835983541e-38
12.3446893787575 2.14340267880991e-37
12.625250501002 5.08766703057733e-37
12.9058116232465 1.20166663733124e-36
13.186372745491 2.82422789924821e-36
13.4669338677355 6.6048945202935e-36
13.74749498998 1.53703036152056e-35
14.0280561122244 3.55917563120228e-35
14.3086172344689 8.20099987606028e-35
14.5891783567134 1.8803322941668e-34
14.8697394789579 4.28995580987678e-34
15.1503006012024 9.73915878686502e-34
15.4308617234469 2.20009021475525e-33
15.7114228456914 4.94549705518028e-33
15.9919839679359 1.10619023811152e-32
16.2725450901804 2.46206838801424e-32
16.5531062124249 5.4528157031992e-32
16.8336673346693 1.20168860935169e-31
17.1142284569138 2.63519943123558e-31
17.3947895791583 5.75023310684564e-31
17.6753507014028 1.24855544634685e-30
17.9559118236473 2.69761915642576e-30
18.2364729458918 5.79967773695686e-30
18.5170340681363 1.24073064943306e-29
18.7975951903808 2.64120178957582e-29
19.0781563126253 5.59469068127551e-29
19.3587174348697 1.17923674739776e-28
19.6392785571142 2.47329748585688e-28
19.9198396793587 5.16181145700241e-28
20.2004008016032 1.07195943086796e-27
20.4809619238477 2.21515945868741e-27
20.7615230460922 4.5549338489737e-27
21.0420841683367 9.31986433322614e-27
21.3226452905812 1.89752517172209e-26
21.6032064128257 3.84428823313222e-26
21.8837675350701 7.74987651648058e-26
22.1643286573146 1.55461935241673e-25
22.4448897795591 3.1031574107936e-25
22.7254509018036 6.1635935988999e-25
23.0060120240481 1.21818890853757e-24
23.2865731462926 2.39577344294071e-24
23.5671342685371 4.68842870149346e-24
23.8476953907816 9.12975984493828e-24
24.1282565130261 1.76905715551512e-23
24.4088176352705 3.41094607384873e-23
24.689378757515 6.54422655433274e-23
24.9699398797595 1.24937349267359e-22
25.250501002004 2.373431734956e-22
25.5310621242485 4.48654166276018e-22
25.811623246493 8.43912024610408e-22
26.0921843687375 1.57954921837531e-21
26.372745490982 2.94184425842259e-21
26.6533066132265 5.45201105302167e-21
26.9338677354709 1.00541260013688e-20
27.2144288577154 1.84494097442738e-20
27.4949899799599 3.36876872948044e-20
27.7755511022044 6.12083255404993e-20
28.0561122244489 1.10662479015084e-19
28.3366733466934 1.9908608210044e-19
28.6172344689379 3.5639535082594e-19
28.8977955911824 6.34853914813785e-19
29.1783567134269 1.12529459881935e-18
29.4589178356713 1.98476616146665e-18
29.7394789579158 3.48339850718304e-18
30.0200400801603 6.08341859248257e-18
30.3006012024048 1.05716552046394e-17
30.5811623246493 1.82805415204483e-17
30.8617234468938 3.1454727659227e-17
31.1422845691383 5.38559436524075e-17
31.4228456913828 9.17555245603204e-17
31.7034068136273 1.55554189951502e-16
31.9839679358717 2.62411070714232e-16
32.2645290581162 4.40487495714917e-16
32.5450901803607 7.35759796652293e-16
32.8256513026052 1.22289624219518e-15
33.1062124248497 2.02252691477251e-15
33.3867735470942 3.32851263370684e-15
33.6673346693387 5.45076393541797e-15
33.9478957915832 8.88210433457181e-15
34.2284569138277 1.44020983848835e-14
34.5090180360721 2.32373801491761e-14
34.7895791583166 3.73078416836703e-14
35.0701402805611 5.9602538074963e-14
35.3507014028056 9.4750421361552e-14
35.6312625250501 1.49881980435189e-13
35.9118236472946 2.35922678876642e-13
36.1923847695391 3.69523516367852e-13
36.4729458917836 5.75926065547456e-13
36.7535070140281 8.93189889171599e-13
37.0340681362725 1.37839419479705e-12
37.314629258517 2.1166829976936e-12
37.5951903807615 3.23438076854031e-12
37.875751503006 4.91789911137011e-12
38.1563126252505 7.44083118004483e-12
38.436873747495 1.12025501168237e-11
38.7174348697395 1.67828691690607e-11
38.997995991984 2.50189843114376e-11
39.2785571142285 3.71131313979556e-11
39.5591182364729 5.47823066374981e-11
39.8396793587174 8.04652158362669e-11
40.1202404809619 1.17606570984296e-10
40.4008016032064 1.71045226545378e-10
40.6813627254509 2.47540778419633e-10
40.9619238476954 3.56483527113066e-10
41.2424849699399 5.10845579726573e-10
41.5230460921844 7.28446831832418e-10
41.8036072144289 1.03362917320291e-09
42.0841683366734 1.45945607327195e-09
42.3647294589178 2.05058330223455e-09
42.6452905811623 2.8669804842705e-09
42.9258517034068 3.98872294126701e-09
43.2064128256513 5.52211921917023e-09
43.4869739478958 7.60749194209429e-09
43.7675350701403 1.04289880679957e-08
44.0480961923848 1.42268586878174e-08
44.3286573146293 1.93127171751399e-08
44.6092184368737 2.60883560860011e-08
44.8897795591182 3.50687753097297e-08
45.1703406813627 4.69101432983122e-08
45.4509018036072 6.2443475511137e-08
45.7314629258517 8.27148642386392e-08
46.0120240480962 1.09033125411537e-07
46.2925851703407 1.43025733579778e-07
46.5731462925852 1.86703887834265e-07
46.8537074148297 2.42537480326307e-07
47.1342685370741 3.13540616265836e-07
47.4148296593186 4.03368151165884e-07
47.6953907815631 5.16423459860896e-07
47.9759519038076 6.57977326328001e-07
48.2565130260521 8.342974399778e-07
48.5370741482966 1.05278750272061e-06
48.8176352705411 1.32213439607068e-06
49.0981963927856 1.65246123852775e-06
49.3787575150301 2.05548349638268e-06
49.6593186372746 2.54466461866312e-06
49.939879759519 3.13536697900511e-06
50.2204408817635 3.84499326063831e-06
50.501002004008 4.69311285857337e-06
50.7815631262525 5.70156744375314e-06
51.062124248497 6.89454958883794e-06
51.3426853707415 8.29864834653944e-06
51.623246492986 9.94285594850965e-06
51.9038076152305 1.18585303989256e-05
52.1843687374749 1.40793097048014e-05
52.4649298597194 1.66409748326795e-05
52.7454909819639 1.95812602077069e-05
53.0260521042084 2.29396126525729e-05
53.3066132264529 2.67569020511582e-05
53.5871743486974 3.10750896386018e-05
53.8677354709419 3.59368625619897e-05
54.1482965931864 4.13852460945388e-05
54.4288577154309 4.74632074695934e-05
54.7094188376754 5.42132675617024e-05
54.9899799599198 6.16771384052738e-05
55.2705410821643 6.98954056336189e-05
55.5511022044088 7.89072751827134e-05
55.8316633266533 8.87504029003089e-05
56.1122244488978 9.94608239355013e-05
56.3927855711423 0.000111072995909078
56.6733466933868 0.000123619965892601
56.9539078156313 0.000137133666232036
57.2344689378758 0.000151645338387601
57.5150300601202 0.000167186077442509
57.7955911823647 0.000183787483040949
58.0761523046092 0.000201482395585805
58.3567134268537 0.000220305689935118
58.6372745490982 0.000240295092979877
58.9178356713427 0.000261491986760661
59.1983967935872 0.000283942155558095
59.4789579158317 0.0003076964339941
59.7595190380761 0.000332811213856128
60.0400801603206 0.000359348770259913
60.3206412825651 0.00038737737294147
60.6012024048096 0.000416971155838196
60.8817635270541 0.000448209727477323
61.1623246492986 0.000481177515709928
61.4428857715431 0.00051596285256951
61.7234468937876 0.000552656817960433
62.0040080160321 0.000591351873887703
62.2845691382766 0.00063214033337944
62.565130260521 0.000675112719473974
62.8456913827655 0.000720356079019464
63.12625250501 0.00076795232300329
63.4068136272545 0.000817976669223529
63.687374749499 0.000870496263991618
63.9679358717435 0.000925569057013586
64.248496993988 0.000983242997596371
64.5290581162325 0.00104355561098943
64.809619238477 0.00110653400128642
65.0901803607214 0.00117219531231426
65.3707414829659 0.00124054766089698
65.6513026052104 0.00131159153847482
65.9318637274549 0.00138532165803624
66.2124248496994 0.00146172920447319
66.4929859719439 0.00154080442859413
66.7735470941884 0.0016225395088875
67.0541082164329 0.00170693159141878
67.3346693386774 0.00179398590756594
67.6152304609219 0.00188371886212889
67.8957915831663 0.00197616098102476
68.1763527054108 0.00207135960848385
68.4569138276553 0.00216938124841561
68.7374749498998 0.0022703134532871
69.0180360721443 0.00237426617616983
69.2985971943888 0.00248137251715498
69.5791583166333 0.00259178881359105
69.8597194388778 0.00270569404395657
70.1402805611222 0.0028232885369732
70.4208416833667 0.00294479200009185
70.7014028056112 0.00307044090403127
70.9819639278557 0.00320048528191749
71.2625250501002 0.00333518502209379
71.5430861723447 0.00347480575222355
71.8236472945892 0.00361961442832569
72.1042084168337 0.00376987475536124
72.3847695390782 0.00392584257548949
72.6653306613226 0.00408776136575975
72.9458917835671 0.00425585798848848
73.2264529058116 0.0044303388346486
73.5070140280561 0.00461138649309059
73.7875751503006 0.00479915706622529
74.0681362725451 0.00499377823591677
74.3486973947896 0.00519534816185588
74.6292585170341 0.00540393526884211
74.9098196392786 0.00561957894957443
75.190380761523 0.00584229117631068
75.4709418837675 0.00607205897887351
75.751503006012 0.00630884770896081
76.0320641282565 0.00655260497279021
76.312625250501 0.00680326507723154
76.5931863727455 0.00706075380041062
76.87374749499 0.00732499326811208
77.1543086172345 0.00759590669405406
77.434869739479 0.00787342272713224
77.7154308617234 0.00815747914379027
77.9959919839679 0.00844802563028058
78.2765531062124 0.00874502541886388
78.5571142284569 0.00904845557458471
78.8376753507014 0.0093583057751511
79.1182364729459 0.00967457548491467
79.3987975951904 0.00999726949348482
79.6793587174349 0.0103263918678026
79.9599198396794 0.0106619384504568
80.2404809619239 0.011003888122873
80.5210420841683 0.011352193135414
80.8016032064128 0.0117067688827303
81.0821643286573 0.0120674835670926
81.3627254509018 0.0124341482403118
81.6432865731463 0.0128065077420367
81.9238476953908 0.0131842330553132
82.2044088176353 0.0135669155769279
82.4849699398798 0.0139540637491484
82.7655310621243 0.0143451024213946
83.0460921843687 0.0147393752070655
83.3266533066132 0.0151361499757363
83.6072144288577 0.0155346274792699
83.8877755511022 0.0159339529584288
84.1683366733467 0.0163332304217427
84.4488977955912 0.0167315391388266
84.7294589178357 0.0171279517544326
85.0100200400802 0.0175215533154583
85.2905811623247 0.0179114604184322
85.5711422845691 0.0182968396359677
85.8517034068136 0.0186769243720538
86.1322645290581 0.0190510293305665
86.4128256513026 0.0194185618595387
86.6933867735471 0.0197790295536137
86.9739478957916 0.0201320436544271
87.2545090180361 0.020477317976819
87.5350701402806 0.0208146632991817
87.815631262525 0.0211439773786888
88.0961923847695 0.0214652309753305
88.376753507014 0.0217784504808078
88.6573146292585 0.0220836979377726
88.937875751503 0.0223810493908275
89.2184368737475 0.0226705726237789
89.498997995992 0.0229523054005546
89.7795591182365 0.0232262353351672
90.060120240481 0.023492282467198
90.3406813627254 0.0237502855145931
90.6212424849699 0.0239999926192665
90.9018036072144 0.0242410572001125
91.1823647294589 0.0244730392921227
91.4629258517034 0.0246954124909643
91.7434869739479 0.0249075763526366
92.0240480961924 0.0251088738313878
92.3046092184369 0.0252986130895955
92.5851703406814 0.025476092793614
92.8657314629259 0.0256406298309511
93.1462925851703 0.0257915882556421
93.4268537074148 0.0259284081967812
93.7074148296593 0.026050633453234
93.9879759519038 0.026157936545824
94.2685370741483 0.0262501401038318
94.5490981963928 0.0263272336195777
94.8296593186373 0.0263893848047571
95.1102204408818 0.0264369450145778
95.3907815631263 0.0264704484587205
95.6713426853707 0.0264906051791159
95.9519038076152 0.0264982880308962
96.2324649298597 0.0264945141427332
96.5130260521042 0.0264804215455157
96.7935871743487 0.0264572418351822
97.0741482965932 0.0264262698699415
97.3547094188377 0.0263888315899829
97.6352705410822 0.0263462510874865
97.9158316633267 0.0262998180471456
98.1963927855711 0.0262507566255672
98.4769539078156 0.0262001967467733
98.7575150300601 0.0261491486670088
99.0380761523046 0.0260984815125937
99.3186372745491 0.026048906327565
99.5991983967936 0.0260009639913342
99.8797595190381 0.0259550181881387
100.160320641283 0.0259112534365909
100.440881763527 0.0258696780250521
100.721442885772 0.0258301315516508
101.002004008016 0.0257922966401294
101.282565130261 0.0257557142967167
101.563126252505 0.0257198022902351
101.84368737475 0.025683875878013
102.124248496994 0.0256471701635413
102.404809619238 0.025608863357231
102.685370741483 0.0255681002177992
102.965931863727 0.0255240149771745
103.246492985972 0.0254757530947203
103.527054108216 0.0254224912453151
103.807615230461 0.0253634550186804
104.088176352705 0.0252979338925372
104.36873747495 0.0252252931378376
104.649298597194 0.0251449824184412
104.929859719439 0.0250565409579316
105.210420841683 0.0249595992602475
105.490981963928 0.0248538774855788
105.771543086172 0.0247391806953487
106.052104208417 0.0246153912866287
106.332665330661 0.0244824590333771
106.613226452906 0.0243403892357861
106.89378757515 0.0241892295462666
107.174348697395 0.0240290560879907
107.454909819639 0.0238599595068574
107.735470941884 0.0236820315983559
108.016032064128 0.0234953531261559
108.296593186373 0.0232999833994899
108.577154308617 0.0230959521027865
108.857715430862 0.0228832537760242
109.138276553106 0.0226618452314484
109.418837675351 0.022431646066162
109.699398797595 0.0221925422959683
109.97995991984 0.021944392999546
110.260521042084 0.0216870397296437
110.541082164329 0.021420318325457
110.821643286573 0.0211440726532826
111.102204408818 0.0208581697157697
111.382765531062 0.0205625155074991
111.663326653307 0.0202570709588851
111.943887775551 0.0199418673028919
112.224448897796 0.0196170202197227
112.50501002004 0.0192827421620892
112.785571142285 0.0189393523352117
113.066132264529 0.018587283897575
113.346693386774 0.0182270880560285
113.627254509018 0.0178594348468843
113.907815631263 0.0174851105177214
114.188376753507 0.0171050115472223
114.468937875752 0.0167201354573874
114.749498997996 0.0163315686793504
115.03006012024 0.0159404718269562
115.310621242485 0.0155480628083914
115.591182364729 0.0151555982636026
115.871743486974 0.0147643538531089
116.152304609218 0.0143756039421739
116.432865731463 0.0139906012240399
116.713426853707 0.0136105568086119
116.993987975952 0.013236621270694
117.274549098196 0.0128698671070329
117.555110220441 0.0125112729965686
117.835671342685 0.0121617101959751
118.11623246493 0.0118219313352199
118.396793587174 0.0114925618076641
118.677354709419 0.0111740938781025
118.957915831663 0.0108668835617333
119.238476953908 0.0105711502587418
119.519038076152 0.0102869790641313
119.799599198397 0.0100143256116424
120.080160320641 0.0097530232549568
120.360721442886 0.00950279233973188
120.64128256513 0.00926325127718057
120.921843687375 0.00903392909472331
121.202404809619 0.00881427911250517
121.482965931864 0.00860369337704803
121.763527054108 0.00840151747564561
122.044088176353 0.00820706535777952
122.324649298597 0.00801963380304798
122.605210420842 0.0078385161987435
122.885771543086 0.00766301532378347
123.166332665331 0.00749245487826593
123.446893787575 0.00732618954812629
123.72745490982 0.00716361345045013
124.008016032064 0.00700416686483811
124.288577154309 0.00684734121745992
124.569138276553 0.00669268234457215
124.849699398798 0.00653979211880747
125.130260521042 0.00638832857210743
125.410821643287 0.00623800469168628
125.691382765531 0.0060885860982026
125.971943887776 0.0059398878372103
126.25250501002 0.00579177052536705
126.533066132265 0.00564413609181518
126.813627254509 0.00549692334323121
127.094188376754 0.00535010355943894
127.374749498998 0.00520367629683919
127.655310621242 0.00505766554122981
127.935871743487 0.00491211631210891
128.216432865731 0.00476709177959156
128.496993987976 0.00462267091489853
128.77755511022 0.00447894665807684
129.058116232465 0.00433602455397779
129.338677354709 0.00419402178094623
129.619238476954 0.00405306647713117
129.899799599198 0.00391329725729979
130.180360721443 0.00377486280856108
130.460921843687 0.00363792145608573
130.741482965932 0.00350264059899014
131.022044088176 0.00336919593099122
131.302605210421 0.00323777037899361
131.583166332665 0.00310855271410022
131.86372745491 0.00298173581228326
132.144288577154 0.00285751456482514
132.424849699399 0.00273608346048339
132.705410821643 0.00261763388117496
132.985971943888 0.00250235117005626
133.266533066132 0.00239041154466111
133.547094188377 0.00228197893793643
133.827655310621 0.00217720185647404
134.108216432866 0.00207621034803632
134.38877755511 0.00197911316980995
134.669338677355 0.00188599524499155
134.949899799599 0.0017969154886653
135.230460921844 0.00171190507486399
135.511022044088 0.00163096620560124
135.791583166333 0.00155407142990364
136.072144288577 0.00148116354681991
136.352705410822 0.00141215611137864
136.633266533066 0.00134693454683742
136.913827655311 0.00128535785063632
137.194388777555 0.00122726086556987
137.4749498998 0.00117245707217231
137.755511022044 0.00112074184354266
138.036072144289 0.00107189609021663
138.316633266533 0.00102569021063834
138.597194388778 0.000981888252720334
138.877755511022 0.00094025218432032
139.158316633267 0.000900546165587764
139.438877755511 0.000862540714362413
139.719438877756 0.000826016657373564
140 0.000790768765011961
};
\addplot [very thick, black, dashed]
table {%
0 0
0.280561122244489 2.13563217594779e-104
0.561122244488978 7.38853221950653e-90
0.841683366733467 2.22874900396824e-81
1.12224448897796 2.22160484555478e-75
1.40280561122244 9.68008946752989e-71
1.68336673346693 5.82434015642017e-67
1.96392785571142 8.9380535939634e-64
2.24448897795591 5.04579159769813e-61
2.5250501002004 1.32709977126018e-58
2.80561122244489 1.91081566453694e-56
3.08617234468938 1.68982195919076e-54
3.36673346693387 9.99224511527496e-53
3.64729458917836 4.21407401260255e-51
3.92785571142285 1.33271175642842e-49
4.20841683366733 3.28873960988382e-48
4.48897795591182 6.53882192292012e-47
4.76953907815631 1.07530078895365e-45
5.0501002004008 1.49468852886614e-44
5.33066132264529 1.78838400346898e-43
5.61122244488978 1.87043597314543e-42
5.89178356713427 1.73260984473076e-41
6.17234468937876 1.43761279498859e-40
6.45290581162325 1.07899363497967e-39
6.73346693386774 7.38823984136151e-39
7.01402805611222 4.65010874601071e-38
7.29458917835671 2.70805195167072e-37
7.5751503006012 1.46780750239334e-36
7.85571142284569 7.44334223865811e-36
8.13627254509018 3.54799739535193e-35
8.41683366733467 1.5963869707019e-34
8.69739478957916 6.80575636452154e-34
8.97795591182365 2.7585770699457e-33
9.25851703406814 1.06638538376222e-32
9.53907815631263 3.94268883222022e-32
9.81963927855711 1.39779056132007e-31
10.1002004008016 4.76310511798136e-31
10.3807615230461 1.56343353291078e-30
10.6613226452906 4.95310081458115e-30
10.9418837675351 1.51734067107044e-29
11.2224448897796 4.50231797721632e-29
11.503006012024 1.2960464349429e-28
11.7835671342685 3.62469126781972e-28
12.064128256513 9.86228297099274e-28
12.3446893787575 2.61390078760017e-27
12.625250501002 6.75644326889438e-27
12.9058116232465 1.70507300786746e-26
13.186372745491 4.20543746197376e-26
13.4669338677355 1.01471015867174e-25
13.74749498998 2.39733542732019e-25
14.0280561122244 5.55061440782277e-25
14.3086172344689 1.26045201096327e-24
14.5891783567134 2.80938893386327e-24
14.8697394789579 6.15044165705603e-24
15.1503006012024 1.32342978712152e-23
15.4308617234469 2.80072791952592e-23
15.7114228456914 5.83279504737937e-23
15.9919839679359 1.1960926817357e-22
16.2725450901804 2.41640066870747e-22
16.5531062124249 4.81183526260852e-22
16.8336673346693 9.44932433535786e-22
17.1142284569138 1.83079212897856e-21
17.3947895791583 3.50119182416797e-21
17.6753507014028 6.61169086940461e-21
17.9559118236473 1.23339321775605e-20
18.2364729458918 2.27377528370065e-20
18.5170340681363 4.14388356715599e-20
18.7975951903808 7.46846805821596e-20
19.0781563126253 1.33156628249856e-19
19.3587174348697 2.34929930628833e-19
19.6392785571142 4.10288474015489e-19
19.9198396793587 7.09481728367261e-19
20.2004008016032 1.21510639287272e-18
20.4809619238477 2.06169240514591e-18
20.7615230460922 3.46641446720552e-18
21.0420841683367 5.77683924667277e-18
21.3226452905812 9.54454346450881e-18
21.6032064128257 1.5637705561448e-17
21.8837675350701 2.54119514920439e-17
22.1643286573146 4.09676111936108e-17
22.4448897795591 6.55342010325797e-17
22.7254509018036 1.04040871626133e-16
23.0060120240481 1.63956762197144e-16
23.2865731462926 2.56520394615325e-16
23.5671342685371 3.98526054396757e-16
23.8476953907816 6.1490237294197e-16
24.1282565130261 9.42410429143246e-16
24.4088176352705 1.43491401242621e-15
24.689378757515 2.1708441737946e-15
24.9699398797595 3.26370802762185e-15
25.250501002004 4.87679625264732e-15
25.5310621242485 7.24364833096663e-15
25.811623246493 1.06963691518808e-14
26.0921843687375 1.57045930765275e-14
26.372745490982 2.29288096516916e-14
26.6533066132265 3.32929587161749e-14
26.9338677354709 4.80827462579041e-14
27.2144288577154 6.90781185577571e-14
27.4949899799599 9.87308750053563e-14
27.7755511022044 1.40401639528283e-13
28.0561122244489 1.98674517820678e-13
28.3366733466934 2.79772950208076e-13
28.6172344689379 3.92106779142126e-13
28.8977955911824 5.46988688033431e-13
29.1783567134269 7.59568295986902e-13
29.4589178356713 1.05004474177099e-12
29.7394789579158 1.44523407078041e-12
30.0200400801603 1.98058675356391e-12
30.3006012024048 2.7027739216038e-12
30.5811623246493 3.67298923305389e-12
30.8617234468938 4.97114764109434e-12
31.1422845691383 6.70120384588474e-12
31.4228456913828 8.99786315378423e-12
31.7034068136273 1.20350173203962e-11
31.9839679358717 1.60363094029278e-11
32.2645290581162 2.1288316609241e-11
32.5450901803607 2.815694077526e-11
32.8256513026052 3.71077148730723e-11
33.1062124248497 4.87308735891912e-11
33.3867735470942 6.37721995738179e-11
33.6673346693387 8.31708478432219e-11
33.9478957915832 1.08105572771396e-10
34.2284569138277 1.40051039260408e-10
34.5090180360721 1.8084619654674e-10
34.7895791583166 2.32777034597219e-10
35.0701402805611 2.98676434493349e-10
35.3507014028056 3.82044271324748e-10
35.6312625250501 4.87191437042588e-10
35.9118236472946 6.1941202798703e-10
36.1923847695391 7.85188594296782e-10
36.4729458917836 9.92436083381328e-10
36.7535070140281 1.25079093464218e-09
37.0340681362725 1.57195270549023e-09
37.314629258517 1.97008663744685e-09
37.5951903807615 2.46229671389679e-09
37.875751503006 3.06918002609554e-09
38.1563126252505 3.81547465948452e-09
38.436873747495 4.73081484625514e-09
38.7174348697395 5.85060881015265e-09
38.997995991984 7.21705656304151e-09
39.2785571142285 8.88032690651074e-09
39.5591182364729 1.0899915051784e-08
39.8396793587174 1.33462046025113e-08
40.1202404809619 1.63022601516414e-08
40.4008016032064 1.98658794281567e-08
40.6813627254509 2.41519367931774e-08
40.9619238476954 2.92950529271948e-08
41.2424849699399 3.5452628768363e-08
41.5230460921844 4.28082851512066e-08
41.8036072144289 5.15757531484543e-08
42.0841683366734 6.20032638261898e-08
42.3647294589178 7.43784899714378e-08
42.6452905811623 8.90340963256607e-08
42.9258517034068 1.06353958938524e-07
43.2064128256513 1.26780118420826e-07
43.4869739478958 1.50820536097696e-07
43.7675350701403 1.79057726312721e-07
44.0480961923848 2.12158342376795e-07
44.3286573146293 2.50883797854432e-07
44.6092184368737 2.96102008992762e-07
44.8897795591182 3.48800348078993e-07
45.1703406813627 4.10099901310586e-07
45.4509018036072 4.81271128325701e-07
45.7314629258517 5.63751023810933e-07
46.0120240480962 6.5916188452052e-07
46.2925851703407 7.69331787539317e-07
46.5731462925852 8.96316887636121e-07
46.8537074148297 1.04242564301438e-06
47.1342685370741 1.21024507960175e-06
47.4148296593186 1.40266920415752e-06
47.6953907815631 1.62292967584177e-06
47.9759519038076 1.87462884440813e-06
48.2565130260521 2.16177526078346e-06
48.5370741482966 2.48882176240199e-06
48.8176352705411 2.86070623120884e-06
49.0981963927856 3.28289511665986e-06
49.3787575150301 3.76142980926234e-06
49.6593186372746 4.30297594216869e-06
49.939879759519 4.91487568898888e-06
50.2204408817635 5.60520311528423e-06
50.501002004008 6.38282262911266e-06
50.7815631262525 7.2574505624648e-06
51.062124248497 8.23971990044899e-06
51.3426853707415 9.34124815864729e-06
51.623246492986 1.05747083911419e-05
51.9038076152305 1.19539032923524e-05
52.1843687374749 1.34938423350163e-05
52.4649298597194 1.52108218644394e-05
52.7454909819639 1.71225080455944e-05
53.0260521042084 1.92480225347963e-05
53.3066132264529 2.16080307216364e-05
53.5871743486974 2.42248323596945e-05
53.8677354709419 2.71224543763477e-05
54.1482965931864 3.03267456229616e-05
54.4288577154309 3.38654732969135e-05
54.7094188376754 3.77684207364912e-05
54.9899799599198 4.20674862588758e-05
55.2705410821643 4.67967826803378e-05
55.5511022044088 5.19927371265931e-05
55.8316633266533 5.76941907103059e-05
56.1122244488978 6.39424976220852e-05
56.3927855711423 7.07816231512566e-05
56.6733466933868 7.82582401234933e-05
56.9539078156313 8.64218232142673e-05
57.2344689378758 9.53247405701778e-05
57.5150300601202 0.000105022342145035
57.7955911823647 0.000115573044134075
58.0761523046092 0.000127038408868335
58.3567134268537 0.000139483219512171
58.6372745490982 0.000152975548890463
58.9178356713427 0.000167586821758302
59.1983967935872 0.000183391869815546
59.4789579158317 0.00020046897876111
59.7595190380761 0.000218899926678099
60.0400801603206 0.000238770013040947
60.3206412825651 0.000260168077638955
60.6012024048096 0.000283186508719012
60.8817635270541 0.000307921239661724
61.1623246492986 0.000334471733521531
61.4428857715431 0.000362940954782237
61.7234468937876 0.000393435327704051
62.0040080160321 0.000426064680668254
62.2845691382766 0.000460942175959308
62.565130260521 0.000498184224463186
62.8456913827655 0.00053791038480359
63.12625250501 0.000580243246485423
63.4068136272545 0.000625308296666088
63.687374749499 0.000673233770232188
63.9679358717435 0.00072415048291794
64.248496993988 0.000778191647266053
64.5290581162325 0.000835492671299229
64.809619238477 0.000896190939840248
65.0901803607214 0.000960425578493061
65.3707414829659 0.00102833720037342
65.6513026052104 0.00110006763575571
65.9318637274549 0.00117575964488455
66.2124248496994 0.00125555661428032
66.4929859719439 0.0013396022369536
66.7735470941884 0.00142804017702585
67.0541082164329 0.00152101371933906
67.3346693386774 0.00161866540472265
67.6152304609219 0.00172113665166826
67.8957915831663 0.00182856736524755
68.1763527054108 0.00194109553418964
68.4569138276553 0.00205885681711203
68.7374749498998 0.00218198411898156
69.0180360721443 0.00231060715894794
69.2985971943888 0.00244485203077228
69.5791583166333 0.00258484075713309
69.8597194388778 0.0027306908391589
70.1402805611222 0.00288251480259271
70.4208416833667 0.00304041974204965
70.7014028056112 0.00320450686487229
70.9819639278557 0.00337487103613323
71.2625250501002 0.00355160032637203
71.5430861723447 0.00373477556367613
71.8236472945892 0.00392446989174747
72.1042084168337 0.00412074833560231
72.3847695390782 0.00432366737657122
72.6653306613226 0.00453327453825725
72.9458917835671 0.00474960798511323
73.2264529058116 0.00497269613528074
73.5070140280561 0.00520255728931377
73.7875751503006 0.00543919927638301
74.0681362725451 0.00568261911952208
74.3486973947896 0.00593280272142977
74.6292585170341 0.00618972457229968
74.9098196392786 0.00645334748108138
75.190380761523 0.00672362233152898
75.4709418837675 0.00700048786430104
75.751503006012 0.00728387048632312
76.0320641282565 0.00757368410852228
76.312625250501 0.00786983001296413
76.5931863727455 0.00817219675032804
76.87374749499 0.00848066006855733
77.1543086172345 0.00879508287341853
77.434869739479 0.00911531522159372
77.7154308617234 0.00944119434682897
77.9959919839679 0.00977254471953532
78.2765531062124 0.0101091781401391
78.5571142284569 0.0104508938663409
78.8376753507014 0.0107974787743459
79.1182364729459 0.011148707553988
79.3987975951904 0.0115043429375638
79.6793587174349 0.0118641359620631
79.9599198396794 0.0122278262643761
80.2404809619239 0.0125951424089156
80.5210420841683 0.0129658022470074
80.8016032064128 0.0133395133072614
81.0821643286573 0.0137159732160392
81.3627254509018 0.0140948701470233
81.6432865731463 0.0144758832987825
81.9238476953908 0.0148586833991462
82.2044088176353 0.0152429332350723
82.4849699398798 0.015628288206651
82.7655310621243 0.0160143969037433
83.0460921843687 0.016400901703739
83.3266533066132 0.0167874393887798
83.6072144288577 0.0171736417807911
83.8877755511022 0.0175591363925454
84.1683366733467 0.0179435470929821
84.4488977955912 0.018326494784916
84.7294589178357 0.0187075980932663
85.0100200400802 0.019086474061869
85.2905811623247 0.0194627388569582
85.5711422845691 0.0198360084753185
85.8517034068136 0.0202058994551856
86.1322645290581 0.0205720295878808
86.4128256513026 0.0209340186282493
86.6933867735471 0.0212914890019228
86.9739478957916 0.0216440665074813
87.2545090180361 0.0219913810116163
87.5350701402806 0.0223330671353937
87.815631262525 0.0226687649298086
88.0961923847695 0.0229981205388243
88.376753507014 0.0233207868481553
88.6573146292585 0.0236364241181199
88.937875751503 0.023944700598943
89.2184368737475 0.0242452931269485
89.498997995992 0.0245378877001716
89.7795591182365 0.0248221800319837
90.060120240481 0.0250978760814061
90.3406813627254 0.025364692558884
90.6212424849699 0.0256223574063606
90.9018036072144 0.0258706102505984
91.1823647294589 0.0261092028287834
91.4629258517034 0.0263378993855297
91.7434869739479 0.0265564770405311
92.0240480961924 0.0267647261261532
92.3046092184369 0.0269624504944222
92.5851703406814 0.0271494677928974
92.8657314629259 0.0273256097090759
93.1462925851703 0.0274907221830241
93.4268537074148 0.0276446655880787
93.7074148296593 0.027787314879498
93.9879759519038 0.0279185597111102
94.2685370741483 0.0280383045200267
94.5490981963928 0.0281464685796404
94.8296593186373 0.0282429860211812
95.1102204408818 0.0283278058242018
95.3907815631263 0.0284008917764481
95.6713426853707 0.0284622224036541
95.9519038076152 0.0285117908698598
96.2324649298597 0.0285496048489693
96.5130260521042 0.0285756863682647
96.7935871743487 0.0285900716247239
97.0741482965932 0.0285928107750363
97.3547094188377 0.028583967700217
97.6352705410822 0.0285636197458664
97.9158316633267 0.0285318574390776
98.1963927855711 0.0284887841831084
98.4769539078156 0.0284345159309274
98.7575150300601 0.02836918083881
99.0380761523046 0.0282929189011794
99.3186372745491 0.0282058815679054
99.5991983967936 0.0281082313453173
99.8797595190381 0.0280001413821764
100.160320641283 0.0278817950418935
100.440881763527 0.0277533854622789
100.721442885772 0.0276151151040903
101.002004008016 0.0274671952896903
101.282565130261 0.0273098457330813
101.563126252505 0.027143294062601
101.84368737475 0.0269677753375448
102.124248496994 0.0267835315599537
102.404809619238 0.0265908111828176
102.685370741483 0.0263898686158909
102.965931863727 0.0261809637303096
103.246492985972 0.0259643613631675
103.527054108216 0.0257403308231811
103.807615230461 0.0255091453985342
104.088176352705 0.0252710818679748
104.36873747495 0.0250264200161646
104.649298597194 0.0247754421542842
104.929859719439 0.0245184326468266
105.210420841683 0.0242556774454858
105.490981963928 0.0239874636309853
105.771543086172 0.0237140789636698
106.052104208417 0.0234358114436194
106.332665330661 0.0231529488809972
106.613226452906 0.0228657784773153
106.89378757515 0.0225745864182356
107.174348697395 0.0222796574784658
107.454909819639 0.0219812746393018
107.735470941884 0.0216797187192803
108.016032064128 0.0213752680183599
108.296593186373 0.0210681979760454
108.577154308617 0.0207587808437562
108.857715430862 0.0204472853717421
109.138276553106 0.0201339765107981
109.418837675351 0.0198191151289645
109.699398797595 0.0195029577433574
109.97995991984 0.0191857562672716
110.260521042084 0.0188677577725958
110.541082164329 0.0185492042675795
110.821643286573 0.0182303324899526
111.102204408818 0.0179113737153342
111.382765531062 0.0175925535808648
111.663326653307 0.0172740919239268
111.943887775551 0.0169562026358197
112.224448897796 0.016639093530202
112.50501002004 0.0163229662260847
112.785571142285 0.0160080160451508
113.066132264529 0.0156944319231188
113.346693386774 0.0153823963348813
113.627254509018 0.0150720852330879
113.907815631263 0.014763667999855
114.188376753507 0.0144573074112415
114.468937875752 0.014153159614129
114.749498997996 0.01385137411511
115.03006012024 0.0135520937809998
115.310621242485 0.013255454850549
115.591182364729 0.012961586956938
115.871743486974 0.0126706131606269
116.152304609218 0.0123826499921085
116.432865731463 0.0120978075041333
116.713426853707 0.0118161893329471
116.993987975952 0.0115378927680882
117.274549098196 0.0112630088302898
117.555110220441 0.0109916223570329
117.835671342685 0.0107238120952854
118.11623246493 0.0104596508009804
118.396793587174 0.0101992053447807
118.677354709419 0.00994253682367561
118.957915831663 0.00968970067797533
119.238476953908 0.00944074681325363
119.519038076152 0.00919571972682488
119.799599198397 0.00895465863831232
120.080160320641 0.00871759762390889
120.360721442886 0.00848456575391356
120.64128256513 0.00825558723315567
120.921843687375 0.0080306815439117
121.202404809619 0.0078098635909469
121.482965931864 0.0075931438483122
121.763527054108 0.00738052850754654
122.044088176353 0.00717201962694605
122.324649298597 0.0069676152815651
122.605210420842 0.00676730971363932
122.885771543086 0.00657109348312949
123.166332665331 0.00637895361809033
123.446893787575 0.00619087376459089
123.72745490982 0.00600683433592609
124.008016032064 0.00582681266086407
124.288577154309 0.00565078313069682
124.569138276553 0.00547871734486763
124.849699398798 0.00531058425496814
125.130260521042 0.00514635030690427
125.410821643287 0.00498597958104987
125.691382765531 0.00482943393021387
125.971943887776 0.0046766731152646
126.25250501002 0.00452765493826356
126.533066132265 0.00438233537297122
126.813627254509 0.00424066869261022
127.094188376754 0.00410260759476452
127.374749498998 0.00396810332332417
127.655310621242 0.00383710578738132
127.935871743487 0.00370956367700083
128.216432865731 0.00358542457579991
128.496993987976 0.0034646350702748
128.77755511022 0.00334714085583082
129.058116232465 0.00323288683947337
129.338677354709 0.00312181723913297
129.619238476954 0.00301387567959868
129.899799599198 0.00290900528505141
130.180360721443 0.00280714876818743
130.460921843687 0.00270824851593624
130.741482965932 0.00261224667177946
131.022044088176 0.00251908521468794
131.302605210421 0.00242870603469758
131.583166332665 0.00234105100515062
131.86372745491 0.00225606205163658
132.144288577154 0.00217368121766962
132.424849699399 0.00209385072714554
132.705410821643 0.00201651304362578
132.985971943888 0.00194161092649828
133.266533066132 0.00186908748407248
133.547094188377 0.00179888622366355
133.827655310621 0.00173095109873125
134.108216432866 0.00166522655313374
134.38877755511 0.00160165756256547
134.669338677355 0.00154018967324778
134.949899799599 0.00148076903794198
135.230460921844 0.00142334244935891
135.511022044088 0.00136785737103724
135.791583166333 0.00131426196576755
136.072144288577 0.00126250512163636
136.352705410822 0.00121253647576796
136.633266533066 0.00116430643584111
136.913827655311 0.00111776619945778
137.194388777555 0.00107286777144205
137.4749498998 0.00102956397914591
137.755511022044 0.000987808485840305
138.036072144289 0.000947555802267618
138.316633266533 0.000908761296432292
138.597194388778 0.000871381201704667
138.877755511022 0.00083537262331375
139.158316633267 0.000800693543301875
139.438877755511 0.000767302824015197
139.719438877756 0.000735160210200082
140 0.000704226329777603
};
\end{axis}

\end{tikzpicture}

%% file: simulation/notebooks/figures/bgmres_gaussian_uq__seed_12345__eq_type_spd_reverse__iteration_3.tikz
\begin{tikzpicture}

\begin{axis}[
xlabel={$z$},
ylabel={$p(z)$},
xmin=0, xmax=140,
ymin=-0.01, ymax=0.3,
width=\figwidth,
height=\figheight,
xtick={0,20,40,60,80,100,120,140},
xticklabels={$0$,$20$,$40$,$60$,$80$,$100$,$120$,$140$},
ytick={-0.05,0,0.05,0.1,0.15,0.2,0.25,0.3},
yticklabels={,$0.00$,$0.05$,$0.10$,$0.15$,$0.20$,$0.25$,$0.30$},
tick align=outside,
tick pos=left,
x grid style={white!69.01960784313725!black},
y grid style={white!69.01960784313725!black},
legend cell align={left},
legend entries={{BayesGMRES},{$\chi^2_{97}$}},
legend style={draw=white!80.0!black}
]
\addlegendimage{no markers, blue!50.98039215686274!black}
\addlegendimage{no markers, black}
\addplot [semithick, blue!50.98039215686274!black]
table {%
0 6.32655560118491e-11
0.280561122244489 1.40601616513209e-10
0.561122244488978 3.06541995854883e-10
0.841683366733467 6.55647586281078e-10
1.12224448897796 1.37573949795259e-09
1.40280561122244 2.83199176135355e-09
1.68336673346693 5.71930172554811e-09
1.96392785571142 1.13316975332073e-08
2.24448897795591 2.2026939868164e-08
2.5250501002004 4.20075707522143e-08
2.80561122244489 7.86000796463469e-08
3.08617234468938 1.44294312243444e-07
3.36673346693387 2.5990593979292e-07
3.64729458917836 4.59341103899623e-07
3.92785571142285 7.96561894727515e-07
4.20841683366733 1.35545101091539e-06
4.48897795591182 2.26331771237351e-06
4.76953907815631 3.70872586893132e-06
5.0501002004008 5.9641000734701e-06
5.33066132264529 9.41311893428301e-06
5.61122244488978 1.45821930237001e-05
5.89178356713427 2.21743431947915e-05
6.17234468937876 3.3102599793319e-05
6.45290581162325 4.85187751039285e-05
6.73346693386774 6.98323544068508e-05
7.01402805611222 9.87136220942996e-05
7.29458917835671 0.000137075346988967
7.5751503006012 0.000187028722342108
7.85571142284569 0.000250811986660319
8.13627254509018 0.000330694197243231
8.41683366733467 0.000428861615726668
8.69739478957916 0.000547299356365039
8.97795591182365 0.000687685291774292
9.25851703406814 0.000851315495439067
9.53907815631263 0.00103907960142005
9.81963927855711 0.0012514996475715
10.1002004008016 0.00148883719211147
10.3807615230461 0.00175126158919536
10.6613226452906 0.00203905901695059
10.9418837675351 0.00235284961851658
11.2224448897796 0.00269377168355247
11.503006012024 0.00306358963198747
11.7835671342685 0.00346468825703947
12.064128256513 0.00389992946883547
12.3446893787575 0.00437236826922953
12.625250501002 0.0048848489718959
12.9058116232465 0.00543952675318115
13.186372745491 0.00603737908312105
13.4669338677355 0.00667778250023743
13.74749498998 0.00735822987047724
14.0280561122244 0.00807425086443375
14.3086172344689 0.00881957514024597
14.5891783567134 0.0095865468098614
14.8697394789579 0.0103667647811264
15.1503006012024 0.0111518917086788
15.4308617234469 0.0119345494553697
15.7114228456914 0.0127092049139123
15.9919839679359 0.0134729487351388
16.2725450901804 0.0142260808534606
16.5531062124249 0.0149724386204065
16.8336673346693 0.0157194322698456
17.1142284569138 0.0164777839573885
17.3947895791583 0.017260996355328
17.6753507014028 0.0180846010817232
17.9559118236473 0.0189652536643444
18.2364729458918 0.0199197492868649
18.5170340681363 0.0209640325797622
18.7975951903808 0.0221122665844273
19.0781563126253 0.0233760127540936
19.3587174348697 0.0247635577285215
19.6392785571142 0.0262794058601079
19.9198396793587 0.0279239410838865
20.2004008016032 0.0296932493853782
20.4809619238477 0.0315790850685117
20.7615230460922 0.0335689609415666
21.0420841683367 0.0356463443936246
21.3226452905812 0.0377909472714651
21.6032064128257 0.0399791058022637
21.8837675350701 0.0421842552130585
22.1643286573146 0.0443775095975684
22.4448897795591 0.046528358697379
22.7254509018036 0.0486054881953278
23.0060120240481 0.0505777187803979
23.2865731462926 0.0524150430374976
23.5671342685371 0.0540897208399411
23.8476953907816 0.0555773768409338
24.1282565130261 0.0568580313484984
24.4088176352705 0.0579169910222527
24.689378757515 0.0587455297453676
24.9699398797595 0.0593413023117069
25.250501002004 0.0597084522887502
25.5310621242485 0.0598573976058578
25.811623246493 0.0598042998443646
26.0921843687375 0.0595702431652326
26.372745490982 0.0591801647100076
26.6533066132265 0.058661589892045
26.9338677354709 0.0580432341161125
27.2144288577154 0.0573535385315215
27.4949899799599 0.0566192126532098
27.7755511022044 0.0558638614512145
28.0561122244489 0.0551067780065706
28.3366733466934 0.0543619831824853
28.6172344689379 0.0536375885452196
28.8977955911824 0.0529355458278843
29.1783567134269 0.0522518244631439
29.4589178356713 0.0515770286360333
29.7394789579158 0.0508974291939559
30.0200400801603 0.0501963472970075
30.3006012024048 0.0494557903795888
30.5811623246493 0.0486582113208453
30.8617234468938 0.0477882425344353
31.1422845691383 0.0468342507607801
31.4228456913828 0.0457895672031954
31.7034068136273 0.0446532715502782
31.9839679358717 0.0434304463456634
32.2645290581162 0.0421318676935121
32.5450901803607 0.0407731554585971
32.8256513026052 0.0393734653190934
33.1062124248497 0.0379538591994115
33.3867735470942 0.0365355319355267
33.6673346693387 0.0351380931498579
33.9478957915832 0.0337780989031762
34.2284569138277 0.0324679959828495
34.5090180360721 0.0312155855413072
34.7895791583166 0.0300240398499093
35.0701402805611 0.0288924276111132
35.3507014028056 0.0278166327908544
35.6312625250501 0.0267905016496258
35.9118236472946 0.0258070314016108
36.1923847695391 0.0248594250226997
36.4729458917836 0.0239418771062494
36.7535070140281 0.0230500165109176
37.0340681362725 0.0221810001812497
37.314629258517 0.0213333151905816
37.5951903807615 0.0205063909022929
37.875751503006 0.0197001427071849
38.1563126252505 0.0189145614223638
38.436873747495 0.0181494323425798
38.7174348697395 0.0174042239053492
38.997995991984 0.0166781390619725
39.2785571142285 0.0159702834755959
39.5591182364729 0.015279881558747
39.8396793587174 0.0146064677010343
40.1202404809619 0.0139499945673309
40.4008016032064 0.0133108276902584
40.6813627254509 0.012689627878636
40.9619238476954 0.0120871519669669
41.2424849699399 0.0115040214998201
41.5230460921844 0.010940514469962
41.8036072144289 0.0103964271832721
42.0841683366734 0.00987103489232546
42.3647294589178 0.00936315631828057
42.6452905811623 0.00887130448420366
42.9258517034068 0.00839388949061864
43.2064128256513 0.00792943113227742
43.4869739478958 0.00747674132458817
43.7675350701403 0.00703504663782705
44.0480961923848 0.0066040366444658
44.3286573146293 0.00618384036570648
44.6092184368737 0.00577494717832018
44.8897795591182 0.00537809741256522
45.1703406813627 0.00499417024878944
45.4509018036072 0.00462409260479631
45.7314629258517 0.00426878395666088
46.0120240480962 0.00392914073482275
46.2925851703407 0.00360605267057882
46.5731462925852 0.00330043462091208
46.8537074148297 0.00301325274609776
47.1342685370741 0.00274552434386303
47.4148296593186 0.00249827604355027
47.6953907815631 0.00227245436352915
47.9759519038076 0.00206879400202553
48.2565130260521 0.00188766039237562
48.5370741482966 0.0017288917042537
48.8176352705411 0.00159166971250155
49.0981963927856 0.00147444767354699
49.3787575150301 0.00137495647881834
49.6593186372746 0.00129029893312233
49.939879759519 0.00121712801459389
50.2204408817635 0.00115189096746594
50.501002004008 0.00109110968535234
50.7815631262525 0.00103166122055572
51.062124248497 0.000971021650242611
51.3426853707415 0.00090744202131677
51.623246492986 0.000840035590623096
51.9038076152305 0.000768769063974471
52.1843687374749 0.000694364525456943
52.4649298597194 0.000618130792359855
52.7454909819639 0.000541751117606226
53.0260521042084 0.000467057448676399
53.3066132264529 0.000395819795617398
53.5871743486974 0.000329573518751338
53.8677354709419 0.000269499004725129
54.1482965931864 0.000216358999756973
54.4288577154309 0.000170490437313285
54.7094188376754 0.000131841151833767
54.9899799599198 0.000100038060472309
55.2705410821643 7.44722878125762e-05
55.5511022044088 5.43878929791808e-05
55.8316633266533 3.89636248625054e-05
56.1122244488978 2.73806702853146e-05
56.3927855711423 1.88729353692419e-05
56.6733466933868 1.27594630827099e-05
56.9539078156313 8.46082671845226e-06
57.2344689378758 5.5026618057956e-06
57.5150300601202 3.50999409264531e-06
57.7955911823647 2.19588081685126e-06
58.0761523046092 1.34733711189766e-06
58.3567134268537 8.10785591692226e-07
58.6372745490982 4.78515207988962e-07
58.9178356713427 2.76976960866915e-07
59.1983967935872 1.57234663605555e-07
59.4789579158317 8.75404043936948e-08
59.7595190380761 4.77995131637536e-08
60.0400801603206 2.55971908703467e-08
60.3206412825651 1.3443563827506e-08
60.6012024048096 6.9245133335235e-09
60.8817635270541 3.49797234676701e-09
61.1623246492986 1.73298720628728e-09
61.4428857715431 8.4202667274555e-10
61.7234468937876 4.01243047516217e-10
62.0040080160321 1.87516750152853e-10
62.2845691382766 8.59455125583096e-11
62.565130260521 3.86328425344669e-11
62.8456913827655 1.70310012252032e-11
63.12625250501 7.36331704458172e-12
63.4068136272545 3.12216876886079e-12
63.687374749499 1.29834036095818e-12
63.9679358717435 5.29504682055851e-13
64.248496993988 2.11787279387597e-13
64.5290581162325 8.30765619248922e-14
64.809619238477 3.19599059258183e-14
65.0901803607214 1.20581465668243e-14
65.3707414829659 4.46173187726835e-15
65.6513026052104 1.61910112865332e-15
65.9318637274549 5.76224649499988e-16
66.2124248496994 2.01120708347057e-16
66.4929859719439 6.88443965789031e-17
66.7735470941884 2.3111439901425e-17
67.0541082164329 7.60907308826048e-18
67.3346693386774 2.45687238282144e-18
67.6152304609219 7.77999426637622e-19
67.8957915831663 2.41613687836174e-19
68.1763527054108 7.35883580098307e-20
68.4569138276553 2.19807118617296e-20
68.7374749498998 6.43901144042955e-21
69.0180360721443 1.84986988557946e-21
69.2985971943888 5.21203651174848e-22
69.5791583166333 1.44018318428298e-22
69.8597194388778 3.90275940748431e-23
70.1402805611222 1.03721626196917e-23
70.4208416833667 2.70339861305392e-24
70.7014028056112 6.91025102851988e-25
70.9819639278557 1.73228867709449e-25
71.2625250501002 4.2588186812791e-26
71.5430861723447 1.02683393314311e-26
71.8236472945892 2.42802519158019e-27
72.1042084168337 5.63051106370646e-28
72.3847695390782 1.28051253068755e-28
72.6653306613226 2.85601823174457e-29
72.9458917835671 6.24710754661416e-30
73.2264529058116 1.34010118896876e-30
73.5070140280561 2.81926918479056e-31
73.7875751503006 5.81668543913423e-32
74.0681362725451 1.17694030120995e-32
74.3486973947896 2.33546226367361e-33
74.6292585170341 4.54496539066075e-34
74.9098196392786 8.67415854684499e-35
75.190380761523 1.62353974590024e-35
75.4709418837675 2.98014251472861e-36
75.751503006012 5.36474923796101e-37
76.0320641282565 9.47108759985727e-38
76.312625250501 1.63978935345531e-38
76.5931863727455 2.78428574319308e-39
76.87374749499 4.63635656796938e-40
77.1543086172345 7.57141220922648e-41
77.434869739479 1.21258933656279e-41
77.7154308617234 1.90452705258745e-42
77.9959919839679 2.93357267742702e-43
78.2765531062124 4.43141702936028e-44
78.5571142284569 6.56484108182297e-45
78.8376753507014 9.53765040593112e-46
79.1182364729459 1.35892013440057e-46
79.3987975951904 1.89881032375168e-47
79.6793587174349 2.60198099428582e-48
79.9599198396794 3.49672364259742e-49
80.2404809619239 4.60842926010021e-50
80.5210420841683 5.95632984057347e-51
80.8016032064128 7.5498535444109e-52
81.0821643286573 9.38495220370218e-53
81.3627254509018 1.14408721527936e-53
81.6432865731463 1.36779027019703e-54
81.9238476953908 1.60366252426769e-55
82.2044088176353 1.8439083383834e-56
82.4849699398798 2.07920992924316e-57
82.7655310621243 2.29926911541783e-58
83.0460921843687 2.49352378365336e-59
83.3266533066132 2.65197399359605e-60
83.6072144288577 2.76602969651146e-61
83.8877755511022 2.82928074963843e-62
84.1683366733467 2.83809336607824e-63
84.4488977955912 2.79195581709658e-64
84.7294589178357 2.69352763733494e-65
85.0100200400802 2.54838575167573e-66
85.2905811623247 2.36450123523975e-67
85.5711422845691 2.15151499180145e-68
85.8517034068136 1.91990382557389e-69
86.1322645290581 1.68013682650188e-70
86.4128256513026 1.44191521288531e-71
86.6933867735471 1.21356910088489e-72
86.9739478957916 1.00165653816826e-73
87.2545090180361 8.10779026980773e-75
87.5350701402806 6.4359900174764e-76
87.815631262525 5.01022493597441e-77
88.0961923847695 3.82496901422229e-78
88.376753507014 2.86369844660972e-79
88.6573146292585 2.10259266283377e-80
88.937875751503 1.51394932235424e-81
89.2184368737475 1.06904429615737e-82
89.498997995992 7.40300549427358e-84
89.7795591182365 5.02745615503497e-85
90.060120240481 3.34823709581359e-86
90.3406813627254 2.18681302774568e-87
90.6212424849699 1.40066567206168e-88
90.9018036072144 8.7980120002947e-90
91.1823647294589 5.41953092150226e-91
91.4629258517034 3.27390281522091e-92
91.7434869739479 1.93953070495481e-93
92.0240480961924 1.12681905702362e-94
92.3046092184369 6.4200470933229e-96
92.5851703406814 3.58714235205116e-97
92.8657314629259 1.965555173635e-98
93.1462925851703 1.05620437191647e-99
93.4268537074148 5.56591593995729e-101
93.7074148296593 2.87641241094636e-102
93.9879759519038 1.45777790505092e-103
94.2685370741483 7.24531286174774e-105
94.5490981963928 3.5314118324352e-106
94.8296593186373 1.68797073062962e-107
95.1102204408818 7.9123711800577e-109
95.3907815631263 3.63725265003052e-110
95.6713426853707 1.63970319205644e-111
95.9519038076152 7.24906266987055e-113
96.2324649298597 3.14284664263124e-114
96.5130260521042 1.33625435759146e-115
96.7935871743487 5.57159546001584e-117
97.0741482965932 2.27821340613314e-118
97.3547094188377 9.13552493335858e-120
97.6352705410822 3.59250003218324e-121
97.9158316633267 1.38542850060677e-122
98.1963927855711 5.2395671854216e-124
98.4769539078156 1.94325845492048e-125
98.7575150300601 7.06788498360961e-127
99.0380761523046 2.52099530458044e-128
99.3186372745491 8.81816371425901e-130
99.5991983967936 3.02487713474742e-131
99.8797595190381 1.01756154575286e-132
100.160320641283 3.35688841878371e-134
100.440881763527 1.08601620350718e-135
100.721442885772 3.4455518790354e-137
101.002004008016 1.07202344020073e-138
101.282565130261 3.27094077975379e-140
101.563126252505 9.78732266704482e-142
101.84368737475 2.87195659027758e-143
102.124248496994 8.26446019400039e-145
102.404809619238 2.33224240077473e-146
102.685370741483 6.45439252191558e-148
102.965931863727 1.75169896073822e-149
103.246492985972 4.66214637633603e-151
103.527054108216 1.2168431272011e-152
103.807615230461 3.1146228267865e-154
104.088176352705 7.81804985984118e-156
104.36873747495 1.92448017124136e-157
104.649298597194 4.64569190088339e-159
104.929859719439 1.09978870519037e-160
105.210420841683 2.55322990843769e-162
105.490981963928 5.8128931361428e-164
105.771543086172 1.29782589401003e-165
106.052104208417 2.84159476444692e-167
106.332665330661 6.10139863218269e-169
106.613226452906 1.28474855975482e-170
106.89378757515 2.65294531343527e-172
107.174348697395 5.37229450691293e-174
107.454909819639 1.06687274922626e-175
107.735470941884 2.07771830714338e-177
108.016032064128 3.9680933971399e-179
108.296593186373 7.43187200081517e-181
108.577154308617 1.36500931683203e-182
108.857715430862 2.45863473188058e-184
109.138276553106 4.34283526486973e-186
109.418837675351 7.52269737070338e-188
109.699398797595 1.27789350963784e-189
109.97995991984 2.12880814383766e-191
110.260521042084 3.47775604939702e-193
110.541082164329 5.57163154467158e-195
110.821643286573 8.75359187158073e-197
111.102204408818 1.34868579124888e-198
111.382765531062 2.03777300103076e-200
111.663326653307 3.01940459486323e-202
111.943887775551 4.38740027782956e-204
112.224448897796 6.25192267738525e-206
112.50501002004 8.73655640848803e-208
112.785571142285 1.19725669828041e-209
113.066132264529 1.60899437081317e-211
113.346693386774 2.12051812985155e-213
113.627254509018 2.74062508663462e-215
113.907815631263 3.4735807525891e-217
114.188376753507 4.31742956960485e-219
114.468937875752 5.26251291302239e-221
114.749498997996 6.29044107856761e-223
115.03006012024 7.37375995533211e-225
115.310621242485 8.47650508045045e-227
115.591182364729 9.55574506093885e-229
115.871743486974 1.05640912186205e-230
116.152304609218 1.14530077422704e-232
116.432865731463 1.2176620030877e-234
116.713426853707 1.26956139309468e-236
116.993987975952 1.29807677315732e-238
117.274549098196 1.30156762018133e-240
117.555110220441 1.27983135103203e-242
117.835671342685 1.23412279024054e-244
118.11623246493 1.16703421116046e-246
118.396793587174 1.08225190286479e-248
118.677354709419 9.84221059561204e-251
118.957915831663 8.77761328417079e-253
119.238476953908 7.67679065039451e-255
119.519038076152 6.58419048217772e-257
119.799599198397 5.53789200365563e-259
120.080160320641 4.56778817246149e-261
120.360721442886 3.69476464500367e-263
120.64128256513 2.9308050121601e-265
120.921843687375 2.27985049230522e-267
121.202404809619 1.7391821908073e-269
121.482965931864 1.30107710580303e-271
121.763527054108 9.5450939718094e-274
122.044088176353 6.86715064978009e-276
122.324649298597 4.84498135387119e-278
122.605210420842 3.35217582234999e-280
122.885771543086 2.2744719719504e-282
123.166332665331 1.51339915385548e-284
123.446893787575 9.87518996299191e-287
123.72745490982 6.31911824993943e-289
124.008016032064 3.9653957432916e-291
124.288577154309 2.44025717208584e-293
124.569138276553 1.47266391906867e-295
124.849699398798 8.71546696596702e-298
125.130260521042 5.05820726578709e-300
125.410821643287 2.87886636830887e-302
125.691382765531 1.60681285152224e-304
125.971943887776 8.79484060584449e-307
126.25250501002 4.72073384764583e-309
126.533066132265 2.48490566794364e-311
126.813627254509 1.28271190931758e-313
127.094188376754 6.49332528924237e-316
127.374749498998 3.22347201841371e-318
127.655310621242 1.5691524911918e-320
127.935871743487 7.4109846876187e-323
128.216432865731 0
128.496993987976 0
128.77755511022 0
129.058116232465 0
129.338677354709 0
129.619238476954 0
129.899799599198 0
130.180360721443 0
130.460921843687 0
130.741482965932 0
131.022044088176 0
131.302605210421 0
131.583166332665 0
131.86372745491 0
132.144288577154 0
132.424849699399 0
132.705410821643 0
132.985971943888 0
133.266533066132 0
133.547094188377 0
133.827655310621 0
134.108216432866 0
134.38877755511 0
134.669338677355 0
134.949899799599 0
135.230460921844 0
135.511022044088 0
135.791583166333 0
136.072144288577 0
136.352705410822 0
136.633266533066 0
136.913827655311 0
137.194388777555 0
137.4749498998 0
137.755511022044 0
138.036072144289 0
138.316633266533 0
138.597194388778 0
138.877755511022 0
139.158316633267 0
139.438877755511 0
139.719438877756 0
140 0
};
\addplot [very thick, black, dashed]
table {%
0 0
0.280561122244489 7.38364315802862e-102
0.561122244488978 1.27723973221705e-87
0.841683366733467 2.56852709616844e-79
1.12224448897796 1.92022034534626e-73
1.40280561122244 6.6935052928122e-69
1.68336673346693 3.35613734037149e-65
1.96392785571142 4.41457763376141e-62
2.24448897795591 2.18063795270882e-59
2.5250501002004 5.09806430383385e-57
2.80561122244489 6.60637218647016e-55
3.08617234468938 5.3111982006955e-53
3.36673346693387 2.8788966685396e-51
3.64729458917836 1.12073529907697e-49
3.92785571142285 3.29118607889817e-48
4.20841683366733 7.58023158748613e-47
4.48897795591182 1.41294016756741e-45
4.76953907815631 2.18688168435813e-44
5.0501002004008 2.87092892312333e-43
5.33066132264529 3.25425379398159e-42
5.61122244488978 3.23338258600567e-41
5.89178356713427 2.85250048688783e-40
6.17234468937876 2.25924584791665e-39
6.45290581162325 1.62194189173666e-38
6.73346693386774 1.06432432452803e-37
7.01402805611222 6.43083467523298e-37
7.29458917835671 3.60103952243732e-36
7.5751503006012 1.87953138990337e-35
7.85571142284569 9.19081873412688e-35
8.13627254509018 4.22989453022713e-34
8.41683366733467 1.83975996530678e-33
8.69739478957916 7.59030012239469e-33
8.97795591182365 2.98043316778089e-32
9.25851703406814 1.11723488593599e-31
9.53907815631263 4.0092009988647e-31
9.81963927855711 1.38076033754235e-30
10.1002004008016 4.57437652828673e-30
10.3807615230461 1.46090488983553e-29
10.6613226452906 4.50648381067999e-29
10.9418837675351 1.34512528391617e-28
11.2224448897796 3.89153030448573e-28
11.503006012024 1.09290131690841e-27
11.7835671342685 2.98377434415444e-27
12.064128256513 7.92963592433495e-27
12.3446893787575 2.05390649062034e-26
12.625250501002 5.19098608800475e-26
12.9058116232465 1.28153181366164e-25
13.186372745491 3.09355303148806e-25
13.4669338677355 7.3087821146113e-25
13.74749498998 1.69151933948365e-24
14.0280561122244 3.83809127402643e-24
14.3086172344689 8.54477012418146e-24
14.5891783567134 1.8678963264531e-23
14.8697394789579 4.01212705561293e-23
15.1503006012024 8.47327671773055e-23
15.4308617234469 1.76056666868591e-22
15.7114228456914 3.60108136069259e-22
15.9919839679359 7.25494662582131e-22
16.2725450901804 1.44040691585524e-21
16.5531062124249 2.81970051109013e-21
16.8336673346693 5.44494816433703e-21
17.1142284569138 1.03765610560833e-20
17.3947895791583 1.95239847770972e-20
17.6753507014028 3.62840899265067e-20
17.9559118236473 6.66293883025079e-20
18.2364729458918 1.20942357205455e-19
18.5170340681363 2.17074021970842e-19
18.7975951903808 3.85390468466771e-19
19.0781563126253 6.7701473499768e-19
19.3587174348697 1.17715459960946e-18
19.6392785571142 2.02644826609921e-18
19.9198396793587 3.45483341027774e-18
20.2004008016032 5.83480106490258e-18
20.4809619238477 9.76439310041885e-18
20.7615230460922 1.61954497544549e-17
21.0420841683367 2.66301285768294e-17
21.3226452905812 4.34195993714874e-17
21.6032064128257 7.02144584685307e-17
21.8837675350701 1.12638707698664e-16
22.1643286573146 1.79290712893705e-16
22.4448897795591 2.83218922551782e-16
22.7254509018036 4.44082037858884e-16
23.0060120240481 6.91289125490278e-16
23.2865731462926 1.06853327543593e-15
23.5671342685371 1.64029392950396e-15
23.8476953907816 2.5011024838244e-15
24.1282565130261 3.78866212639707e-15
24.4088176352705 5.70231058649152e-15
24.689378757515 8.5288450117029e-15
24.9699398797595 1.26784317544927e-14
25.250501002004 1.87342515092771e-14
25.5310621242485 2.75207464806732e-14
25.811623246493 4.01969220542322e-14
26.0921843687375 5.83832118804277e-14
26.372745490982 8.43330679005205e-14
26.6533066132265 1.21163840657069e-13
26.9338677354709 1.73165860648912e-13
27.2144288577154 2.4621415114516e-13
27.4949899799599 3.48314179510517e-13
27.7755511022044 4.90321829587845e-13
28.0561122244489 6.86888763291012e-13
28.3366733466934 9.57698027505067e-13
28.6172344689379 1.32907173885272e-12
28.8977955911824 1.83605363848005e-12
29.1783567134269 2.52509507078675e-12
29.4589178356713 3.45750446503003e-12
29.7394789579158 4.7138588091633e-12
30.0200400801603 6.39962220512371e-12
30.3006012024048 8.65227289202308e-12
30.5811623246493 1.1650307853703e-11
30.8617234468938 1.56245752773952e-11
31.1422845691383 2.08724819660462e-11
31.4228456913828 2.77757378974879e-11
31.7034068136273 3.68224363690973e-11
31.9839679358717 4.86344288239305e-11
32.2645290581162 6.40011421637947e-11
32.5450901803607 8.39212071640938e-11
32.8256513026052 1.09653523992756e-10
33.1062124248497 1.42779689729887e-10
33.3867735470942 1.85280058581726e-10
33.6673346693387 2.39626104056882e-10
33.9478957915832 3.08892210085827e-10
34.2284569138277 3.96890541763559e-10
34.5090180360721 5.08333243406033e-10
34.7895791583166 6.49026889724036e-10
35.0701402805611 8.26104883358974e-10
35.3507014028056 1.04830435742246e-09
35.6312625250501 1.32629511401419e-09
35.9118236472946 1.67306921822861e-09
36.1923847695391 2.1044010813813e-09
36.4729458917836 2.63938921669819e-09
36.7535070140281 3.30109234511926e-09
37.0340681362725 4.11727417769729e-09
37.314629258517 5.12127301355199e-09
37.5951903807615 6.35301427733176e-09
37.875751503006 7.86018628587857e-09
38.1563126252505 9.69960188776404e-09
38.436873747495 1.19387711680546e-08
38.7174348697395 1.46577131593073e-08
38.997995991984 1.79510374522561e-08
39.2785571142285 2.19303297579515e-08
39.5591182364729 2.67268788374619e-08
39.8396793587174 3.24947857834686e-08
40.1202404809619 3.94145004055897e-08
40.4008016032064 4.76968334306083e-08
40.6813627254509 5.75874973694676e-08
40.9619238476954 6.93722332111063e-08
41.2424849699399 8.33825845614696e-08
41.5230460921844 1.00002385433103e-07
41.8036072144289 1.19675032581239e-07
42.0841683366734 1.42911618046625e-07
42.3647294589178 1.70300002369326e-07
42.6452905811623 2.02514913742999e-07
42.9258517034068 2.40329163141988e-07
43.2064128256513 2.84626070126316e-07
43.4869739478958 3.36413198559316e-07
43.7675350701403 3.96837505801955e-07
44.0480961923848 4.67202013014753e-07
44.3286573146293 5.48984107936177e-07
44.6092184368737 6.43855594846223e-07
44.8897795591182 7.53704609288729e-07
45.1703406813627 8.80659517441725e-07
45.4509018036072 1.02711492170851e-06
45.7314629258517 1.19575989507101e-06
46.0120240480962 1.38960856691841e-06
46.2925851703407 1.61203318234918e-06
46.5731462925852 1.86679975526043e-06
46.8537074148297 2.15810643279839e-06
47.1342685370741 2.49062468486242e-06
47.4148296593186 2.86954342725433e-06
47.6953907815631 3.30061618066258e-06
47.9759519038076 3.790211359895e-06
48.2565130260521 4.34536577855905e-06
48.5370741482966 4.97384144366408e-06
48.8176352705411 5.68418570234825e-06
49.0981963927856 6.48579478904853e-06
49.3787575150301 7.38898080591392e-06
49.6593186372746 8.40504215209023e-06
49.939879759519 9.54633739864095e-06
50.2204408817635 1.08263625853595e-05
50.501002004008 1.22598318935295e-05
50.7815631262525 1.38627616249004e-05
51.062124248497 1.56525573917359e-05
51.3426853707415 1.76481043959015e-05
51.623246492986 1.98698606466013e-05
51.9038076152305 2.23399529366694e-05
52.1843687374749 2.50822753664286e-05
52.4649298597194 2.81225901720573e-05
52.7454909819639 3.14886305824819e-05
53.0260521042084 3.52102053950015e-05
53.3066132264529 3.931930492554e-05
53.5871743486974 4.38502079546117e-05
53.8677354709419 4.88395892551466e-05
54.1482965931864 5.43266272534497e-05
54.4288577154309 6.03531113398575e-05
54.7094188376754 6.69635483116618e-05
54.9899799599198 7.42052673975354e-05
55.2705410821643 8.21285232805118e-05
55.5511022044088 9.07865965057179e-05
55.8316633266533 0.000100235890629969
56.1122244488978 0.00011053602544292
56.3927855711423 0.000121749925564684
56.6733466933868 0.00013394390370217
56.9539078156313 0.000147187737826889
57.2344689378758 0.000161554741520253
57.5150300601202 0.00017712182671938
57.7955911823647 0.000193969558086742
58.0761523046092 0.000212182198221331
58.3567134268537 0.000231847742927466
58.6372745490982 0.000253057945760255
58.9178356713427 0.000275908331073711
59.1983967935872 0.00030049819480981
59.4789579158317 0.000326930592283485
59.7595190380761 0.000355312312240111
60.0400801603206 0.000385753836489347
60.3206412825651 0.000418369284450439
60.6012024048096 0.000453276341981704
60.8817635270541 0.000490596173908704
61.1623246492986 0.000530453319713063
61.4428857715431 0.000572975571895784
61.7234468937876 0.000618293836586327
62.0040080160321 0.000666541976030561
62.2845691382766 0.000717856632656317
62.565130260521 0.000772377034487231
62.8456913827655 0.000830244781748969
63.12625250501 0.00089160361459155
63.4068136272545 0.000956599161931998
63.687374749499 0.0010253786715088
63.9679358717435 0.00109809072132449
64.248496993988 0.00117488491274544
64.5290581162325 0.00125591154561791
64.809619238477 0.00134132127585304
65.0901803607214 0.00143126475602836
65.3707414829659 0.00152589225964668
65.6513026052104 0.00162535328978894
65.9318637274549 0.00172979617298933
66.2124248496994 0.00183936763925574
66.4929859719439 0.00195421238924854
66.7735470941884 0.00207447264971738
67.0541082164329 0.00220028771838523
67.3346693386774 0.00233179349954741
67.6152304609219 0.0024691220317338
67.8957915831663 0.00261240100885715
68.1763527054108 0.00276175329633684
68.4569138276553 0.0029172964437551
68.7374749498998 0.00307914219565783
69.0180360721443 0.00324739600216486
69.2985971943888 0.00342215653110093
69.5791583166333 0.00360351518339622
69.8597194388778 0.00379155561353437
70.1402805611222 0.00398635325685426
70.4208416833667 0.00418797486552173
70.7014028056112 0.00439647805500034
70.9819639278557 0.00461191086284474
71.2625250501002 0.00483431132163625
71.5430861723447 0.00506370704786043
71.8236472945892 0.00530011484850042
72.1042084168337 0.00554354034708618
72.3847695390782 0.0057939776309021
72.6653306613226 0.00605140892099457
72.9458917835671 0.00631580426657817
73.2264529058116 0.00658712126535298
73.5070140280561 0.00686530481119577
73.7875751503006 0.00715028687058029
74.0681362725451 0.00744198628902151
74.3486973947896 0.00774030862871626
74.6292585170341 0.00804514603848053
74.9098196392786 0.00835637715694969
75.190380761523 0.00867386704992002
75.4709418837675 0.00899746718257487
75.751503006012 0.00932701542723554
76.0320641282565 0.00966233610713772
76.312625250501 0.010003240076615
76.5931863727455 0.0103495248379419
76.87374749499 0.0107009746949528
77.1543086172345 0.0110573609434309
77.434869739479 0.0114184420981055
77.7154308617234 0.0117839641559967
77.9959919839679 0.0121536608956751
78.2765531062124 0.0125272542119086
78.5571142284569 0.012904454485013
78.8376753507014 0.0132849609841043
79.1182364729459 0.0136684623033353
79.3987975951904 0.014054636830058
79.6793587174349 0.0144431532437561
79.9599198396794 0.0148336710444762
80.2404809619239 0.0152258411093593
80.5210420841683 0.0156193062758062
80.8016032064128 0.0160137019496867
81.0821643286573 0.0164086567369243
81.3627254509018 0.016803793096709
81.6432865731463 0.0171987280145057
81.9238476953908 0.0175930736929769
82.2044088176353 0.01798643825883
82.4849699398798 0.0183784264836374
82.7655310621243 0.0187686405165112
83.0460921843687 0.019156680626595
83.3266533066132 0.019542145953226
83.6072144288577 0.0199246352616399
83.8877755511022 0.020303747702064
84.1683366733467 0.0206790835700388
84.4488977955912 0.0210502450658355
84.7294589178357 0.021416837050812
85.0100200400802 0.0217784677986011
85.2905811623247 0.0221347497390357
85.5711422845691 0.0224853001927596
85.8517034068136 0.0228297420945226
86.1322645290581 0.0231677047031683
86.4128256513026 0.0234988242964558
86.6933867735471 0.0238227448488222
86.9739478957916 0.0241391186903603
87.2545090180361 0.0244476071452746
87.5350701402806 0.0247478811482242
87.815631262525 0.0250396218370045
88.0961923847695 0.025322521120126
88.376753507014 0.0255962822179433
88.6573146292585 0.0258606201760701
88.937875751503 0.0261152623499474
89.2184368737475 0.0263599488594709
89.498997995992 0.0265944330127951
89.7795591182365 0.0268184816983952
90.060120240481 0.0270318757447336
90.3406813627254 0.027234410246842
90.6212424849699 0.0274258948593562
90.9018036072144 0.0276061540555494
91.1823647294589 0.0277750273521231
91.4629258517034 0.0279323694995132
91.7434869739479 0.0280780506376774
92.0240480961924 0.0282119564173395
92.3046092184369 0.0283339880868765
92.5851703406814 0.0284440625450137
92.8657314629259 0.0285421123596978
93.1462925851703 0.0286280857535488
93.4268537074148 0.0287019465564096
93.7074148296593 0.0287636741255854
93.9879759519038 0.0288132632344746
94.2685370741483 0.0288507239303331
94.5490981963928 0.0288760813620258
94.8296593186373 0.0288893755786824
95.1102204408818 0.0288906613002291
95.3907815631263 0.028880007660829
95.6713426853707 0.0288574979263522
95.9519038076152 0.0288232291870053
96.2324649298597 0.0287773120263369
96.5130260521042 0.0287198701678388
96.7935871743487 0.0286510401004463
97.0741482965932 0.0285709706842297
97.3547094188377 0.0284798227376202
97.6352705410822 0.0283777686075359
97.9158316633267 0.0282649917237762
98.1963927855711 0.0281416861390819
98.4769539078156 0.0280080560562608
98.7575150300601 0.0278643153437685
99.0380761523046 0.0277106870411533
99.3186372745491 0.0275474028557585
99.5991983967936 0.0273747026520607
99.8797595190381 0.0271928339350214
100.160320641283 0.0270020513287866
100.440881763527 0.0268026160520892
100.721442885772 0.0265947953916294
101.002004008016 0.0263788621747384
101.282565130261 0.0261550942425466
101.563126252505 0.0259237739248642
101.84368737475 0.0256851875179692
102.124248496994 0.0254396247664135
102.404809619238 0.0251873783499394
102.685370741483 0.0249287433765602
102.965931863727 0.0246640168827986
103.246492985972 0.0243934973420305
103.527054108216 0.0241174841818507
103.807615230461 0.0238362773112977
104.088176352705 0.0235501766587522
104.36873747495 0.0232594817212447
104.649298597194 0.0229644911258868
104.929859719439 0.0226655022040562
105.210420841683 0.0223628105789313
105.490981963928 0.0220567097669187
105.771543086172 0.0217474907934545
106.052104208417 0.0214354418236112
106.332665330661 0.0211208478078953
106.613226452906 0.0208039901435624
106.89378757515 0.0204851463517409
107.174348697395 0.0201645897705716
107.454909819639 0.0198425892645682
107.735470941884 0.019519408950322
108.016032064128 0.0191953079386397
108.296593186373 0.0188705400931634
108.577154308617 0.0185453538054696
108.857715430862 0.0182199917866057
109.138276553106 0.0178946908749942
109.418837675351 0.0175696818605725
109.699398797595 0.0172451893250222
109.97995991984 0.0169214314978999
110.260521042084 0.0165986201284434
110.541082164329 0.0162769603728004
110.821643286573 0.0159566506964046
111.102204408818 0.015637882791186
111.382765531062 0.0153208415072794
111.663326653307 0.0150057047988842
111.943887775551 0.0146926436838805
112.224448897796 0.0143818222168278
112.50501002004 0.0140733974749051
112.785571142285 0.0137675195563863
113.066132264529 0.013464331591186
113.346693386774 0.0131639697630347
113.627254509018 0.0128665633428068
113.907815631263 0.0125722347325296
114.188376753507 0.0122810995195915
114.468937875752 0.0119932665406635
114.749498997996 0.0117088379548319
115.03006012024 0.0114279093254659
115.310621242485 0.011150569710308
115.591182364729 0.0108769017593046
115.871743486974 0.0106069818196789
116.152304609218 0.0103408800477578
116.432865731463 0.0100786605270666
116.713426853707 0.00982038139221587
116.993987975952 0.00956609495809779
117.274549098196 0.00931584785393908
117.555110220441 0.00906968116173649
117.835671342685 0.00882763055864111
118.11623246493 0.00858972646284115
118.396793587174 0.00835599418252224
118.677354709419 0.00812645406748368
118.957915831663 0.00790112166300592
119.238476953908 0.00768000786557855
119.519038076152 0.00746311908010598
119.799599198397 0.00725045737822422
120.080160320641 0.00704202065737727
120.360721442886 0.00683780280031089
120.64128256513 0.00663779383465828
120.921843687375 0.00644198009230977
121.202404809619 0.0062503443682641
121.482965931864 0.00606286607868454
121.763527054108 0.00587952141788635
122.044088176353 0.00570028351400786
122.324649298597 0.00552512258311919
122.605210420842 0.00535400608155096
122.885771543086 0.00518689885622843
123.166332665331 0.00502376329281508
123.446893787575 0.00486455946148518
123.72745490982 0.0047092452601528
124.008016032064 0.00455777655500653
124.288577154309 0.00441010731820577
124.569138276553 0.00426618976260673
124.849699398798 0.00412597447340659
125.130260521042 0.00398941053659659
125.410821643287 0.00385644566413486
125.691382765531 0.00372702631575494
125.971943887776 0.00360109781734262
126.25250501002 0.00347860447582175
126.533066132265 0.00335948969049617
126.813627254509 0.00324369606081559
127.094188376754 0.00313116549052957
127.374749498998 0.00302183928821376
127.655310621242 0.00291565826415407
127.935871743487 0.00281256282358825
128.216432865731 0.00271249305630577
128.496993987976 0.00261538882262182
128.77755511022 0.00252118983574198
129.058116232465 0.00242983574054394
129.338677354709 0.00234126618880938
129.619238476954 0.00225542091094023
129.899799599198 0.00217223978420775
130.180360721443 0.00209166289757673
130.460921843687 0.00201363061316223
130.741482965932 0.00193808362437376
131.022044088176 0.00186496301080667
131.302605210421 0.00179421028994914
131.583166332665 0.00172576746576768
131.86372745491 0.00165957707424572
132.144288577154 0.0015955822259457
132.424849699399 0.00153372664567228
132.705410821643 0.00147395470931164
132.985971943888 0.00141621147792798
133.266533066132 0.00136044272919638
133.547094188377 0.00130659498625431
133.827655310621 0.00125461554405347
134.108216432866 0.00120445249329545
134.38877755511 0.00115605474203485
134.669338677355 0.00110937203503292
134.949899799599 0.00106435497094609
135.230460921844 0.00102095501743211
135.511022044088 0.000979124524257858
135.791583166333 0.000938816734490092
136.072144288577 0.000899985793852206
136.352705410822 0.000862586758327359
136.633266533066 0.000826575600088242
136.913827655311 0.000791909211831899
137.194388777555 0.00075854540959846
137.4749498998 0.000726442934148675
137.755511022044 0.000695561450976522
138.036072144289 0.00066586154902963
138.316633266533 0.000637304738209397
138.597194388778 0.000609853445721672
138.877755511022 0.000583471011345672
139.158316633267 0.00055812168168839
139.438877755511 0.000533770603489631
139.719438877756 0.000510383816040083
140 0.000487928242774498
};
\end{axis}

\end{tikzpicture}

%% file: simulation/notebooks/figures/bgmres_gaussian_uq__seed_12345__eq_type_spd_reverse__iteration_5.tikz
\begin{tikzpicture}

\begin{axis}[
xlabel={$z$},
ylabel={$p(z)$},
xmin=0, xmax=140,
ymin=-0.01, ymax=0.3,
width=\figwidth,
height=\figheight,
xtick={0,20,40,60,80,100,120,140},
xticklabels={$0$,$20$,$40$,$60$,$80$,$100$,$120$,$140$},
ytick={-0.05,0,0.05,0.1,0.15,0.2,0.25,0.3},
yticklabels={,$0.00$,$0.05$,$0.10$,$0.15$,$0.20$,$0.25$,$0.30$},
tick align=outside,
tick pos=left,
x grid style={white!69.01960784313725!black},
y grid style={white!69.01960784313725!black},
legend cell align={left},
legend entries={{BayesGMRES},{$\chi^2_{95}$}},
legend style={draw=white!80.0!black}
]
\addlegendimage{no markers, blue!50.98039215686274!black}
\addlegendimage{no markers, black}
\addplot [semithick, blue!50.98039215686274!black]
table {%
0 4.70384615714458e-05
0.280561122244489 9.06141218948545e-05
0.561122244488978 0.000166718564956341
0.841683366733467 0.000293239828052877
1.12224448897796 0.000493610962316949
1.40280561122244 0.000796195188792904
1.68336673346693 0.00123245063235257
1.96392785571142 0.00183394264565527
2.24448897795591 0.00262872994934568
2.5250501002004 0.00363808169925298
2.80561122244489 0.00487466732834406
3.08617234468938 0.00634310907762692
3.36673346693387 0.00804304624416538
3.64729458917836 0.00997379657031468
3.92785571142285 0.0121386998755139
4.20841683366733 0.0145467924825279
4.48897795591182 0.0172100010813901
4.76953907815631 0.0201356451779961
5.0501002004008 0.0233162787436846
5.33066132264529 0.0267208843458685
5.61122244488978 0.030292079826427
5.89178356713427 0.033952582665295
6.17234468937876 0.0376208454234625
6.45290581162325 0.0412317202938461
6.73346693386774 0.0447550807083743
7.01402805611222 0.0482051479758989
7.29458917835671 0.0516363065881661
7.5751503006012 0.0551263248395564
7.85571142284569 0.0587528491269945
8.13627254509018 0.0625714883920073
8.41683366733467 0.0666025933985626
8.69739478957916 0.0708295871616201
8.97795591182365 0.0752064899629075
9.25851703406814 0.0796685568094694
9.53907815631263 0.0841393411619998
9.81963927855711 0.0885301325897745
10.1002004008016 0.0927322643865398
10.3807615230461 0.0966071852402796
10.6613226452906 0.0999815301041887
10.9418837675351 0.102653670582477
11.2224448897796 0.104414529915868
11.503006012024 0.105080090781956
11.7835671342685 0.104527990413307
12.064128256513 0.102727889631669
12.3446893787575 0.0997561487810825
12.625250501002 0.0957896142435322
12.9058116232465 0.0910794910778288
13.186372745491 0.0859120652870255
13.4669338677355 0.0805664127668435
13.74749498998 0.0752792617913442
14.0280561122244 0.0702242568498336
14.3086172344689 0.0655082935060725
14.5891783567134 0.0611828439513464
14.8697394789579 0.0572644519742574
15.1503006012024 0.0537565910822466
15.4308617234469 0.050665244895052
15.7114228456914 0.0480029389652834
15.9919839679359 0.0457800906962989
16.2725450901804 0.0439873769187622
16.5531062124249 0.0425767660525878
16.8336673346693 0.0414503303523505
17.1142284569138 0.0404640324227781
17.3947895791583 0.0394486838349013
17.6753507014028 0.0382438185884994
17.9559118236473 0.0367347248527676
18.2364729458918 0.0348806205760978
18.5170340681363 0.0327241120569262
18.7975951903808 0.0303780817103152
19.0781563126253 0.0279938184309784
19.3587174348697 0.0257205448794935
19.6392785571142 0.023669010972046
19.9198396793587 0.0218896264973403
20.2004008016032 0.0203698240982441
20.4809619238477 0.01904854434103
20.7615230460922 0.0178406730344889
21.0420841683367 0.016662611770232
21.3226452905812 0.0154518954497904
21.6032064128257 0.0141774391012872
21.8837675350701 0.0128406759095085
22.1643286573146 0.0114701320313285
22.4448897795591 0.0101125048093627
22.7254509018036 0.0088226179771847
23.0060120240481 0.00765362861126399
23.2865731462926 0.0066482421435949
23.5671342685371 0.00583162604876462
23.8476953907816 0.00520692762830125
24.1282565130261 0.00475436608912886
24.4088176352705 0.00443448587123059
24.689378757515 0.00419531416786347
24.9699398797595 0.00398212005251235
25.250501002004 0.00374762737800224
25.5310621242485 0.00346025700278798
25.811623246493 0.00310844075290054
26.0921843687375 0.00270015183334275
26.372745490982 0.00225816693584374
26.6533066132265 0.00181272896826922
26.9338677354709 0.00139381829584552
27.2144288577154 0.00102502454954055
27.4949899799599 0.000720205362052882
27.7755511022044 0.000483094445314886
28.0561122244489 0.000309173143016016
28.3366733466934 0.000188695150998182
28.6172344689379 0.000109784761892922
28.8977955911824 6.08703079936286e-05
29.1783567134269 3.21536363619825e-05
29.4589178356713 1.61774192318825e-05
29.7394789579158 7.75081373056606e-06
30.0200400801603 3.53555275387157e-06
30.3006012024048 1.5351804312187e-06
30.5811623246493 6.34424092295093e-07
30.8617234468938 2.49487740276065e-07
31.1422845691383 9.33477204213077e-08
31.4228456913828 3.32263693762071e-08
31.7034068136273 1.12494191821437e-08
31.9839679358717 3.62236417921805e-09
32.2645290581162 1.10922290063609e-09
32.5450901803607 3.22970682725828e-10
32.8256513026052 8.94088591563144e-11
33.1062124248497 2.35304057109802e-11
33.3867735470942 5.88667949852869e-12
33.6673346693387 1.39980172989299e-12
33.9478957915832 3.16361105381721e-13
34.2284569138277 6.79496833954915e-14
34.5090180360721 1.38691215576182e-14
34.7895791583166 2.68991712504842e-15
35.0701402805611 4.9571302941249e-16
35.3507014028056 8.67956621566365e-17
35.6312625250501 1.44383755596487e-17
35.9118236472946 2.28175585334366e-18
36.1923847695391 3.42554657887732e-19
36.4729458917836 4.8851864593156e-20
36.7535070140281 6.61766771323315e-21
37.0340681362725 8.5149970479025e-22
37.314629258517 1.04064860360965e-22
37.5951903807615 1.20795056993699e-23
37.875751503006 1.33169899333703e-24
38.1563126252505 1.39431908101773e-25
38.436873747495 1.3864546958725e-26
38.7174348697395 1.30925927018444e-27
38.997995991984 1.1741184614143e-28
39.2785571142285 9.99896403170798e-30
39.5591182364729 8.08623105308084e-31
39.8396793587174 6.20979394252698e-32
40.1202404809619 4.52835612533425e-33
40.4008016032064 3.13566232191499e-34
40.6813627254509 2.06175391580257e-35
40.9619238476954 1.28723301034349e-36
41.2424849699399 7.63105746102461e-38
41.5230460921844 4.2955059013136e-39
41.8036072144289 2.2958376056428e-40
42.0841683366734 1.16509361583941e-41
42.3647294589178 5.61395832095551e-43
42.6452905811623 2.56839873832934e-44
42.9258517034068 1.11567344429699e-45
43.2064128256513 4.60139878016256e-47
43.4869739478958 1.80184640606862e-48
43.7675350701403 6.69912367083215e-50
44.0480961923848 2.3647649483212e-51
44.3286573146293 7.92547569737835e-53
44.6092184368737 2.52190005172221e-54
44.8897795591182 7.61892221177462e-56
45.1703406813627 2.18534728860934e-57
45.4509018036072 5.95123230501045e-59
45.7314629258517 1.53869042364217e-60
46.0120240480962 3.77704379850463e-62
46.2925851703407 8.80253661974027e-64
46.5731462925852 1.9476792657749e-65
46.8537074148297 4.0914726299274e-67
47.1342685370741 8.16006062930593e-69
47.4148296593186 1.54510374926774e-70
47.6953907815631 2.77761154372259e-72
47.9759519038076 4.74060879576293e-74
48.2565130260521 7.68147614793241e-76
48.5370741482966 1.18168707415803e-77
48.8176352705411 1.72586507819998e-79
49.0981963927856 2.39307953700317e-81
49.3787575150301 3.15030633873378e-83
49.6593186372746 3.93725254658559e-85
49.939879759519 4.67173392678788e-87
50.2204408817635 5.26267823053668e-89
50.501002004008 5.62832291796641e-91
50.7815631262525 5.71471112770872e-93
51.062124248497 5.50874016149656e-95
51.3426853707415 5.04141784920829e-97
51.623246492986 4.38021268713964e-99
51.9038076152305 3.61309615522021e-101
52.1843687374749 2.82947198335809e-103
52.4649298597194 2.10364568976042e-105
52.7454909819639 1.48484420724509e-107
53.0260521042084 9.9501602804391e-110
53.3066132264529 6.33023698535943e-112
53.5871743486974 3.82340578039559e-114
53.8677354709419 2.19240703021117e-116
54.1482965931864 1.19352687195565e-118
54.4288577154309 6.16855301766901e-121
54.7094188376754 3.02673473209089e-123
54.9899799599198 1.40995502730736e-125
55.2705410821643 6.23556619528539e-128
55.5511022044088 2.61809958246157e-130
55.8316633266533 1.04360456652288e-132
56.1122244488978 3.9493464852559e-135
56.3927855711423 1.41890683842298e-137
56.6733466933868 4.83973778653842e-140
56.9539078156313 1.56721659735594e-142
57.2344689378758 4.81809533912611e-145
57.5150300601202 1.40624462923055e-147
57.7955911823647 3.89659646030453e-150
58.0761523046092 1.02505933600466e-152
58.3567134268537 2.56006830685474e-155
58.6372745490982 6.07006069066233e-158
58.9178356713427 1.36638601273305e-160
59.1983967935872 2.92006581907864e-163
59.4789579158317 5.92448630336954e-166
59.7595190380761 1.14116266554966e-168
60.0400801603206 2.08681129055843e-171
60.3206412825651 3.62290997994912e-174
60.6012024048096 5.97132386944727e-177
60.8817635270541 9.3437757573867e-180
61.1623246492986 1.38807478082606e-182
61.4428857715431 1.95768123517324e-185
61.7234468937876 2.62125778017847e-188
62.0040080160321 3.3320853050784e-191
62.2845691382766 4.02125023835767e-194
62.565130260521 4.60728082064988e-197
62.8456913827655 5.011489810032e-200
63.12625250501 5.17520539737157e-203
63.4068136272545 5.07372461597576e-206
63.687374749499 4.72242145307105e-209
63.9679358717435 4.17293035678855e-212
64.248496993988 3.50070942493242e-215
64.5290581162325 2.78810751259761e-218
64.809619238477 2.10814989920236e-221
65.0901803607214 1.51332417955963e-224
65.3707414829659 1.0313378583256e-227
65.6513026052104 6.67280523762498e-231
65.9318637274549 4.09877834299316e-234
66.2124248496994 2.39022567170109e-237
66.4929859719439 1.32331082673874e-240
66.7735470941884 6.95541914228234e-244
67.0541082164329 3.47074913671815e-247
67.3346693386774 1.6442263975107e-250
67.6152304609219 7.39500512309365e-254
67.8957915831663 3.15757598069018e-257
68.1763527054108 1.27999303159546e-260
68.4569138276553 4.9260618825024e-264
68.7374749498998 1.79982624853354e-267
69.0180360721443 6.24309231930757e-271
69.2985971943888 2.05592557357799e-274
69.5791583166333 6.42766855868461e-278
69.8597194388778 1.90782270154545e-281
70.1402805611222 5.37602154766402e-285
70.4208416833667 1.43821044481795e-288
70.7014028056112 3.65276997156461e-292
70.9819639278557 8.80766232600584e-296
71.2625250501002 2.01621774862645e-299
71.5430861723447 4.38180105203975e-303
71.8236472945892 9.04078872821171e-307
72.1042084168337 1.77091790134904e-310
72.3847695390782 3.29328243044774e-314
72.6653306613226 5.81430779929708e-318
72.9458917835671 9.73309322307256e-322
73.2264529058116 0
73.5070140280561 0
73.7875751503006 0
74.0681362725451 0
74.3486973947896 0
74.6292585170341 0
74.9098196392786 0
75.190380761523 0
75.4709418837675 0
75.751503006012 0
76.0320641282565 0
76.312625250501 0
76.5931863727455 0
76.87374749499 0
77.1543086172345 0
77.434869739479 0
77.7154308617234 0
77.9959919839679 0
78.2765531062124 0
78.5571142284569 0
78.8376753507014 0
79.1182364729459 0
79.3987975951904 0
79.6793587174349 0
79.9599198396794 0
80.2404809619239 0
80.5210420841683 0
80.8016032064128 0
81.0821643286573 0
81.3627254509018 0
81.6432865731463 0
81.9238476953908 0
82.2044088176353 0
82.4849699398798 0
82.7655310621243 0
83.0460921843687 0
83.3266533066132 0
83.6072144288577 0
83.8877755511022 0
84.1683366733467 0
84.4488977955912 0
84.7294589178357 0
85.0100200400802 0
85.2905811623247 0
85.5711422845691 0
85.8517034068136 0
86.1322645290581 0
86.4128256513026 0
86.6933867735471 0
86.9739478957916 0
87.2545090180361 0
87.5350701402806 0
87.815631262525 0
88.0961923847695 0
88.376753507014 0
88.6573146292585 0
88.937875751503 0
89.2184368737475 0
89.498997995992 0
89.7795591182365 0
90.060120240481 0
90.3406813627254 0
90.6212424849699 0
90.9018036072144 0
91.1823647294589 0
91.4629258517034 0
91.7434869739479 0
92.0240480961924 0
92.3046092184369 0
92.5851703406814 0
92.8657314629259 0
93.1462925851703 0
93.4268537074148 0
93.7074148296593 0
93.9879759519038 0
94.2685370741483 0
94.5490981963928 0
94.8296593186373 0
95.1102204408818 0
95.3907815631263 0
95.6713426853707 0
95.9519038076152 0
96.2324649298597 0
96.5130260521042 0
96.7935871743487 0
97.0741482965932 0
97.3547094188377 0
97.6352705410822 0
97.9158316633267 0
98.1963927855711 0
98.4769539078156 0
98.7575150300601 0
99.0380761523046 0
99.3186372745491 0
99.5991983967936 0
99.8797595190381 0
100.160320641283 0
100.440881763527 0
100.721442885772 0
101.002004008016 0
101.282565130261 0
101.563126252505 0
101.84368737475 0
102.124248496994 0
102.404809619238 0
102.685370741483 0
102.965931863727 0
103.246492985972 0
103.527054108216 0
103.807615230461 0
104.088176352705 0
104.36873747495 0
104.649298597194 0
104.929859719439 0
105.210420841683 0
105.490981963928 0
105.771543086172 0
106.052104208417 0
106.332665330661 0
106.613226452906 0
106.89378757515 0
107.174348697395 0
107.454909819639 0
107.735470941884 0
108.016032064128 0
108.296593186373 0
108.577154308617 0
108.857715430862 0
109.138276553106 0
109.418837675351 0
109.699398797595 0
109.97995991984 0
110.260521042084 0
110.541082164329 0
110.821643286573 0
111.102204408818 0
111.382765531062 0
111.663326653307 0
111.943887775551 0
112.224448897796 0
112.50501002004 0
112.785571142285 0
113.066132264529 0
113.346693386774 0
113.627254509018 0
113.907815631263 0
114.188376753507 0
114.468937875752 0
114.749498997996 0
115.03006012024 0
115.310621242485 0
115.591182364729 0
115.871743486974 0
116.152304609218 0
116.432865731463 0
116.713426853707 0
116.993987975952 0
117.274549098196 0
117.555110220441 0
117.835671342685 0
118.11623246493 0
118.396793587174 0
118.677354709419 0
118.957915831663 0
119.238476953908 0
119.519038076152 0
119.799599198397 0
120.080160320641 0
120.360721442886 0
120.64128256513 0
120.921843687375 0
121.202404809619 0
121.482965931864 0
121.763527054108 0
122.044088176353 0
122.324649298597 0
122.605210420842 0
122.885771543086 0
123.166332665331 0
123.446893787575 0
123.72745490982 0
124.008016032064 0
124.288577154309 0
124.569138276553 0
124.849699398798 0
125.130260521042 0
125.410821643287 0
125.691382765531 0
125.971943887776 0
126.25250501002 0
126.533066132265 0
126.813627254509 0
127.094188376754 0
127.374749498998 0
127.655310621242 0
127.935871743487 0
128.216432865731 0
128.496993987976 0
128.77755511022 0
129.058116232465 0
129.338677354709 0
129.619238476954 0
129.899799599198 0
130.180360721443 0
130.460921843687 0
130.741482965932 0
131.022044088176 0
131.302605210421 0
131.583166332665 0
131.86372745491 0
132.144288577154 0
132.424849699399 0
132.705410821643 0
132.985971943888 0
133.266533066132 0
133.547094188377 0
133.827655310621 0
134.108216432866 0
134.38877755511 0
134.669338677355 0
134.949899799599 0
135.230460921844 0
135.511022044088 0
135.791583166333 0
136.072144288577 0
136.352705410822 0
136.633266533066 0
136.913827655311 0
137.194388777555 0
137.4749498998 0
137.755511022044 0
138.036072144289 0
138.316633266533 0
138.597194388778 0
138.877755511022 0
139.158316633267 0
139.438877755511 0
139.719438877756 0
140 0
};
\addplot [very thick, black, dashed]
table {%
0 0
0.280561122244489 2.50015431361688e-99
0.561122244488978 2.16241248234828e-85
0.841683366733467 2.89907207128255e-77
1.12224448897796 1.62550081198465e-71
1.40280561122244 4.53293740579673e-67
1.68336673346693 1.89402012643235e-63
1.96392785571142 2.13543931355579e-60
2.24448897795591 9.22974483465739e-58
2.5250501002004 1.91804554224801e-55
2.80561122244489 2.23696481071162e-53
3.08617234468938 1.63491786171413e-51
3.36673346693387 8.12345812929295e-50
3.64729458917836 2.91914598092005e-48
3.92785571142285 7.9601365341925e-47
4.20841683366733 1.71114704002281e-45
4.48897795591182 2.99019770730069e-44
4.76953907815631 4.35584564063032e-43
5.0501002004008 5.40065022224857e-42
5.33066132264529 5.79954515427439e-41
5.61122244488978 5.47423219605723e-40
5.89178356713427 4.5994144755416e-39
6.17234468937876 3.47725809806788e-38
6.45290581162325 2.3878309123534e-37
6.73346693386774 1.50161591084084e-36
7.01402805611222 8.71010622226947e-36
7.29458917835671 4.68976040003146e-35
7.5751503006012 2.35712131053892e-34
7.85571142284569 1.11145602574311e-33
8.13627254509018 4.93887069471486e-33
8.41683366733467 2.07651955131835e-32
8.69739478957916 8.29074141249149e-32
8.97795591182365 3.15373737318434e-31
9.25851703406814 1.14637488674885e-30
9.53907815631263 3.99277675107528e-30
9.81963927855711 1.33581517961625e-29
10.1002004008016 4.30254601832215e-29
10.3807615230461 1.33695359657634e-28
10.6613226452906 4.01559896701669e-28
10.9418837675351 1.16786930556864e-27
11.2224448897796 3.29424989435982e-27
11.503006012024 9.0259559108091e-27
11.7835671342685 2.40554120382048e-26
12.064128256513 6.24425898659647e-26
12.3446893787575 1.58060774655615e-25
12.625250501002 3.90601103971221e-25
12.9058116232465 9.43338751966316e-25
13.186372745491 2.22872160270054e-24
13.4669338677355 5.15584547832087e-24
13.74749498998 1.1688990421024e-23
14.0280561122244 2.5992102406461e-23
14.3086172344689 5.67317685905931e-23
14.5891783567134 1.21631353510319e-22
14.8697394789579 2.56327335675657e-22
15.1503006012024 5.31317040746073e-22
15.4308617234469 1.08389172635141e-21
15.7114228456914 2.17741405489333e-21
15.9919839679359 4.30978376938675e-21
16.2725450901804 8.4091736263694e-21
16.5531062124249 1.61825548097134e-20
16.8336673346693 3.07283056821909e-20
17.1142284569138 5.75996342931666e-20
17.3947895791583 1.06628398428377e-19
17.6753507014028 1.95016698748993e-19
17.9559118236473 3.52518543803627e-19
18.2364729458918 6.30029938826875e-19
18.5170340681363 1.11367900557664e-18
18.7975951903808 1.9477009762972e-18
19.0781563126253 3.37120625131999e-18
19.3587174348697 5.77670950253487e-18
19.6392785571142 9.8024265361668e-18
19.9198396793587 1.6476496761994e-17
20.2004008016032 2.74403516350904e-17
20.4809619238477 4.52916883488616e-17
20.7615230460922 7.41066887654386e-17
21.0420841683367 1.20228690017582e-16
21.3226452905812 1.93449822199758e-16
21.6032064128257 3.08767755445346e-16
21.8837675350701 4.88977833191876e-16
22.1643286573146 7.68469823212148e-16
22.4448897795591 1.19874937710424e-15
22.7254509018036 1.85641172880961e-15
23.0060120240481 2.85457848378643e-15
23.2865731462926 4.35919276437546e-15
23.5671342685371 6.61208620137226e-15
23.8476953907816 9.96342548283183e-15
24.1282565130261 1.49170704403536e-14
24.4088176352705 2.21935987974252e-14
24.689378757515 3.28173618327755e-14
24.9699398797595 4.82360399134629e-14
25.250501002004 7.04839041902624e-14
25.5310621242485 1.02403531155131e-13
25.811623246493 1.47945270961249e-13
26.0921843687375 2.12569597480166e-13
26.372745490982 3.03784884789074e-13
26.6533066132265 4.31862546342001e-13
26.9338677354709 6.10783305361738e-13
27.2144288577154 8.59483198456291e-13
27.4949899799599 1.20348641980295e-12
27.7755511022044 1.67703508886092e-12
28.0561122244489 2.32585441598654e-12
28.3366733466934 3.21072666151903e-12
28.6172344689379 4.41209004063823e-12
28.8977955911824 6.03593084134183e-12
29.1783567134269 8.22130026309396e-12
29.4589178356713 1.11498638887586e-11
29.7394789579158 1.50579836151209e-11
30.0200400801603 2.025194196488e-11
30.3006012024048 2.71270500295221e-11
30.5811623246493 3.61915362912715e-11
30.8617234468938 4.80962981185026e-11
31.1422845691383 6.36718151609034e-11
31.4228456913828 8.39737790197992e-11
31.7034068136273 1.10339291787429e-10
31.9839679358717 1.44455833232985e-10
32.2645290581162 1.88445599023277e-10
32.5450901803607 2.44968277439284e-10
32.8256513026052 3.17345867208591e-10
33.1062124248497 4.09713752520904e-10
33.3867735470942 5.27202951804727e-10
33.6673346693387 6.76159253739092e-10
33.9478957915832 8.6440585709085e-10
34.2284569138277 1.10155715060316e-09
34.5090180360721 1.39939241600834e-09
34.7895791583166 1.77229952231386e-09
35.0701402805611 2.23780011403617e-09
35.3507014028056 2.81716939136123e-09
35.6312625250501 3.53616534757288e-09
35.9118236472946 4.42588427958313e-09
36.1923847695391 5.52376153172112e-09
36.4729458917836 6.87473878118583e-09
36.7535070140281 8.53262173502624e-09
37.0340681362725 1.05616548914364e-08
37.314629258517 1.30383430294006e-08
37.5951903807615 1.60535523356572e-08
37.875751503006 1.97149275598982e-08
38.1563126252505 2.4149665309321e-08
38.436873747495 2.95076875506595e-08
38.7174348697395 3.59652635774838e-08
38.997995991984 4.37291331153251e-08
39.2785571142285 5.30411878661076e-08
39.5591182364729 6.41837736215754e-08
39.8396793587174 7.74856800837708e-08
40.1202404809619 9.33288906956561e-08
40.4008016032064 1.12156170078275e-07
40.6813627254509 1.34479572059096e-07
40.9619238476954 1.60889956720771e-07
41.2424849699399 1.92067610356491e-07
41.5230460921844 2.28794067637853e-07
41.8036072144289 2.71965240628658e-07
42.0841683366734 3.22605964452397e-07
42.3647294589178 3.81886074376441e-07
42.6452905811623 4.51138133739997e-07
42.9258517034068 5.318769364494e-07
43.2064128256513 6.25820911611232e-07
43.4869739478958 7.34915561184542e-07
43.7675350701403 8.6135906421896e-07
44.0480961923848 1.0076301832104e-06
44.3286573146293 1.17651860925476e-06
44.6092184368737 1.37115788291493e-06
44.8897795591182 1.5950610269345e-06
45.1703406813627 1.85215902503665e-06
45.4509018036072 2.1468422779362e-06
45.7314629258517 2.48400516327089e-06
46.0120240480962 2.86909382032972e-06
46.2925851703407 3.3081572731283e-06
46.5731462925852 3.80790199647685e-06
46.8537074148297 4.37575001911074e-06
47.1342685370741 5.01990064565932e-06
47.4148296593186 5.74939586513069e-06
47.6953907815631 6.57418949766023e-06
47.9759519038076 7.50522011344319e-06
48.2565130260521 8.55448773806461e-06
48.5370741482966 9.73513433678369e-06
48.8176352705411 1.10615280467909e-05
49.0981963927856 1.25493511010145e-05
49.3787575150301 1.42156913597551e-05
49.6593186372746 1.60791373373625e-05
49.939879759519 1.81598765803602e-05
50.2204408817635 2.04797972210283e-05
50.501002004008 2.3062592496539e-05
50.7815631262525 2.5933867988494e-05
51.062124248497 2.91212513012272e-05
51.3426853707415 3.26545038597868e-05
51.623246492986 3.65656344701917e-05
51.9038076152305 4.08890142456695e-05
52.1843687374749 4.5661492463348e-05
52.4649298597194 5.09225128764874e-05
52.7454909819639 5.67142299681826e-05
53.0260521042084 6.3081624593729e-05
53.3066132264529 7.00726184208763e-05
53.5871743486974 7.77381865403288e-05
53.8677354709419 8.61324675833424e-05
54.1482965931864 9.53128706495126e-05
54.4288577154309 0.000105340178316125
54.7094188376754 0.000116278644971225
54.9899799599198 0.000128196089685869
55.2705410821643 0.00014116398281772
55.5511022044088 0.000155257525517809
55.8316633266533 0.00017055572129626
56.1122244488978 0.000187141438790064
56.3927855711423 0.000205101464868303
56.6733466933868 0.000224526547206556
56.9539078156313 0.000245511425463893
57.2344689378758 0.000268154850201949
57.5150300601202 0.000292559588697996
57.7955911823647 0.000318832416820466
58.0761523046092 0.00034708409615881
58.3567134268537 0.000377429335627634
58.6372745490982 0.00040998673679991
58.9178356713427 0.000444878722263589
59.1983967935872 0.000482231446342574
59.4789579158317 0.000522174687574094
59.7595190380761 0.000564841722392452
60.0400801603206 0.000610369179531983
60.3206412825651 0.000658896874730025
60.6012024048096 0.000710567625385022
60.8817635270541 0.000765527044902651
61.1623246492986 0.000823923316546471
61.4428857715431 0.000885906946696683
61.7234468937876 0.000951630497512198
62.0040080160321 0.00102124829908628
62.2845691382766 0.00109491614128295
62.565130260521 0.00117279094554348
62.8456913827655 0.00125503041705392
63.12625250501 0.00134179267776871
63.4068136272545 0.0014332358808909
63.687374749499 0.00152951780751651
63.9679358717435 0.00163079544625279
64.248496993988 0.00173722455672797
64.5290581162325 0.00184895921801298
64.809619238477 0.00196615136307409
65.0901803607214 0.00208895030047802
65.3707414829659 0.00221750222466444
65.6513026052104 0.00235194971619184
65.9318637274549 0.00249243123345171
66.2124248496994 0.00263908059742498
66.4929859719439 0.00279202647113706
66.7735470941884 0.00295139183552988
67.0541082164329 0.00311729346354023
67.3346693386774 0.0032898413942276
67.6152304609219 0.00346913840884249
67.8957915831663 0.00365527951076967
68.1763527054108 0.00384835141131252
68.4569138276553 0.00404843202330836
68.7374749498998 0.00425558996458208
69.0180360721443 0.00446988407324706
69.2985971943888 0.0046913629368666
69.5791583166333 0.00492006443746837
69.8597194388778 0.005156015314389
70.1402805611222 0.0053992307468907
70.4208416833667 0.00564971395845365
70.7014028056112 0.00590745584459464
70.9819639278557 0.00617243462600672
71.2625250501002 0.00644461552874495
71.5430861723447 0.00672395049310449
71.8236472945892 0.00701037791275586
72.1042084168337 0.00730382240560391
72.3847695390782 0.0076041946177442
72.6653306613226 0.00791139106176938
72.9458917835671 0.00822529399058074
73.2264529058116 0.00854577130771945
73.5070140280561 0.0088726765151238
73.7875751503006 0.00920584869907291
74.0681362725451 0.00954511255495323
74.3486973947896 0.00989027845132915
74.6292585170341 0.0102411425336786
74.9098196392786 0.0105974868679833
75.190380761523 0.0109590796242396
75.4709418837675 0.0113256752997866
75.751503006012 0.0116970149822253
76.0320641282565 0.0120728266515254
76.312625250501 0.012452825520797
76.5931863727455 0.0128367144150351
76.87374749499 0.0132241841870241
77.1543086172345 0.0136149141694376
77.434869739479 0.0140085726620271
77.7154308617234 0.0144048174526828
77.9959919839679 0.0148032963710045
78.2765531062124 0.0152036478729018
78.5571142284569 0.0156055016546441
78.8376753507014 0.0160084792946494
79.1182364729459 0.0164121949212167
79.3987975951904 0.0168162559043142
79.6793587174349 0.0172202635694234
79.9599198396794 0.0176238139314137
80.2404809619239 0.0180264984462836
80.5210420841683 0.0184279047786119
80.8016032064128 0.0188276175824635
81.0821643286573 0.0192252192934727
81.3627254509018 0.0196202909297906
81.6432865731463 0.0200124128995508
81.9238476953908 0.0204011658125137
82.2044088176353 0.0207861312935118
82.4849699398798 0.021166892795356
82.7655310621243 0.0215430364088434
83.0460921843687 0.0219141516675614
83.3266533066132 0.0222798323451823
83.6072144288577 0.0226396772430031
83.8877755511022 0.0229932909655126
84.1683366733467 0.0233402846818499
84.4488977955912 0.0236802768710476
84.7294589178357 0.0240128940490486
85.0100200400802 0.0243377714755474
85.2905811623247 0.0246545538387927
85.5711422845691 0.0249628959165763
85.8517034068136 0.0252624632117381
86.1322645290581 0.0255529325605806
86.4128256513026 0.0258339927127434
86.6933867735471 0.0261053448811477
86.9739478957916 0.0263667032607498
87.2545090180361 0.0266177955149699
87.5350701402806 0.026858363228746
87.815631262525 0.0270881623273215
88.0961923847695 0.0273069634599542
88.376753507014 0.0275145523478828
88.6573146292585 0.0277107300959909
88.937875751503 0.0278953134677619
89.2184368737475 0.028068135123164
89.498997995992 0.0282290438193363
89.7795591182365 0.0283779045739385
90.060120240481 0.0285145987912576
90.3406813627254 0.0286390243511886
90.6212424849699 0.0287510956613849
90.9018036072144 0.0288507436729133
91.1823647294589 0.0289379158599433
91.4629258517034 0.0290125761639894
91.7434869739479 0.0290747049034323
92.0240480961924 0.029124298649041
92.3046092184369 0.0291613700664009
92.5851703406814 0.029185947726113
92.8657314629259 0.0291980758828549
93.1462925851703 0.029197814224333
93.4268537074148 0.0291852375913038
93.7074148296593 0.029160435669875
93.9879759519038 0.029123512657363
94.2685370741483 0.0290745869030077
94.5490981963928 0.0290137905249441
94.8296593186373 0.0289412690048073
95.1102204408818 0.028857180761429
95.3907815631263 0.0287616967050762
95.6713426853707 0.0286549997737488
95.9519038076152 0.0285372844530083
96.2324649298597 0.0284087562808939
96.5130260521042 0.0282696313394191
96.7935871743487 0.0281201357341964
97.0741482965932 0.0279605050637063
97.3547094188377 0.0277909838797228
97.6352705410822 0.0276118251403998
97.9158316633267 0.0274232896575046
98.1963927855711 0.0272256455392493
98.4769539078156 0.0270191676301808
98.7575150300601 0.0268041369495008
99.0380761523046 0.0265808401292171
99.3186372745491 0.0263495688534553
99.5991983967936 0.026110619300221
99.8797595190381 0.0258642915868726
100.160320641283 0.0256108892205117
100.440881763527 0.0253507185544552
100.721442885772 0.0250840882518947
101.002004008016 0.0248113087578077
101.282565130261 0.0245326917801337
101.563126252505 0.0242485497811409
101.84368737475 0.0239591954799167
102.124248496994 0.0236649413667951
102.404809619238 0.0233660992305078
102.685370741483 0.0230629796987874
102.965931863727 0.0227558917930928
103.246492985972 0.0224451424980394
103.527054108216 0.022131036346122
103.807615230461 0.0218138750181878
104.088176352705 0.0214939569601113
104.36873747495 0.0211715770160452
104.649298597194 0.0208470260785658
104.929859719439 0.0205205907559836
105.210420841683 0.020192553057033
105.490981963928 0.0198631900931007
105.771543086172 0.0195327737981001
106.052104208417 0.0192015706660682
106.332665330661 0.0188698415064702
106.613226452906 0.0185378412172103
106.89378757515 0.0182058185752587
107.174348697395 0.0178740160447641
107.454909819639 0.0175426696025151
107.735470941884 0.017212008580544
108.016032064128 0.0168822555256259
108.296593186373 0.016553626075434
108.577154308617 0.0162263288510209
108.857715430862 0.015900565365318
109.138276553106 0.0155765299472845
109.418837675351 0.0152544096813273
109.699398797595 0.0149343843615762
109.97995991984 0.0146166264606043
110.260521042084 0.0143013011121206
110.541082164329 0.0139885661071903
110.821643286573 0.0136785719034913
111.102204408818 0.013371461647117
111.382765531062 0.0130673712064164
111.663326653307 0.0127664292173568
111.943887775551 0.0124687571398922
112.224448897796 0.012174469324799
112.50501002004 0.0118836730904508
112.785571142285 0.0115964688090002
113.066132264529 0.0113129500014214
113.346693386774 0.011033203440889
113.627254509018 0.0107573092639467
113.907815631263 0.010485341088944
114.188376753507 0.0102173661412118
114.468937875752 0.00995344538445657
114.749498997996 0.00969363365785575
115.03006012024 0.00943797981835761
115.310621242485 0.00918652688768114
115.591182364729 0.00893931220353388
115.871743486974 0.0086963675745744
116.152304609218 0.00845771943864674
116.432865731463 0.00822338902384865
116.713426853707 0.00799339251198489
116.993987975952 0.00776774120398495
117.274549098196 0.00754644168687649
117.555110220441 0.00732949600191162
117.835671342685 0.007116901813476
118.11623246493 0.00690865257840166
118.396793587174 0.00670473771534329
118.677354709419 0.00650514277387844
118.957915831663 0.00630984960301225
119.238476953908 0.00611883651878579
119.519038076152 0.00593207847069807
119.799599198397 0.00574954720666988
120.080160320641 0.00557121143629809
120.360721442886 0.00539703699215364
120.64128256513 0.00522698698890343
120.921843687375 0.00506102198004561
121.202404809619 0.0048991001120629
121.482965931864 0.00474117727581734
121.763527054108 0.00458720725501822
122.044088176353 0.00443714187161817
122.324649298597 0.00429093112799427
122.605210420842 0.00414852334579815
122.885771543086 0.00400986530136183
123.166332665331 0.0038749023575644
123.446893787575 0.00374357859207323
123.72745490982 0.00361583692189101
124.008016032064 0.00349161922414493
124.288577154309 0.00337086645307241
124.569138276553 0.00325351875315928
124.849699398798 0.0031395155684084
125.130260521042 0.00302879574771569
125.410821643287 0.00292129764634577
125.691382765531 0.00281695922350714
125.971943887776 0.00271571813603457
126.25250501002 0.00261751182819582
126.533066132265 0.00252227761764294
126.813627254509 0.00242995277754376
127.094188376754 0.00234047461492525
127.374749498998 0.00225378054527647
127.655310621242 0.00216980816345726
127.935871743487 0.00208849531096813
128.216432865731 0.00200978013964021
128.496993987976 0.00193360117180894
128.77755511022 0.00185989735703938
129.058116232465 0.0017886081254734
129.338677354709 0.0017196734378759
129.619238476954 0.00165303383245402
129.899799599198 0.00158863046853393
130.180360721443 0.00152640516717407
130.460921843687 0.00146630044880126
130.741482965932 0.00140825956795589
131.022044088176 0.00135222654523236
131.302605210421 0.00129814619650551
131.583166332665 0.00124596415952967
131.86372745491 0.00119562691800331
132.144288577154 0.00114708182318713
132.424849699399 0.00110027711316732
132.705410821643 0.00105516192985383
132.985971943888 0.00101168633380314
133.266533066132 0.000969801316955759
133.547094188377 0.000929458813376165
133.827655310621 0.000890611708084089
134.108216432866 0.000853213844062656
134.38877755511 0.000817220027530028
134.669338677355 0.000782586031558582
134.949899799599 0.000749268598124454
135.230460921844 0.000717225438668809
135.511022044088 0.000686415233251186
135.791583166333 0.000656797628372252
136.072144288577 0.000628333233542927
136.352705410822 0.000600983616674155
136.633266533066 0.000574711298359976
136.913827655311 0.000549479745124297
137.194388777555 0.00052525336170049
137.4749498998 0.000501997482409887
137.755511022044 0.000479678361704146
138.036072144289 0.000458263163933644
138.316633266533 0.000437719952402453
138.597194388778 0.000418017677768031
138.877755511022 0.000399126165841874
139.158316633267 0.000381016104845038
139.438877755511 0.000363659032170532
139.719438877756 0.000347027320702533
140 0.000331094164739846
};
\end{axis}

\end{tikzpicture}

%% file: simulation/notebooks/figures/bgmres_gaussian_uq__seed_12345__eq_type_spd_reverse__iteration_8.tikz
\begin{tikzpicture}

\begin{axis}[
xlabel={$z$},
ylabel={$p(z)$},
xmin=0, xmax=140,
ymin=-0.01, ymax=0.3,
width=\figwidth,
height=\figheight,
xtick={0,20,40,60,80,100,120,140},
xticklabels={$0$,$20$,$40$,$60$,$80$,$100$,$120$,$140$},
ytick={-0.05,0,0.05,0.1,0.15,0.2,0.25,0.3},
yticklabels={,$0.00$,$0.05$,$0.10$,$0.15$,$0.20$,$0.25$,$0.30$},
tick align=outside,
tick pos=left,
x grid style={white!69.01960784313725!black},
y grid style={white!69.01960784313725!black},
legend cell align={left},
legend entries={{BayesGMRES},{$\chi^2_{92}$}},
legend style={draw=white!80.0!black}
]
\addlegendimage{no markers, blue!50.98039215686274!black}
\addlegendimage{no markers, black}
\addplot [semithick, blue!50.98039215686274!black]
table {%
0 0.0075678087243642
0.280561122244489 0.0138925700519116
0.561122244488978 0.0233299876056187
0.841683366733467 0.0359755154402183
1.12224448897796 0.0512027843916152
1.40280561122244 0.067763804144742
1.68336673346693 0.0842665168974694
1.96392785571142 0.0997906324001106
2.24448897795591 0.114199798600763
2.5250501002004 0.127844851150765
2.80561122244489 0.140816379491532
3.08617234468938 0.152336244966668
3.36673346693387 0.160853324293783
3.64729458917836 0.16487724709283
3.92785571142285 0.163980046988112
4.20841683366733 0.159240730703398
4.48897795591182 0.152808149387889
4.76953907815631 0.146848787728671
5.0501002004008 0.142492726742854
5.33066132264529 0.139349076974153
5.61122244488978 0.135849604668707
5.89178356713427 0.130246861436133
6.17234468937876 0.121692932801473
6.45290581162325 0.110719617868608
6.73346693386774 0.098832186590357
7.01402805611222 0.087583661690456
7.29458917835671 0.077831200557214
7.5751503006012 0.0695992080845549
7.85571142284569 0.0624346154575628
8.13627254509018 0.0558684521829077
8.41683366733467 0.0497041590321972
8.69739478957916 0.0440667772947485
8.97795591182365 0.0392633505367123
9.25851703406814 0.0355481834663283
9.53907815631263 0.0329191248647773
9.81963927855711 0.0310610467212842
10.1002004008016 0.0294655352061666
10.3807615230461 0.0276493777643304
10.6613226452906 0.0253535487561038
10.9418837675351 0.0226332507130999
11.2224448897796 0.0197982303994091
11.503006012024 0.0172220500178404
11.7835671342685 0.0151198400452134
12.064128256513 0.0134384249687261
12.3446893787575 0.0119328842661194
12.625250501002 0.0103590845700854
12.9058116232465 0.00863107482141586
13.186372745491 0.00684436804416645
13.4669338677355 0.0051813434190043
13.74749498998 0.00378861396802459
14.0280561122244 0.00271030297614226
14.3086172344689 0.00190252354278957
14.5891783567134 0.00129309477413655
14.8697394789579 0.000829981749440883
15.1503006012024 0.000489945728961502
15.4308617234469 0.000260424987434963
15.7114228456914 0.00012286234344
15.9919839679359 5.0986649634084e-05
16.2725450901804 1.85126291504285e-05
16.5531062124249 5.86259474611864e-06
16.8336673346693 1.61631561835708e-06
17.1142284569138 3.87535430915547e-07
17.3947895791583 8.07553635673569e-08
17.6753507014028 1.46198608002028e-08
17.9559118236473 2.2989500702312e-09
18.2364729458918 3.13957957727853e-10
18.5170340681363 3.72334603270352e-11
18.7975951903808 3.83436481751777e-12
19.0781563126253 3.42876866246347e-13
19.3587174348697 2.66231468593412e-14
19.6392785571142 1.79494819249249e-15
19.9198396793587 1.0507819216869e-16
20.2004008016032 5.34120722910092e-18
20.4809619238477 2.35738865036448e-19
20.7615230460922 9.03414588358671e-21
21.0420841683367 3.00612251842983e-22
21.3226452905812 8.6853958630153e-24
21.6032064128257 2.1788921332578e-25
21.8837675350701 4.74618756882136e-27
22.1643286573146 8.97670683910038e-29
22.4448897795591 1.47418544159818e-30
22.7254509018036 2.1020841377812e-32
23.0060120240481 2.6026211506279e-34
23.2865731462926 2.79791608858209e-36
23.5671342685371 2.61168791754292e-38
23.8476953907816 2.11675567677186e-40
24.1282565130261 1.48964588906724e-42
24.4088176352705 9.10244798569323e-45
24.689378757515 4.82943385690912e-47
24.9699398797595 2.22483132800832e-49
25.250501002004 8.89940173163336e-52
25.5310621242485 3.09091729126112e-54
25.811623246493 9.32130691077087e-57
26.0921843687375 2.44078250945334e-59
26.372745490982 5.54937672761175e-62
26.6533066132265 1.09552473020473e-64
26.9338677354709 1.87785914724198e-67
27.2144288577154 2.79490270085363e-70
27.4949899799599 3.61187925821704e-73
27.7755511022044 4.05286936390986e-76
28.0561122244489 3.94870634844834e-79
28.3366733466934 3.34048804358066e-82
28.6172344689379 2.45373612250582e-85
28.8977955911824 1.56497946109381e-88
29.1783567134269 8.66667037514378e-92
29.4589178356713 4.16733839200277e-95
29.7394789579158 1.73991529571157e-98
30.0200400801603 6.3075441071215e-102
30.3006012024048 1.98543346273042e-105
30.5811623246493 5.42641781008641e-109
30.8617234468938 1.28775715842741e-112
31.1422845691383 2.65349068300813e-116
31.4228456913828 4.74748949541068e-120
31.7034068136273 7.37519162717056e-124
31.9839679358717 9.94822173041719e-128
32.2645290581162 1.16514628570245e-131
32.5450901803607 1.18489073797386e-135
32.8256513026052 1.04625851475327e-139
33.1062124248497 8.02162872678267e-144
33.3867735470942 5.34009471591679e-148
33.6673346693387 3.08672695164728e-152
33.9478957915832 1.54921009743941e-156
34.2284569138277 6.75126686241443e-161
34.5090180360721 2.55460102517593e-165
34.7895791583166 8.39312594947418e-170
35.0701402805611 2.39434787613027e-174
35.3507014028056 5.93080451962717e-179
35.6312625250501 1.27556566675044e-183
35.9118236472946 2.38207199677913e-188
36.1923847695391 3.86251160382163e-193
36.4729458917836 5.43810466113819e-198
36.7535070140281 6.64795662711226e-203
37.0340681362725 7.05653734115238e-208
37.314629258517 6.50366198829993e-213
37.5951903807615 5.20459728237307e-218
37.875751503006 3.61642239932874e-223
38.1563126252505 2.18189587612531e-228
38.436873747495 1.1430145097148e-233
38.7174348697395 5.19914963089574e-239
38.997995991984 2.05341032160732e-244
39.2785571142285 7.04177291824082e-250
39.5591182364729 2.09677168000123e-255
39.8396793587174 5.42104633766567e-261
40.1202404809619 1.21696460732225e-266
40.4008016032064 2.3721143862446e-272
40.6813627254509 4.01472828672126e-278
40.9619238476954 5.89983056578804e-284
41.2424849699399 7.5281068082877e-290
41.5230460921844 8.34055318256632e-296
41.8036072144289 8.02355420109574e-302
42.0841683366734 6.70195606678884e-308
42.3647294589178 4.86070443250577e-314
42.6452905811623 3.06073667598652e-320
42.9258517034068 0
43.2064128256513 0
43.4869739478958 0
43.7675350701403 0
44.0480961923848 0
44.3286573146293 0
44.6092184368737 0
44.8897795591182 0
45.1703406813627 0
45.4509018036072 0
45.7314629258517 0
46.0120240480962 0
46.2925851703407 0
46.5731462925852 0
46.8537074148297 0
47.1342685370741 0
47.4148296593186 0
47.6953907815631 0
47.9759519038076 0
48.2565130260521 0
48.5370741482966 0
48.8176352705411 0
49.0981963927856 0
49.3787575150301 0
49.6593186372746 0
49.939879759519 0
50.2204408817635 0
50.501002004008 0
50.7815631262525 0
51.062124248497 0
51.3426853707415 0
51.623246492986 0
51.9038076152305 0
52.1843687374749 0
52.4649298597194 0
52.7454909819639 0
53.0260521042084 0
53.3066132264529 0
53.5871743486974 0
53.8677354709419 0
54.1482965931864 0
54.4288577154309 0
54.7094188376754 0
54.9899799599198 0
55.2705410821643 0
55.5511022044088 0
55.8316633266533 0
56.1122244488978 0
56.3927855711423 0
56.6733466933868 0
56.9539078156313 0
57.2344689378758 0
57.5150300601202 0
57.7955911823647 0
58.0761523046092 0
58.3567134268537 0
58.6372745490982 0
58.9178356713427 0
59.1983967935872 0
59.4789579158317 0
59.7595190380761 0
60.0400801603206 0
60.3206412825651 0
60.6012024048096 0
60.8817635270541 0
61.1623246492986 0
61.4428857715431 0
61.7234468937876 0
62.0040080160321 0
62.2845691382766 0
62.565130260521 0
62.8456913827655 0
63.12625250501 0
63.4068136272545 0
63.687374749499 0
63.9679358717435 0
64.248496993988 0
64.5290581162325 0
64.809619238477 0
65.0901803607214 0
65.3707414829659 0
65.6513026052104 0
65.9318637274549 0
66.2124248496994 0
66.4929859719439 0
66.7735470941884 0
67.0541082164329 0
67.3346693386774 0
67.6152304609219 0
67.8957915831663 0
68.1763527054108 0
68.4569138276553 0
68.7374749498998 0
69.0180360721443 0
69.2985971943888 0
69.5791583166333 0
69.8597194388778 0
70.1402805611222 0
70.4208416833667 0
70.7014028056112 0
70.9819639278557 0
71.2625250501002 0
71.5430861723447 0
71.8236472945892 0
72.1042084168337 0
72.3847695390782 0
72.6653306613226 0
72.9458917835671 0
73.2264529058116 0
73.5070140280561 0
73.7875751503006 0
74.0681362725451 0
74.3486973947896 0
74.6292585170341 0
74.9098196392786 0
75.190380761523 0
75.4709418837675 0
75.751503006012 0
76.0320641282565 0
76.312625250501 0
76.5931863727455 0
76.87374749499 0
77.1543086172345 0
77.434869739479 0
77.7154308617234 0
77.9959919839679 0
78.2765531062124 0
78.5571142284569 0
78.8376753507014 0
79.1182364729459 0
79.3987975951904 0
79.6793587174349 0
79.9599198396794 0
80.2404809619239 0
80.5210420841683 0
80.8016032064128 0
81.0821643286573 0
81.3627254509018 0
81.6432865731463 0
81.9238476953908 0
82.2044088176353 0
82.4849699398798 0
82.7655310621243 0
83.0460921843687 0
83.3266533066132 0
83.6072144288577 0
83.8877755511022 0
84.1683366733467 0
84.4488977955912 0
84.7294589178357 0
85.0100200400802 0
85.2905811623247 0
85.5711422845691 0
85.8517034068136 0
86.1322645290581 0
86.4128256513026 0
86.6933867735471 0
86.9739478957916 0
87.2545090180361 0
87.5350701402806 0
87.815631262525 0
88.0961923847695 0
88.376753507014 0
88.6573146292585 0
88.937875751503 0
89.2184368737475 0
89.498997995992 0
89.7795591182365 0
90.060120240481 0
90.3406813627254 0
90.6212424849699 0
90.9018036072144 0
91.1823647294589 0
91.4629258517034 0
91.7434869739479 0
92.0240480961924 0
92.3046092184369 0
92.5851703406814 0
92.8657314629259 0
93.1462925851703 0
93.4268537074148 0
93.7074148296593 0
93.9879759519038 0
94.2685370741483 0
94.5490981963928 0
94.8296593186373 0
95.1102204408818 0
95.3907815631263 0
95.6713426853707 0
95.9519038076152 0
96.2324649298597 0
96.5130260521042 0
96.7935871743487 0
97.0741482965932 0
97.3547094188377 0
97.6352705410822 0
97.9158316633267 0
98.1963927855711 0
98.4769539078156 0
98.7575150300601 0
99.0380761523046 0
99.3186372745491 0
99.5991983967936 0
99.8797595190381 0
100.160320641283 0
100.440881763527 0
100.721442885772 0
101.002004008016 0
101.282565130261 0
101.563126252505 0
101.84368737475 0
102.124248496994 0
102.404809619238 0
102.685370741483 0
102.965931863727 0
103.246492985972 0
103.527054108216 0
103.807615230461 0
104.088176352705 0
104.36873747495 0
104.649298597194 0
104.929859719439 0
105.210420841683 0
105.490981963928 0
105.771543086172 0
106.052104208417 0
106.332665330661 0
106.613226452906 0
106.89378757515 0
107.174348697395 0
107.454909819639 0
107.735470941884 0
108.016032064128 0
108.296593186373 0
108.577154308617 0
108.857715430862 0
109.138276553106 0
109.418837675351 0
109.699398797595 0
109.97995991984 0
110.260521042084 0
110.541082164329 0
110.821643286573 0
111.102204408818 0
111.382765531062 0
111.663326653307 0
111.943887775551 0
112.224448897796 0
112.50501002004 0
112.785571142285 0
113.066132264529 0
113.346693386774 0
113.627254509018 0
113.907815631263 0
114.188376753507 0
114.468937875752 0
114.749498997996 0
115.03006012024 0
115.310621242485 0
115.591182364729 0
115.871743486974 0
116.152304609218 0
116.432865731463 0
116.713426853707 0
116.993987975952 0
117.274549098196 0
117.555110220441 0
117.835671342685 0
118.11623246493 0
118.396793587174 0
118.677354709419 0
118.957915831663 0
119.238476953908 0
119.519038076152 0
119.799599198397 0
120.080160320641 0
120.360721442886 0
120.64128256513 0
120.921843687375 0
121.202404809619 0
121.482965931864 0
121.763527054108 0
122.044088176353 0
122.324649298597 0
122.605210420842 0
122.885771543086 0
123.166332665331 0
123.446893787575 0
123.72745490982 0
124.008016032064 0
124.288577154309 0
124.569138276553 0
124.849699398798 0
125.130260521042 0
125.410821643287 0
125.691382765531 0
125.971943887776 0
126.25250501002 0
126.533066132265 0
126.813627254509 0
127.094188376754 0
127.374749498998 0
127.655310621242 0
127.935871743487 0
128.216432865731 0
128.496993987976 0
128.77755511022 0
129.058116232465 0
129.338677354709 0
129.619238476954 0
129.899799599198 0
130.180360721443 0
130.460921843687 0
130.741482965932 0
131.022044088176 0
131.302605210421 0
131.583166332665 0
131.86372745491 0
132.144288577154 0
132.424849699399 0
132.705410821643 0
132.985971943888 0
133.266533066132 0
133.547094188377 0
133.827655310621 0
134.108216432866 0
134.38877755511 0
134.669338677355 0
134.949899799599 0
135.230460921844 0
135.511022044088 0
135.791583166333 0
136.072144288577 0
136.352705410822 0
136.633266533066 0
136.913827655311 0
137.194388777555 0
137.4749498998 0
137.755511022044 0
138.036072144289 0
138.316633266533 0
138.597194388778 0
138.877755511022 0
139.158316633267 0
139.438877755511 0
139.719438877756 0
140 0
};
\addplot [very thick, black, dashed]
table {%
0 0
0.280561122244489 1.49665712386311e-95
0.561122244488978 4.57666420015818e-82
0.841683366733467 3.3398936059432e-74
1.12224448897796 1.21633360631583e-68
1.40280561122244 2.42705835188139e-64
1.68336673346693 7.71459650884803e-61
1.96392785571142 6.90232931834307e-58
2.24448897795591 2.44180071483749e-55
2.5250501002004 4.25256198974746e-53
2.80561122244489 4.23462205081325e-51
3.08617234468938 2.68264010656478e-49
3.36673346693387 1.16983491601201e-47
3.64729458917836 3.72817498076268e-46
3.92785571142285 9.09670195607764e-45
4.20841683366733 1.76321772950406e-43
4.48897795591182 2.79689029844102e-42
4.76953907815631 3.72009988943451e-41
5.0501002004008 4.23343401931799e-40
5.33066132264529 4.19197838875244e-39
5.61122244488978 3.66381657801746e-38
5.89178356713427 2.86107411348909e-37
6.17234468937876 2.01724388812431e-36
6.45290581162325 1.2958875530026e-35
6.73346693386774 7.64535227327835e-35
7.01402805611222 4.17127613480956e-34
7.29458917835671 2.11761076246828e-33
7.5751503006012 1.00575405073823e-32
7.85571142284569 4.49066584556875e-32
8.13627254509018 1.89315495056849e-31
8.41683366733467 7.56501240979033e-31
8.69739478957916 2.8754533839431e-30
8.97795591182365 1.04293215985436e-29
9.25851703406814 3.62002238492853e-29
9.53907815631263 1.2056247497629e-28
9.81963927855711 3.86188920578219e-28
10.1002004008016 1.19241457343159e-27
10.3807615230461 3.55606139263053e-27
10.6613226452906 1.02619629242529e-26
10.9418837675351 2.87046883515591e-26
11.2224448897796 7.79510655283153e-26
11.503006012024 2.05813057773232e-25
11.7835671342685 5.29047135107834e-25
12.064128256513 1.32566480539097e-24
12.3446893787575 3.24190730995869e-24
12.625250501002 7.74586989651558e-24
12.9058116232465 1.81003273900948e-23
13.186372745491 4.14061201567987e-23
13.4669338677355 9.28097364155922e-23
13.74749498998 2.04003860066731e-22
14.0280561122244 4.40090410285887e-22
14.3086172344689 9.32452163431463e-22
14.5891783567134 1.94176212836465e-21
14.8697394789579 3.97682629134317e-21
15.1503006012024 8.01527832827391e-21
15.4308617234469 1.59073328799204e-20
15.7114228456914 3.11038739839102e-20
15.9919839679359 5.99513131414289e-20
16.2725450901804 1.13963776293498e-19
16.5531062124249 2.13759079933464e-19
16.8336673346693 3.95792209181926e-19
17.1142284569138 7.23736500050069e-19
17.3947895791583 1.30749746097014e-18
17.6753507014028 2.3346214884334e-18
17.9559118236473 4.12161582484023e-18
18.2364729458918 7.19691966350796e-18
18.5170340681363 1.24336784470204e-17
18.7975951903808 2.1260113573301e-17
19.0781563126253 3.59896362195856e-17
19.3587174348697 6.03340382994218e-17
19.6392785571142 1.00194074609903e-16
19.9198396793587 1.64866666397782e-16
20.2004008016032 2.68872527689997e-16
20.4809619238477 4.34700082318037e-16
20.7615230460922 6.96891762903589e-16
21.0420841683367 1.10808154655634e-15
21.3226452905812 1.74784714231366e-15
21.6032064128257 2.73559219093313e-15
21.8837675350701 4.24915692267493e-15
22.1643286573146 6.55151415074726e-15
22.4448897795591 1.00287987507318e-14
22.7254509018036 1.52441171501735e-14
23.0060120240481 2.30131841012469e-14
23.2865731462926 3.45099585610041e-14
23.5671342685371 5.14132526202608e-14
23.8476953907816 7.61089605743857e-14
24.1282565130261 1.11967341331798e-13
24.4088176352705 1.63720987041893e-13
24.689378757515 2.37977095438764e-13
24.9699398797595 3.4390788890015e-13
25.250501002004 4.94176046178598e-13
25.5310621242485 7.06168508607238e-13
25.811623246493 1.00363280630767e-12
26.0921843687375 1.41883609934887e-12
26.372745490982 1.99539961051008e-12
26.6533066132265 2.79200142435351e-12
26.9338677354709 3.88719037472853e-12
27.2144288577154 5.38561485579048e-12
27.4949899799599 7.42604510629261e-12
27.7755511022044 1.01916584286784e-11
28.0561122244489 1.39231669093212e-11
28.3366733466934 1.89354999436256e-11
28.6172344689379 2.56389131668952e-11
28.8977955911824 3.45655857085195e-11
29.1783567134269 4.64029940593935e-11
29.4589178356713 6.20356189264202e-11
29.7394789579158 8.25968575594925e-11
30.0200400801603 1.0953338518804e-10
30.3006012024048 1.44684643137504e-10
30.5811623246493 1.90380636404816e-10
30.8617234468938 2.49561809220078e-10
31.1422845691383 3.25925443447738e-10
31.4228456913828 4.24103802392257e-10
31.7034068136273 5.49880133226402e-10
31.9839679358717 7.10449657120231e-10
32.2645290581162 9.14733830137338e-10
32.5450901803607 1.17375746358042e-09
32.8256513026052 1.50109976368856e-09
33.1062124248497 1.91343200174901e-09
33.3867735470942 2.43115637132424e-09
33.6673346693387 3.07916264397997e-09
33.9478957915832 3.8877215130225e-09
34.2284569138277 4.89353603029311e-09
34.5090180360721 6.14097530731898e-09
34.7895791583166 7.68351768142211e-09
35.0701402805611 9.58543385076017e-09
35.3507014028056 1.19237440691096e-08
35.6312625250501 1.47904873688624e-08
35.9118236472946 1.82953449541298e-08
36.1923847695391 2.25686643771256e-08
36.4729458917836 2.77649358789694e-08
36.7535070140281 3.40667773357631e-08
37.0340681362725 4.16894895932784e-08
37.314629258517 5.08862495849849e-08
37.5951903807615 6.19540144898892e-08
37.875751503006 7.52402162743446e-08
38.1563126252505 9.11503322454951e-08
38.436873747495 1.10156423687189e-07
38.7174348697395 1.32806741208545e-07
38.997995991984 1.59736502061674e-07
39.2785571142285 1.91679951320838e-07
39.5591182364729 2.29483825396809e-07
39.8396793587174 2.74122342815634e-07
40.1202404809619 3.26713853441988e-07
40.4008016032064 3.88539283287434e-07
40.6813627254509 4.6106251762084e-07
40.9619238476954 5.45952870191932e-07
41.2424849699399 6.45109790884113e-07
41.5230460921844 7.60689967917858e-07
41.8036072144289 8.95136983714911e-07
42.0841683366734 1.05121368558463e-06
42.3647294589178 1.23203743338716e-06
42.6452905811623 1.4411183861308e-06
42.9258517034068 1.68240098795562e-06
43.2064128256513 1.96030881100637e-06
43.4869739478958 2.279792908186e-06
43.7675350701403 2.64638382258422e-06
44.0480961923848 3.06624739237258e-06
44.3286573146293 3.54624448004577e-06
44.6092184368737 4.09399474295672e-06
44.8897795591182 4.71794454805596e-06
45.1703406813627 5.42743911751791e-06
45.4509018036072 6.23279897345199e-06
45.7314629258517 7.14540072910575e-06
46.0120240480962 8.17776225082677e-06
46.2925851703407 9.34363218955435e-06
46.5731462925852 1.06580838527541e-05
46.8537074148297 1.21376133575161e-05
47.1342685370741 1.38002419730646e-05
47.4148296593186 1.5665622526228e-05
47.6953907815631 1.77551497066199e-05
47.9759519038076 2.00920740694547e-05
48.2565130260521 2.2701619493289e-05
48.5370741482966 2.56111038076217e-05
48.8176352705411 2.88500622615137e-05
49.0981963927856 3.24503734593066e-05
49.3787575150301 3.64463873435078e-05
49.6593186372746 4.08750547581287e-05
49.939879759519 4.57760580786313e-05
50.2204408817635 5.11919423473447e-05
50.501002004008 5.71682463062466e-05
50.7815631262525 6.37536326726726e-05
51.062124248497 7.10000169580583e-05
51.3426853707415 7.89626940859787e-05
51.623246492986 8.77004620235101e-05
51.9038076152305 9.72757416000594e-05
52.1843687374749 0.000107754691650559
52.4649298597194 0.000119207318585583
52.7454909819639 0.000131707579460432
53.0260521042084 0.000145333477588171
53.3066132264529 0.000160167149719346
53.5871743486974 0.000176294943793043
53.8677354709419 0.000193807486251461
54.1482965931864 0.000212799737902714
54.4288577154309 0.00023337103731502
54.7094188376754 0.000255625130730184
54.9899799599198 0.000279670187494765
55.2705410821643 0.000305618800025186
55.5511022044088 0.00033358796734674
55.8316633266533 0.000363699061277856
56.1122244488978 0.00039607777436916
56.3927855711423 0.000430854048751544
56.6733466933868 0.000468161985100431
56.9539078156313 0.000508139730982184
57.2344689378758 0.000550929347915019
57.5150300601202 0.000596676656550229
57.7955911823647 0.00064553105945867
58.0761523046092 0.000697645341093999
58.3567134268537 0.000753175444595991
58.6372745490982 0.000812280225195544
58.9178356713427 0.000875121180085336
59.1983967935872 0.00094186215472904
59.4789579158317 0.00101266902569367
59.7595190380761 0.0010877093602058
60.0400801603206 0.00116715205275235
60.3206412825651 0.00125116693916729
60.6012024048096 0.00133992438877081
60.8817635270541 0.00143359487525166
61.1623246492986 0.00153234852710888
61.4428857715431 0.00163635465859533
61.7234468937876 0.00174578128222829
62.0040080160321 0.0018607946040569
62.2845691382766 0.00198155850299407
62.565130260521 0.00210823399563923
62.8456913827655 0.00224097868813078
63.12625250501 0.00237994621667346
63.4068136272545 0.0025252856784896
63.687374749499 0.00267714105504192
63.9679358717435 0.002835650629458
64.248496993988 0.00300094640017679
64.5290581162325 0.0031731534929049
64.809619238477 0.0033523895730377
65.0901803607214 0.00353876426075612
65.3707414829659 0.0037323785510547
65.6513026052104 0.00393332424099386
65.9318637274549 0.00414168336649209
66.2124248496994 0.00435752765098947
66.4929859719439 0.00458091796831897
66.7735470941884 0.0048119038221096
67.0541082164329 0.00505052284402973
67.3346693386774 0.00529680031314789
67.6152304609219 0.00555074869864216
67.8957915831663 0.00581236722804376
68.1763527054108 0.00608164148312546
68.4569138276553 0.00635854302547942
68.7374749498998 0.00664302905373863
69.0180360721443 0.00693504209429581
69.2985971943888 0.00723450972727889
69.5791583166333 0.00754134434941385
69.8597194388778 0.00785544297528832
70.1402805611222 0.00817668707839514
70.4208416833667 0.0085049424731982
70.7014028056112 0.00884005923930733
70.9819639278557 0.00918187168870229
71.2625250501002 0.00953019837678853
71.5430861723447 0.00988484215789685
71.8236472945892 0.0102455902856775
72.1042084168337 0.0106122145586656
72.3847695390782 0.0109844715111254
72.6653306613226 0.0113621026490973
72.9458917835671 0.0117448347314199
73.2264529058116 0.0121323800952879
73.5070140280561 0.0125244370257781
73.7875751503006 0.012920690168559
74.0681362725451 0.0133208109848783
74.3486973947896 0.0137244582477115
74.6292585170341 0.0141312785778402
74.9098196392786 0.0145409070184304
75.190380761523 0.0149529676465743
75.4709418837675 0.015367074220078
75.751503006012 0.0157828308576807
76.0320641282565 0.0161998327507297
76.312625250501 0.0166176669042418
76.5931863727455 0.0170359129051589
76.87374749499 0.0174541437155149
77.1543086172345 0.0178719264881366
77.434869739479 0.0182888234024078
77.7154308617234 0.0187043925175943
77.9959919839679 0.0191181886411213
78.2765531062124 0.0195297642091794
78.5571142284569 0.0199386701769895
78.8376753507014 0.0203444569160281
79.1182364729459 0.0207466751154969
79.3987975951904 0.0211448766853378
79.6793587174349 0.0215386156580535
79.9599198396794 0.0219274490866835
80.2404809619239 0.0223109379362373
80.5210420841683 0.0226886479659816
80.8016032064128 0.0230601506000017
81.0821643286573 0.0234250237835093
81.3627254509018 0.023782852822468
81.6432865731463 0.0241332312041361
81.9238476953908 0.0244757613962594
82.2044088176353 0.024810055622704
82.4849699398798 0.0251357366134344
82.7655310621243 0.0254524383268426
83.0460921843687 0.0257598066425587
83.3266533066132 0.026057500022964
83.6072144288577 0.0263451901417937
83.8877755511022 0.0266225624782925
84.1683366733467 0.0268893168755605
84.4488977955912 0.0271451680618432
84.7294589178357 0.0273898461336582
85.0100200400802 0.0276230969997794
85.2905811623247 0.0278446827852686
85.5711422845691 0.0280543821948478
85.8517034068136 0.028251990835093
86.1322645290581 0.028437321495
86.4128256513026 0.0286102043847127
86.6933867735471 0.0287704873322325
86.9739478957916 0.0289180359381459
87.2545090180361 0.0290527336885114
87.5350701402806 0.0291744820261507
87.815631262525 0.029283200380775
88.0961923847695 0.0293788261584377
88.376753507014 0.0294613146909596
88.6573146292585 0.029530639146078
88.937875751503 0.0295867903991821
89.2184368737475 0.0296297768675755
89.498997995992 0.0296596243083532
89.7795591182365 0.0296763755810017
90.060120240481 0.0296800903759995
90.3406813627254 0.0296708449106947
90.6212424849699 0.0296487315938637
90.9018036072144 0.02961385866038
91.1823647294589 0.0295663497775348
91.4629258517034 0.0295063436245252
91.7434869739479 0.0294339934467659
92.0240480961924 0.0293494665866155
92.3046092184369 0.0292529439922591
92.5851703406814 0.0291446197063747
92.8657314629259 0.0290247003363805
93.1462925851703 0.0288934045079448
93.4268537074148 0.02875096230353
93.7074148296593 0.0285976146876992
93.9879759519038 0.0284336129209352
94.2685370741483 0.0282592179636855
94.5490981963928 0.0280746998723622
94.8296593186373 0.0278803371889664
95.1102204408818 0.0276764163260191
95.3907815631263 0.02746323094842
95.6713426853707 0.0272410813538331
95.9519038076152 0.0270102738531543
96.2324649298597 0.0267711201525924
96.5130260521042 0.0265239367387967
96.7935871743487 0.0262690442684755
97.0741482965932 0.0260067669638681
97.3547094188377 0.0257374320153498
97.6352705410822 0.0254613689924572
97.9158316633267 0.0251789092644866
98.1963927855711 0.0248903854318111
98.4769539078156 0.0245961307689825
98.7575150300601 0.0242964786805867
99.0380761523046 0.0239917621708129
99.3186372745491 0.0236823133275856
99.5991983967936 0.0233684628220589
99.8797595190381 0.0230505394241967
100.160320641283 0.0227288695351007
100.440881763527 0.0224037767366983
100.721442885772 0.0220755813592817
101.002004008016 0.0217446000674004
101.282565130261 0.021411145464475
101.563126252505 0.0210755257164681
101.84368737475 0.0207380441949072
102.124248496994 0.0203989991394336
102.404809619238 0.0200586833400553
102.685370741483 0.0197173838391812
102.965931863727 0.0193753816534828
103.246492985972 0.0190329515155503
103.527054108216 0.018690361635296
103.807615230461 0.0183478734809684
104.088176352705 0.0180057415796225
104.36873747495 0.0176642133368299
104.649298597194 0.0173235288753735
104.929859719439 0.0169839208926406
105.210420841683 0.0166456145363818
105.490981963928 0.0163088272984595
105.771543086172 0.0159737689262096
106.052104208417 0.0156406413509561
106.332665330661 0.0153096386332517
106.613226452906 0.0149809469243437
106.89378757515 0.0146547444433744
107.174348697395 0.0143312014697718
107.454909819639 0.0140104803503067
107.735470941884 0.0136927355202533
108.016032064128 0.0133781135380655
108.296593186373 0.013066753133012
108.577154308617 0.0127587852651543
108.857715430862 0.0124543331970789
109.138276553106 0.0121535125767833
109.418837675351 0.0118564315310984
109.699398797595 0.0115631907690437
109.97995991984 0.0112738836945136
110.260521042084 0.0109885965276749
110.541082164329 0.0107074084344889
110.821643286573 0.0104303916637519
111.102204408818 0.0101576116910707
111.382765531062 0.00988912736918614
111.663326653307 0.00962499108408098
111.943887775551 0.00936524891629861
112.224448897796 0.00910994080694061
112.50501002004 0.00885910072778741
112.785571142285 0.00861275685503478
113.066132264529 0.00837093174612452
113.346693386774 0.00813364251918773
113.627254509018 0.00790090103460902
113.907815631263 0.00767271407826366
114.188376753507 0.00744908354597134
114.468937875752 0.00723000662874616
114.749498997996 0.00701547599842949
115.03006012024 0.00680547999331193
115.310621242485 0.0066000028033735
115.591182364729 0.00639902465478258
115.871743486974 0.00620252199331898
116.152304609218 0.00601046766639499
116.432865731463 0.00582283110338074
116.713426853707 0.0056395784939444
116.993987975952 0.00546067296414351
117.274549098196 0.00528607475002113
117.555110220441 0.00511574136847198
117.835671342685 0.00494962778517265
118.11623246493 0.00478768657937218
118.396793587174 0.00462986810536937
118.677354709419 0.00447612065051046
118.957915831663 0.00432639058956132
119.238476953908 0.0041806225353219
119.519038076152 0.00403875948536574
119.799599198397 0.00390074296480212
120.080160320641 0.00376651316497383
120.360721442886 0.00363600907801405
120.64128256513 0.00350916862720062
120.921843687375 0.00338592879306075
121.202404809619 0.00326622573518562
121.482965931864 0.00314999490973131
121.763527054108 0.00303717118259058
122.044088176353 0.00292768893822985
122.324649298597 0.00282148218419861
122.605210420842 0.00271848465132344
122.885771543086 0.00261862988961262
123.166332665331 0.0025218513599006
123.446893787575 0.00242808252127195
123.72745490982 0.00233725691431266
124.008016032064 0.00224930824024017
124.288577154309 0.00216417043597287
124.569138276553 0.00208177774520251
124.849699398798 0.00200206478554148
125.130260521042 0.00192496661181876
125.410821643287 0.0018504187756044
125.691382765531 0.00177835738104574
125.971943887776 0.00170871913710179
126.25250501002 0.00164144140626647
126.533066132265 0.00157646224987153
126.813627254509 0.00151372047006587
127.094188376754 0.00145315564856611
127.374749498998 0.00139470818227829
127.655310621242 0.00133831931588898
127.935871743487 0.00128393117152751
128.216432865731 0.00123148677559938
128.496993987976 0.00118093008289312
128.77755511022 0.00113220599806238
129.058116232465 0.00108526039458409
129.338677354709 0.00104004013129479
129.619238476954 0.00099649306660466
129.899799599198 0.000954568070490386
130.180360721443 0.00091421503436416
130.460921843687 0.000875384878917822
130.741482965932 0.00083802956003777
131.022044088176 0.000802102072885103
131.302605210421 0.000767556454235637
131.583166332665 0.000734347783169555
131.86372745491 0.000702432180201909
132.144288577154 0.000671766804940615
132.424849699399 0.000642309852358125
132.705410821643 0.00061402054776055
132.985971943888 0.000586859140534773
133.266533066132 0.000560786896753909
133.547094188377 0.000535766090716983
133.827655310621 0.000511759995497904
134.108216432866 0.000488732872575881
134.38877755511 0.000466649960616715
134.669338677355 0.000445477463472977
134.949899799599 0.00042518253746764
135.230460921844 0.000405733278023913
135.511022044088 0.000387098705701556
135.791583166333 0.000369248751697285
136.072144288577 0.000352154242864928
136.352705410822 0.00033578688630831
136.633266533066 0.000320119253597728
136.913827655311 0.000305124764658113
137.194388777555 0.000290777671375712
137.4749498998 0.00027705304096652
137.755511022044 0.00026392673914889
138.036072144289 0.000251375413159448
138.316633266533 0.000239376474650176
138.597194388778 0.000227908082501953
138.877755511022 0.000216949125588174
139.158316633267 0.00020647920551976
139.438877755511 0.00019647861940156
139.719438877756 0.000186928342627513
140 0.000177810011740858
};
\end{axis}

\end{tikzpicture}

%% file: simulation/notebooks/figures/bgmres_gaussian_uq__seed_12345__eq_type_spd_reverse__iteration_10.tikz
\begin{tikzpicture}

\begin{axis}[
xlabel={$z$},
ylabel={$p(z)$},
xmin=0, xmax=140,
ymin=-0.01, ymax=0.3,
width=\figwidth,
height=\figheight,
xtick={0,20,40,60,80,100,120,140},
xticklabels={$0$,$20$,$40$,$60$,$80$,$100$,$120$,$140$},
ytick={-0.05,0,0.05,0.1,0.15,0.2,0.25,0.3},
yticklabels={,$0.00$,$0.05$,$0.10$,$0.15$,$0.20$,$0.25$,$0.30$},
tick align=outside,
tick pos=left,
x grid style={white!69.01960784313725!black},
y grid style={white!69.01960784313725!black},
legend cell align={left},
legend entries={{BayesGMRES},{$\chi^2_{90}$}},
legend style={draw=white!80.0!black}
]
\addlegendimage{no markers, blue!50.98039215686274!black}
\addlegendimage{no markers, black}
\addplot [semithick, blue!50.98039215686274!black]
table {%
0 0.0405768111774501
0.280561122244489 0.0694973460915201
0.561122244488978 0.103917475175022
0.841683366733467 0.13840220787556
1.12224448897796 0.168107742958473
1.40280561122244 0.190828432192365
1.68336673346693 0.206635447141317
1.96392785571142 0.215977799541008
2.24448897795591 0.218568635739349
2.5250501002004 0.214226830894041
2.80561122244489 0.204261534359589
3.08617234468938 0.19129187482381
3.36673346693387 0.177678338063002
3.64729458917836 0.164504583401439
3.92785571142285 0.151797522859249
4.20841683366733 0.139024035270087
4.48897795591182 0.125414663491368
4.76953907815631 0.110575371797746
5.0501002004008 0.0951707182737848
5.33066132264529 0.0807263783098223
5.61122244488978 0.0685115804824647
5.89178356713427 0.0586535892674208
6.17234468937876 0.0503705626374621
6.45290581162325 0.0430313110993583
6.73346693386774 0.0369475459539335
7.01402805611222 0.0330041618490371
7.29458917835671 0.0314169598091236
7.5751503006012 0.0309713372141604
7.85571142284569 0.0296967358194918
8.13627254509018 0.0264115799494326
8.41683366733467 0.0216189031583772
8.69739478957916 0.0168746220762911
8.97795591182365 0.0133426813277345
9.25851703406814 0.011029959672971
9.53907815631263 0.00926547920956178
9.81963927855711 0.00760892920780565
10.1002004008016 0.0061837665242217
10.3807615230461 0.00530815347225628
10.6613226452906 0.00500664467428203
10.9418837675351 0.00490951445032786
11.2224448897796 0.0045493946393345
11.503006012024 0.00372230443283319
11.7835671342685 0.00259592052712479
12.064128256513 0.00151626248153188
12.3446893787575 0.000734626180476177
12.625250501002 0.000293408637458966
12.9058116232465 9.6162672123185e-05
13.186372745491 2.57668626086158e-05
13.4669338677355 5.6270572344592e-06
13.74749498998 9.9886284943065e-07
14.0280561122244 1.43796776217041e-07
14.3086172344689 1.67565561235463e-08
14.5891783567134 1.5780872016266e-09
14.8697394789579 1.19959311049776e-10
15.1503006012024 7.35274436032159e-12
15.4308617234469 3.63098992450951e-13
15.7114228456914 1.44372116953552e-14
15.9919839679359 4.61963216031434e-16
16.2725450901804 1.18912741862977e-17
16.5531062124249 2.4615874589822e-19
16.8336673346693 4.09700546798522e-21
17.1142284569138 5.48155422370836e-23
17.3947895791583 5.89475964536722e-25
17.6753507014028 5.0945613901384e-27
17.9559118236473 3.53825837135152e-29
18.2364729458918 1.97463247787734e-31
18.5170340681363 8.85472914305674e-34
18.7975951903808 3.19036030526264e-36
19.0781563126253 9.23563125803186e-39
19.3587174348697 2.14805654990979e-41
19.6392785571142 4.01392490644873e-44
19.9198396793587 6.0260269166119e-47
20.2004008016032 7.26819251685316e-50
20.4809619238477 7.04290127930158e-53
20.7615230460922 5.48281652988846e-56
21.0420841683367 3.42910280416444e-59
21.3226452905812 1.72297992343355e-62
21.6032064128257 6.95506948824379e-66
21.8837675350701 2.25550191486326e-69
22.1643286573146 5.87630808217766e-73
22.4448897795591 1.22994271903011e-76
22.7254509018036 2.06815780155439e-80
23.0060120240481 2.7938334557313e-84
23.2865731462926 3.03204116460505e-88
23.5671342685371 2.64354777580421e-92
23.8476953907816 1.85163990091911e-96
24.1282565130261 1.04194079541179e-100
24.4088176352705 4.71027842445715e-105
24.689378757515 1.71067318126047e-109
24.9699398797595 4.99119310147563e-114
25.250501002004 1.16992628304534e-118
25.5310621242485 2.20307548406089e-123
25.811623246493 3.33285927377454e-128
26.0921843687375 4.0506177936954e-133
26.372745490982 3.95496117785608e-138
26.6533066132265 3.10227172854615e-143
26.9338677354709 1.95494292880726e-148
27.2144288577154 9.89703165224326e-154
27.4949899799599 4.02524624929937e-159
27.7755511022044 1.31521441823372e-164
28.0561122244489 3.45236953621139e-170
28.3366733466934 7.28038914402664e-176
28.6172344689379 1.23341297030434e-181
28.8977955911824 1.67872294392166e-187
29.1783567134269 1.83554965215056e-193
29.4589178356713 1.61238905446377e-199
29.7394789579158 1.13786342683947e-205
30.0200400801603 6.45099978704653e-212
30.3006012024048 2.93819427053791e-218
30.5811623246493 1.07510424807252e-224
30.8617234468938 3.1603649614421e-231
31.1422845691383 7.46346478515088e-238
31.4228456913828 1.41599045922087e-244
31.7034068136273 2.15822518071057e-251
31.9839679358717 2.64271160777161e-258
32.2645290581162 2.59967691701288e-265
32.5450901803607 2.05449728816399e-272
32.8256513026052 1.30439268672765e-279
33.1062124248497 6.65315630355824e-287
33.3867735470942 2.72623749304644e-294
33.6673346693387 8.97462341722074e-302
33.9478957915832 2.37347887493952e-309
34.2284569138277 5.04279267311464e-317
34.5090180360721 0
34.7895791583166 0
35.0701402805611 0
35.3507014028056 0
35.6312625250501 0
35.9118236472946 0
36.1923847695391 0
36.4729458917836 0
36.7535070140281 0
37.0340681362725 0
37.314629258517 0
37.5951903807615 0
37.875751503006 0
38.1563126252505 0
38.436873747495 0
38.7174348697395 0
38.997995991984 0
39.2785571142285 0
39.5591182364729 0
39.8396793587174 0
40.1202404809619 0
40.4008016032064 0
40.6813627254509 0
40.9619238476954 0
41.2424849699399 0
41.5230460921844 0
41.8036072144289 0
42.0841683366734 0
42.3647294589178 0
42.6452905811623 0
42.9258517034068 0
43.2064128256513 0
43.4869739478958 0
43.7675350701403 0
44.0480961923848 0
44.3286573146293 0
44.6092184368737 0
44.8897795591182 0
45.1703406813627 0
45.4509018036072 0
45.7314629258517 0
46.0120240480962 0
46.2925851703407 0
46.5731462925852 0
46.8537074148297 0
47.1342685370741 0
47.4148296593186 0
47.6953907815631 0
47.9759519038076 0
48.2565130260521 0
48.5370741482966 0
48.8176352705411 0
49.0981963927856 0
49.3787575150301 0
49.6593186372746 0
49.939879759519 0
50.2204408817635 0
50.501002004008 0
50.7815631262525 0
51.062124248497 0
51.3426853707415 0
51.623246492986 0
51.9038076152305 0
52.1843687374749 0
52.4649298597194 0
52.7454909819639 0
53.0260521042084 0
53.3066132264529 0
53.5871743486974 0
53.8677354709419 0
54.1482965931864 0
54.4288577154309 0
54.7094188376754 0
54.9899799599198 0
55.2705410821643 0
55.5511022044088 0
55.8316633266533 0
56.1122244488978 0
56.3927855711423 0
56.6733466933868 0
56.9539078156313 0
57.2344689378758 0
57.5150300601202 0
57.7955911823647 0
58.0761523046092 0
58.3567134268537 0
58.6372745490982 0
58.9178356713427 0
59.1983967935872 0
59.4789579158317 0
59.7595190380761 0
60.0400801603206 0
60.3206412825651 0
60.6012024048096 0
60.8817635270541 0
61.1623246492986 0
61.4428857715431 0
61.7234468937876 0
62.0040080160321 0
62.2845691382766 0
62.565130260521 0
62.8456913827655 0
63.12625250501 0
63.4068136272545 0
63.687374749499 0
63.9679358717435 0
64.248496993988 0
64.5290581162325 0
64.809619238477 0
65.0901803607214 0
65.3707414829659 0
65.6513026052104 0
65.9318637274549 0
66.2124248496994 0
66.4929859719439 0
66.7735470941884 0
67.0541082164329 0
67.3346693386774 0
67.6152304609219 0
67.8957915831663 0
68.1763527054108 0
68.4569138276553 0
68.7374749498998 0
69.0180360721443 0
69.2985971943888 0
69.5791583166333 0
69.8597194388778 0
70.1402805611222 0
70.4208416833667 0
70.7014028056112 0
70.9819639278557 0
71.2625250501002 0
71.5430861723447 0
71.8236472945892 0
72.1042084168337 0
72.3847695390782 0
72.6653306613226 0
72.9458917835671 0
73.2264529058116 0
73.5070140280561 0
73.7875751503006 0
74.0681362725451 0
74.3486973947896 0
74.6292585170341 0
74.9098196392786 0
75.190380761523 0
75.4709418837675 0
75.751503006012 0
76.0320641282565 0
76.312625250501 0
76.5931863727455 0
76.87374749499 0
77.1543086172345 0
77.434869739479 0
77.7154308617234 0
77.9959919839679 0
78.2765531062124 0
78.5571142284569 0
78.8376753507014 0
79.1182364729459 0
79.3987975951904 0
79.6793587174349 0
79.9599198396794 0
80.2404809619239 0
80.5210420841683 0
80.8016032064128 0
81.0821643286573 0
81.3627254509018 0
81.6432865731463 0
81.9238476953908 0
82.2044088176353 0
82.4849699398798 0
82.7655310621243 0
83.0460921843687 0
83.3266533066132 0
83.6072144288577 0
83.8877755511022 0
84.1683366733467 0
84.4488977955912 0
84.7294589178357 0
85.0100200400802 0
85.2905811623247 0
85.5711422845691 0
85.8517034068136 0
86.1322645290581 0
86.4128256513026 0
86.6933867735471 0
86.9739478957916 0
87.2545090180361 0
87.5350701402806 0
87.815631262525 0
88.0961923847695 0
88.376753507014 0
88.6573146292585 0
88.937875751503 0
89.2184368737475 0
89.498997995992 0
89.7795591182365 0
90.060120240481 0
90.3406813627254 0
90.6212424849699 0
90.9018036072144 0
91.1823647294589 0
91.4629258517034 0
91.7434869739479 0
92.0240480961924 0
92.3046092184369 0
92.5851703406814 0
92.8657314629259 0
93.1462925851703 0
93.4268537074148 0
93.7074148296593 0
93.9879759519038 0
94.2685370741483 0
94.5490981963928 0
94.8296593186373 0
95.1102204408818 0
95.3907815631263 0
95.6713426853707 0
95.9519038076152 0
96.2324649298597 0
96.5130260521042 0
96.7935871743487 0
97.0741482965932 0
97.3547094188377 0
97.6352705410822 0
97.9158316633267 0
98.1963927855711 0
98.4769539078156 0
98.7575150300601 0
99.0380761523046 0
99.3186372745491 0
99.5991983967936 0
99.8797595190381 0
100.160320641283 0
100.440881763527 0
100.721442885772 0
101.002004008016 0
101.282565130261 0
101.563126252505 0
101.84368737475 0
102.124248496994 0
102.404809619238 0
102.685370741483 0
102.965931863727 0
103.246492985972 0
103.527054108216 0
103.807615230461 0
104.088176352705 0
104.36873747495 0
104.649298597194 0
104.929859719439 0
105.210420841683 0
105.490981963928 0
105.771543086172 0
106.052104208417 0
106.332665330661 0
106.613226452906 0
106.89378757515 0
107.174348697395 0
107.454909819639 0
107.735470941884 0
108.016032064128 0
108.296593186373 0
108.577154308617 0
108.857715430862 0
109.138276553106 0
109.418837675351 0
109.699398797595 0
109.97995991984 0
110.260521042084 0
110.541082164329 0
110.821643286573 0
111.102204408818 0
111.382765531062 0
111.663326653307 0
111.943887775551 0
112.224448897796 0
112.50501002004 0
112.785571142285 0
113.066132264529 0
113.346693386774 0
113.627254509018 0
113.907815631263 0
114.188376753507 0
114.468937875752 0
114.749498997996 0
115.03006012024 0
115.310621242485 0
115.591182364729 0
115.871743486974 0
116.152304609218 0
116.432865731463 0
116.713426853707 0
116.993987975952 0
117.274549098196 0
117.555110220441 0
117.835671342685 0
118.11623246493 0
118.396793587174 0
118.677354709419 0
118.957915831663 0
119.238476953908 0
119.519038076152 0
119.799599198397 0
120.080160320641 0
120.360721442886 0
120.64128256513 0
120.921843687375 0
121.202404809619 0
121.482965931864 0
121.763527054108 0
122.044088176353 0
122.324649298597 0
122.605210420842 0
122.885771543086 0
123.166332665331 0
123.446893787575 0
123.72745490982 0
124.008016032064 0
124.288577154309 0
124.569138276553 0
124.849699398798 0
125.130260521042 0
125.410821643287 0
125.691382765531 0
125.971943887776 0
126.25250501002 0
126.533066132265 0
126.813627254509 0
127.094188376754 0
127.374749498998 0
127.655310621242 0
127.935871743487 0
128.216432865731 0
128.496993987976 0
128.77755511022 0
129.058116232465 0
129.338677354709 0
129.619238476954 0
129.899799599198 0
130.180360721443 0
130.460921843687 0
130.741482965932 0
131.022044088176 0
131.302605210421 0
131.583166332665 0
131.86372745491 0
132.144288577154 0
132.424849699399 0
132.705410821643 0
132.985971943888 0
133.266533066132 0
133.547094188377 0
133.827655310621 0
134.108216432866 0
134.38877755511 0
134.669338677355 0
134.949899799599 0
135.230460921844 0
135.511022044088 0
135.791583166333 0
136.072144288577 0
136.352705410822 0
136.633266533066 0
136.913827655311 0
137.194388777555 0
137.4749498998 0
137.755511022044 0
138.036072144289 0
138.316633266533 0
138.597194388778 0
138.877755511022 0
139.158316633267 0
139.438877755511 0
139.719438877756 0
140 0
};
\addplot [very thick, black, dashed]
table {%
0 0
0.280561122244489 4.80106224519259e-93
0.561122244488978 7.34064247246822e-80
0.841683366733467 3.57130052006939e-72
1.12224448897796 9.75456111779384e-67
1.40280561122244 1.55713129404283e-62
1.68336673346693 4.12455391919496e-59
1.96392785571142 3.16309805802856e-56
2.24448897795591 9.79118483065695e-54
2.5250501002004 1.51573459491719e-51
2.80561122244489 1.35840625930022e-49
3.08617234468938 7.82320566141746e-48
3.36673346693387 3.12721940941077e-46
3.64729458917836 9.19957903220104e-45
3.92785571142285 2.08435145330341e-43
4.20841683366733 3.77076705866809e-42
4.48897795591182 5.6075153260264e-41
4.76953907815631 7.0197347073324e-40
5.0501002004008 7.5445841987134e-39
5.33066132264529 7.07750937740147e-38
5.61122244488978 5.87650008995611e-37
5.89178356713427 4.37043668152382e-36
6.17234468937876 2.94137737063851e-35
6.45290581162325 1.80740093184315e-34
6.73346693386774 1.02188324581234e-33
7.01402805611222 5.35234317755158e-33
7.29458917835671 2.61268954237511e-32
7.5751503006012 1.19493159837713e-31
7.85571142284569 5.14479089603316e-31
8.13627254509018 2.09412780369544e-30
8.41683366733467 8.08915969818331e-30
8.69739478957916 2.97549796020485e-29
8.97795591182365 1.04549293078261e-28
9.25851703406814 3.51894383781703e-28
9.53907815631263 1.13749175445072e-27
9.81963927855711 3.53953967819764e-27
10.1002004008016 1.06252655739713e-26
10.3807615230461 3.08306403751051e-26
10.6613226452906 8.66287133323723e-26
10.9418837675351 2.36103947594977e-25
11.2224448897796 6.25139705870853e-25
11.503006012024 1.61028996944185e-24
11.7835671342685 4.04073245539009e-24
12.064128256513 9.88963561629746e-24
12.3446893787575 2.36353989107552e-23
12.625250501002 5.52169868337352e-23
12.9058116232465 1.26224488057329e-22
13.186372745491 2.82606209155306e-22
13.4669338677355 6.20250783098873e-22
13.74749498998 1.33554130548065e-21
14.0280561122244 2.82349433227714e-21
14.3086172344689 5.86504575065952e-21
14.5891783567134 1.19786452177007e-20
14.8697394789579 2.40699823105429e-20
15.1503006012024 4.76145700691535e-20
15.4308617234469 9.27790025502925e-20
15.7114228456914 1.78172829160392e-19
15.9919839679359 3.37395172077905e-19
16.2725450901804 6.30309506569117e-19
16.5531062124249 1.16221795155114e-18
16.8336673346693 2.11607477551917e-18
17.1142284569138 3.80597262497073e-18
17.3947895791583 6.76494366038831e-18
17.6753507014028 1.18875114564114e-17
17.9559118236473 2.06586793184801e-17
18.2364729458918 3.55179848448527e-17
18.5170340681363 6.04325215428258e-17
18.7975951903808 1.01790159976225e-16
19.0781563126253 1.69778840611521e-16
19.3587174348697 2.80497066255402e-16
19.6392785571142 4.59154682727643e-16
19.9198396793587 7.4488551186365e-16
20.2004008016032 1.19792313676172e-15
20.4809619238477 1.9102133754309e-15
20.7615230460922 3.02098543166043e-15
21.0420841683367 4.73942307198553e-15
21.3226452905812 7.37742623696504e-15
21.6032064128257 1.13966090254926e-14
21.8837675350701 1.74752415199029e-14
22.1643286573146 2.66029385632978e-14
22.4448897795591 4.02136921335163e-14
22.7254509018036 6.03715433169587e-14
23.0060120240481 9.00280573159453e-14
23.2865731462926 1.33377128999548e-13
23.5671342685371 1.96340916256465e-13
23.8476953907816 2.8723137978116e-13
24.1282565130261 4.17645622858077e-13
24.4088176352705 6.0367073300916e-13
24.689378757515 8.67496051635979e-13
24.9699398797595 1.2395588515655e-12
25.250501002004 1.76138462173666e-12
25.5310621242485 2.48932713669961e-12
25.811623246493 3.49946811578268e-12
26.0921843687375 4.89400378047318e-12
26.372745490982 6.80952861003135e-12
26.6533066132265 9.42772811787401e-12
26.9338677354709 1.29891160512699e-11
27.2144288577154 1.78106011173469e-11
27.4949899799599 2.43078488136746e-11
27.7755511022044 3.30236204929268e-11
28.0561122244489 4.46635304212604e-11
28.3366733466934 6.01409690571616e-11
28.6172344689379 8.06333046446288e-11
28.8977955911824 1.07651903895262e-10
29.1783567134269 1.43129015330181e-10
29.4589178356713 1.89525145985418e-10
29.7394789579158 2.49961244811111e-10
30.0200400801603 3.28380796314756e-10
30.3006012024048 4.2974783884296e-10
30.5811623246493 5.60287967296239e-10
30.8617234468938 7.27780574810016e-10
31.1422845691383 9.41911947570066e-10
31.4228456913828 1.21470036769369e-09
31.7034068136273 1.56100611777492e-09
31.9839679358717 1.99914123441544e-09
32.2645290581162 2.55159604419065e-09
32.5450901803607 3.24590195131768e-09
32.8256513026052 4.11565264879466e-09
33.1062124248497 5.20170891032387e-09
33.3867735470942 6.55361540433236e-09
33.6673346693387 8.23126156792538e-09
33.9478957915832 1.0306822499991e-08
34.2284569138277 1.28670200890207e-08
34.5090180360721 1.60157491899946e-08
34.7895791583166 1.98771186101775e-08
35.0701402805611 2.45989619564375e-08
35.3507014028056 3.0356879033092e-08
35.6312625250501 3.7358874450823e-08
35.9118236472946 4.58506664001121e-08
36.1923847695391 5.61217451371403e-08
36.4729458917836 6.85122676002496e-08
36.7535070140281 8.34208816875231e-08
37.0340681362725 1.01313581040812e-07
37.314629258517 1.22733698649934e-07
37.5951903807615 1.48313155156776e-07
37.875751503006 1.78785085337614e-07
38.1563126252505 2.14997973799648e-07
38.436873747495 2.57931438362451e-07
38.7174348697395 3.08713806815523e-07
38.997995991984 3.68641639650059e-07
39.2785571142285 4.39201357847926e-07
39.5591182364729 5.22093140758424e-07
39.8396793587174 6.19257264378803e-07
40.1202404809619 7.32903054849169e-07
40.4008016032064 8.65540635537659e-07
40.6813627254509 1.02001564858882e-06
40.9619238476954 1.19954713308809e-06
41.2424849699399 1.40776874191481e-06
41.5230460921844 1.64877347776028e-06
41.8036072144289 1.92716212553391e-06
42.0841683366734 2.24809555331468e-06
42.3647294589178 2.61735104699244e-06
42.6452905811623 3.0413828346399e-06
42.9258517034068 3.52738694533566e-06
43.2064128256513 4.08337053350178e-06
43.4869739478958 4.71822578371594e-06
43.7675350701403 5.44180849231969e-06
44.0480961923848 6.26502140088537e-06
44.3286573146293 7.19990233267916e-06
44.6092184368737 8.25971715661259e-06
44.8897795591182 9.4590575738036e-06
45.1703406813627 1.08139436897844e-05
45.4509018036072 1.23419313006058e-05
45.7314629258517 1.40622237837054e-05
46.0120240480962 1.59957884444536e-05
46.2925851703407 1.81654771269657e-05
46.5731462925852 2.05961508531499e-05
46.8537074148297 2.33148082072746e-05
47.1342685370741 2.63507171347942e-05
47.4148296593186 2.97355497740041e-05
47.6953907815631 3.3503519887578e-05
47.9759519038076 3.76915224084894e-05
48.2565130260521 4.23392745616132e-05
48.5370741482966 4.74894579686349e-05
48.8176352705411 5.31878610905025e-05
49.0981963927856 5.94835213084705e-05
49.3787575150301 6.6428865892733e-05
49.6593186372746 7.40798510568045e-05
49.939879759519 8.24960982468474e-05
50.2204408817635 9.17410267685281e-05
50.501002004008 0.000101881981810066
50.7815631262525 0.000112990356879634
51.062124248497 0.000125141709638406
51.3426853707415 0.000138415870078121
51.623246492986 0.000152897039964127
51.9038076152305 0.000168673882442427
52.1843687374749 0.000185839600692274
52.4649298597194 0.000204492004494992
52.7454909819639 0.000224733563585416
53.0260521042084 0.000246671446654767
53.3066132264529 0.000270417544883311
53.5871743486974 0.000296088478898496
53.8677354709419 0.000323805588078638
54.1482965931864 0.000353694901155115
54.4288577154309 0.000385887087106705
54.7094188376754 0.000420517385388016
54.9899799599198 0.000457725514591487
55.2705410821643 0.000497655558706756
55.5511022044088 0.000540455830214384
55.8316633266533 0.000586278709331987
56.1122244488978 0.000635280458818549
56.3927855711423 0.000687621013839117
56.6733466933868 0.000743463746494411
56.9539078156313 0.000802975204729442
57.2344689378758 0.000866324825450429
57.5150300601202 0.000933684621800396
57.7955911823647 0.00100522884467025
58.0761523046092 0.00108113361865191
58.3567134268537 0.00116157655277497
58.6372745490982 0.00124673632650492
58.9178356713427 0.0013367922516202
59.1983967935872 0.00143192381072723
59.4789579158317 0.00153231017331216
59.7595190380761 0.00163812969037037
60.0400801603206 0.00174955936879543
60.3206412825651 0.00186677432684405
60.6012024048096 0.00198994723213297
60.8817635270541 0.00211924772375096
61.1623246492986 0.00225484182019864
61.4428857715431 0.00239689131498758
61.7234468937876 0.00254555316184663
62.0040080160321 0.00270097885159013
62.2845691382766 0.00286331378280136
62.565130260521 0.00303269662857655
62.8456913827655 0.00320925870165682
63.12625250501 0.00339312332034312
63.4068136272545 0.00358440517765406
63.687374749499 0.00378320971623466
63.9679358717435 0.00398963251155907
64.248496993988 0.00420375866599952
64.5290581162325 0.00442566221634672
64.809619238477 0.0046554055573634
65.0901803607214 0.00489303888394594
65.3707414829659 0.00513859965444121
65.6513026052104 0.00539211207762628
65.9318637274549 0.00565358662581055
66.2124248496994 0.00592301957645119
66.4929859719439 0.00620039258460566
66.7735470941884 0.00648567228844396
67.0541082164329 0.00677880994995182
67.3346693386774 0.00707974113284171
67.6152304609219 0.00738838541956237
67.8957915831663 0.00770464616916931
68.1763527054108 0.00802841031767127
68.4569138276553 0.00835954822231535
68.7374749498998 0.00869791355111982
69.0180360721443 0.00904334321878157
69.2985971943888 0.00939565736992813
69.5791583166333 0.00975465941048933
69.8597194388778 0.0101201360877858
70.1402805611222 0.0104918576197351
70.4208416833667 0.0108695778734019
70.7014028056112 0.0112530345928944
70.9819639278557 0.0116419496764439
71.2625250501002 0.0120360295022944
71.5430861723447 0.0124349653028338
71.8236472945892 0.0128384335862108
72.1042084168337 0.0132460966044941
72.3847695390782 0.0136576028672387
72.6653306613226 0.0140725876991446
72.9458917835671 0.0144906738403318
73.2264529058116 0.014911472087559
73.5070140280561 0.0153345819745832
73.7875751503006 0.0157595924896797
74.0681362725451 0.0161860828282171
74.3486973947896 0.016613623178026
74.6292585170341 0.0170417755351995
74.9098196392786 0.0174700945478257
75.190380761523 0.0178981283850661
75.4709418837675 0.0183254196288832
75.751503006012 0.018751506185674
76.0320641282565 0.0191759222149526
76.312625250501 0.0195981990722042
76.5931863727455 0.0200178662629705
76.87374749499 0.0204344524052091
77.1543086172345 0.0208474861969413
77.434869739479 0.0212564973861841
77.7154308617234 0.0216610177402062
77.9959919839679 0.0220605820111202
78.2765531062124 0.0224547288948871
78.5571142284569 0.0228430019808329
78.8376753507014 0.023224950688837
79.1182364729459 0.0236001311914133
79.3987975951904 0.0239681073179853
79.6793587174349 0.0243284514387128
79.9599198396794 0.0246807453253883
80.2404809619239 0.0250245809869246
80.5210420841683 0.0253595614771601
80.8016032064128 0.0256853016727705
81.0821643286573 0.0260014290192156
81.3627254509018 0.0263075842427857
81.6432865731463 0.0266034220269461
81.9238476953908 0.0268886116513229
82.2044088176353 0.027162837591801
82.4849699398798 0.0274258000804031
82.7655310621243 0.0276772156236936
83.0460921843687 0.0279168174787004
83.3266533066132 0.0281443560854097
83.6072144288577 0.0283595994551296
83.8877755511022 0.0285623335140983
84.1683366733467 0.0287523624019394
84.4488977955912 0.0289295087246666
84.7294589178357 0.0290936137621251
85.0100200400802 0.0292445376298944
85.2905811623247 0.0293821593958287
85.5711422845691 0.0295063771515383
85.8517034068136 0.0296171080393107
86.1322645290581 0.0297142882350056
86.4128256513026 0.0297978728876975
86.6933867735471 0.0298678360168891
86.9739478957916 0.0299241703682535
87.2545090180361 0.0299668872290106
87.5350701402806 0.0299960162040869
87.815631262525 0.0300116049543738
88.0961923847695 0.030013718898442
88.376753507014 0.0300024408791621
88.6573146292585 0.0299778707967999
88.937875751503 0.02994012521017
89.2184368737475 0.0298893369075221
89.498997995992 0.0298256544489076
89.7795591182365 0.0297492416817597
90.060120240481 0.0296602772315581
90.3406813627254 0.0295589539693735
90.6212424849699 0.0294454784582151
90.9018036072144 0.0293200703800203
91.1823647294589 0.0291829619452558
91.4629258517034 0.0290343972869734
91.7434869739479 0.0288746318412912
92.0240480961924 0.0287039317161364
92.3046092184369 0.0285225730502045
92.5851703406814 0.0283308413639246
92.8657314629259 0.0281290309043351
93.1462925851703 0.0279174439856248
93.4268537074148 0.0276963903271475
93.7074148296593 0.0274661863906041
93.9879759519038 0.0272271547181074
94.2685370741483 0.0269796232727291
94.5490981963928 0.0267239247831249
94.8296593186373 0.0264603960937561
95.1102204408818 0.0261893775221574
95.3907815631263 0.0259112122246553
95.6713426853707 0.0256262455718616
95.9519038076152 0.0253348245351969
96.2324649298597 0.0250372970856518
96.5130260521042 0.0247340116058846
96.7935871743487 0.0244253163167139
97.0741482965932 0.0241115587189793
97.3547094188377 0.0237930850516554
97.6352705410822 0.0234702397670629
97.9158316633267 0.023143365023907
98.1963927855711 0.0228128001988303
98.4769539078156 0.0224788814170748
98.7575150300601 0.0221419411027837
99.0380761523046 0.0218023075493981
99.3186372745491 0.0214603045105319
99.5991983967936 0.0211162508116443
99.8797595190381 0.0207704599827593
100.160320641283 0.0204232399123935
100.440881763527 0.0200748925228481
100.721442885772 0.0197257134668799
101.002004008016 0.0193759918457737
101.282565130261 0.0190260099487458
101.563126252505 0.0186760430135482
101.84368737475 0.0183263590081332
102.124248496994 0.0179772184331238
102.404809619238 0.0176288741448513
102.685370741483 0.0172815711986274
102.965931863727 0.0169355467119104
103.246492985972 0.0165910297469605
103.527054108216 0.0162482412125664
103.807615230461 0.0159073937843687
104.088176352705 0.0155686918432977
104.36873747495 0.0152323314315873
104.649298597194 0.0148985002258339
104.929859719439 0.0145673775265187
105.210420841683 0.0142391342634079
105.490981963928 0.0139139330162206
105.771543086172 0.0135919280499453
106.052104208417 0.0132732653641622
106.332665330661 0.0129580827557359
106.613226452906 0.0126465098942165
106.89378757515 0.0123386684092982
107.174348697395 0.0120346719896684
107.454909819639 0.0117346264925832
107.735470941884 0.0114386300635158
108.016032064128 0.0111467732652049
108.296593186373 0.0108591392154627
108.577154308617 0.0105758037330769
108.857715430862 0.010296835491179
109.138276553106 0.0100222961774395
109.418837675351 0.00975224066046964
109.699398797595 0.00948671716181576
109.97995991984 0.00922576743295598
110.260521042084 0.00896942693671211
110.541082164329 0.00871772503250385
110.821643286573 0.00847068516490259
111.102204408818 0.00822832505494213
111.382765531062 0.0079906568936701
111.663326653307 0.00775768753743885
111.943887775551 0.0075294187044574
112.224448897796 0.0073058471721379
112.50501002004 0.00708696497479417
112.785571142285 0.00687275960127293
113.066132264529 0.00666321419210337
113.346693386774 0.00645830773579804
113.627254509018 0.00625801526392081
113.907815631263 0.00606230804459594
114.188376753507 0.005871153774124
114.468937875752 0.0056845167664042
114.749498997996 0.00550235813987918
115.03006012024 0.00532463600173613
115.310621242485 0.00515130562911911
115.591182364729 0.00498231964712735
115.871743486974 0.00481762820338928
116.152304609218 0.00465717913902358
116.432865731463 0.00450091815581482
116.713426853707 0.00434878897944812
116.993987975952 0.0042007335186656
117.274549098196 0.0040566920202231
117.555110220441 0.00391660321953773
117.835671342685 0.00378040448694078
118.11623246493 0.0036480319694539
118.396793587174 0.00351942072803223
118.677354709419 0.00339450487021996
118.957915831663 0.00327321767818747
119.238476953908 0.00315549173212288
119.519038076152 0.0030412590289702
119.799599198397 0.00293045109650835
120.080160320641 0.00282299910278678
120.360721442886 0.00271883396093271
120.64128256513 0.00261788642936185
120.921843687375 0.00252008720743063
121.202404809619 0.00242536702657393
121.482965931864 0.00233365673698513
121.763527054108 0.0022448873898971
122.044088176353 0.00215899031553219
122.324649298597 0.00207589719679496
122.605210420842 0.00199554013878623
122.885771543086 0.00191785173422218
123.166332665331 0.00184276512484774
123.446893787575 0.0017702140589339
123.72745490982 0.00170013294495929
124.008016032064 0.00163245690157062
124.288577154309 0.0015671218039269
124.569138276553 0.00150406432652912
124.849699398798 0.0014432219826432
125.130260521042 0.00138453316042252
125.410821643287 0.00132793715583889
125.691382765531 0.00127337420253137
125.971943887776 0.00122078549868344
126.25250501002 0.00117011323103856
126.533066132265 0.00112130059616303
126.813627254509 0.00107429181906861
127.094188376754 0.0010290321693016
127.374749498998 0.000985467974608537
127.655310621242 0.000943546632285347
127.935871743487 0.000903216618316152
128.216432865731 0.000864427494407158
128.496993987976 0.000827129913018275
128.77755511022 0.000791275620494606
129.058116232465 0.000756817458397098
129.338677354709 0.000723709363130626
129.619238476954 0.000691906363964408
129.899799599198 0.000661364579539098
130.180360721443 0.000632041212950986
130.460921843687 0.000603894545502348
130.741482965932 0.000576883929204453
131.022044088176 0.000550969778116737
131.302605210421 0.000526113558603845
131.583166332665 0.000502277778588881
131.86372745491 0.000479425975879454
132.144288577154 0.00045752270563972
132.424849699399 0.000436533527079371
132.705410821643 0.000416424989428082
132.985971943888 0.000397164617260658
133.266533066132 0.000378720895236369
133.547094188377 0.000361063252312437
133.827655310621 0.00034416204548985
134.108216432866 0.000327988543146802
134.38877755511 0.000312514908012192
134.669338677355 0.000297714179829943
134.949899799599 0.000283560257761694
135.230460921844 0.000270027882573429
135.511022044088 0.000257092618649182
135.791583166333 0.000244730835872576
136.072144288577 0.000232919691414794
136.352705410822 0.000221637111465415
136.633266533066 0.000210861772940367
136.913827655311 0.000200573085199011
137.194388777555 0.000190751171800812
137.4749498998 0.000181376852329544
137.755511022044 0.000172431624311568
138.036072144289 0.000163897645252491
138.316633266533 0.000155757714815124
138.597194388778 0.000147995257159678
138.877755511022 0.000140594303465586
139.158316633267 0.000133539474652837
139.438877755511 0.000126815964319123
139.719438877756 0.000120409521907658
140 0.000114306436119128
};
\end{axis}

\end{tikzpicture}

%% file: linear_solvers_stco.bbl
\begin{thebibliography}{18}
\providecommand{\natexlab}[1]{#1}
\providecommand{\url}[1]{\texttt{#1}}
\expandafter\ifx\csname urlstyle\endcsname\relax
  \providecommand{\doi}[1]{doi: #1}\else
  \providecommand{\doi}{doi: \begingroup \urlstyle{rm}\Url}\fi

\bibitem[Bartels and Hennig(2016)]{Bartels:2016eh}
S.~Bartels and P.~Hennig.
\newblock Probabilistic approximate least-squares.
\newblock In \emph{Proceedings of Artificial Intelligence and Statistics
  (AISTATS)}, 2016.

\bibitem[Cockayne et~al.(2016)Cockayne, Oates, Sullivan, and
  Girolami]{Cockayne:2016}
J.~Cockayne, C.~Oates, T.~Sullivan, and M.~Girolami.
\newblock Probabilistic numerical methods for partial differential equations
  and bayesian inverse problems, 2016.

\bibitem[Cockayne et~al.(2017)Cockayne, Oates, Sullivan, and
  Girolami]{Cockayne:2017}
J.~Cockayne, C.~Oates, T.~Sullivan, and M.~Girolami.
\newblock Bayesian probabilistic numerical methods, 2017.

\bibitem[Cockayne et~al.(2018)Cockayne, Oates, and Girolami]{Cockayne:2018}
J.~Cockayne, C.~Oates, and M.~Girolami.
\newblock A bayesian conjugate gradient method, 2018.

\bibitem[Diaconis and Shahshahani(1987)]{Diaconis1987}
P.~Diaconis and M.~Shahshahani.
\newblock The subgroup algorithm for generating uniform random variables.
\newblock \emph{Probability in the Engineering and Informational Sciences},
  1\penalty0 (01):\penalty0 15, jan 1987.
\newblock \doi{10.1017/s0269964800000255}.

\bibitem[Golub and Van~Loan(2013)]{GolubVanLoan2013}
G.~H. Golub and C.~F. Van~Loan.
\newblock \emph{Matrix computations}.
\newblock Johns Hopkins Studies in the Mathematical Sciences. Johns Hopkins
  University Press, Baltimore, MD, fourth edition, 2013.

\bibitem[Hennig(2015)]{Hennig2015}
P.~Hennig.
\newblock Probabilistic interpretation of linear solvers.
\newblock \emph{{SIAM} Journal on Optimization}, 25\penalty0 (1):\penalty0
  234--260, jan 2015.
\newblock \doi{10.1137/140955501}.
\newblock URL \url{https://doi.org/10.1137/140955501}.

\bibitem[Hennig et~al.(2015)Hennig, Osborne, and Girolami]{HenOsbGirRSPA2015}
P.~Hennig, M.~A. Osborne, and M.~Girolami.
\newblock Probabilistic numerics and uncertainty in computations.
\newblock \emph{Proceedings of the Royal Society of London A: Mathematical,
  Physical and Engineering Sciences}, 2015.

\bibitem[Karvonen and Sarkka(2017)]{Karvonen2017}
T.~Karvonen and S.~Sarkka.
\newblock Classical quadrature rules via gaussian processes.
\newblock In \emph{2017 {IEEE} 27th International Workshop on Machine Learning
  for Signal Processing ({MLSP})}. {IEEE}, sep 2017.
\newblock \doi{10.1109/mlsp.2017.8168195}.

\bibitem[{Kersting} et~al.(2018){Kersting}, {Sullivan}, and
  {Hennig}]{2018arXiv180709737K}
H.~{Kersting}, T.~J. {Sullivan}, and P.~{Hennig}.
\newblock {Convergence Rates of Gaussian ODE Filters}.
\newblock \emph{ArXiv e-prints}, 1807.09737, 7 2018.

\bibitem[Liesen and Strakos(2012)]{Liesen:2012tt}
J.~Liesen and Z.~Strakos.
\newblock \emph{{Krylov Subspace Methods}}.
\newblock Principles and Analysis. Oxford University Press, Oct. 2012.
\newblock \doi{10.1093/acprof:oso/9780199655410.001.0001}.

\bibitem[Nocedal and Wright(1999)]{nocedal1999numerical}
J.~Nocedal and S.~J. Wright.
\newblock \emph{{Numerical Optimization}}.
\newblock Springer Verlag, 1999.

\bibitem[Saad(2003)]{Saad2003iterative}
Y.~Saad.
\newblock \emph{Iterative Methods for Sparse Linear Systems}.
\newblock Society for Industrial and Applied Mathematics, second edition, 2003.

\bibitem[Saad and Schultz(1986)]{Saad1986}
Y.~Saad and M.~H. Schultz.
\newblock {GMRES}: A generalized minimal residual algorithm for solving
  nonsymmetric linear systems.
\newblock \emph{{SIAM} Journal on Scientific and Statistical Computing},
  7\penalty0 (3):\penalty0 856--869, jul 1986.
\newblock \doi{10.1137/0907058}.
\newblock URL \url{https://doi.org/10.1137/0907058}.

\bibitem[Schober et~al.(2014)Schober, Duvenaud, and Hennig]{SchoberDH2014}
M.~Schober, D.~Duvenaud, and P.~Hennig.
\newblock Probabilistic {ODE} solvers with runge-kutta means.
\newblock In \emph{Advances in Neural Information Processing Systems 27}, pages
  739--747. Curran Associates, Inc., 2014.
\newblock URL
  \url{http://papers.nips.cc/paper/5451-probabilistic-ode-solvers-with-runge-kutta-means.pdf}.

\bibitem[Schober et~al.(2018)Schober, S{\"a}rkk{\"a}, and
  Hennig]{SchoberSarkkaHennig2018}
M.~Schober, S.~S{\"a}rkk{\"a}, and P.~Hennig.
\newblock A probabilistic model for the numerical solution of initial value
  problems.
\newblock \emph{Statistics and Computing}, 2018.

\bibitem[Soodhalter et~al.(2014)Soodhalter, Szyld, and Xue]{Soodhalter2014}
K.~M. Soodhalter, D.~B. Szyld, and F.~Xue.
\newblock Krylov subspace recycling for sequences of shifted linear systems.
\newblock \emph{Applied Numerical Mathematics}, 81:\penalty0 105--118, jul
  2014.
\newblock \doi{10.1016/j.apnum.2014.02.006}.
\newblock URL \url{https://doi.org/10.1016/j.apnum.2014.02.006}.

\bibitem[Xi et~al.(2018)Xi, Briol, and Girolami]{Xi:2018}
X.~Xi, F.-X. Briol, and M.~Girolami.
\newblock Bayesian quadrature for multiple related integrals.
\newblock In \emph{Proceedings of the 35th International Conference on Machine
  Learning (ICML)}, 2018.
\newblock \arXiv{801.04153}.

\end{thebibliography}
